\documentclass{vldb}

\usepackage{amsmath}									
\usepackage{amssymb}									
\usepackage{algorithm}
\usepackage{algorithmic}
\usepackage{graphicx}
\usepackage[caption=false]{subfig}
\usepackage{enumerate}
\usepackage[usenames,dvipsnames]{color}
\usepackage{balance}

\newcommand{\hide}[1]{}

\def\mathbi#1{\textbf{\textit #1}}

\newtheorem{definition}{Definition}

\newtheorem{lemma}{Lemma}
\newtheorem{theorem}{Theorem}

\newtheorem{observation}{Observation}

\begin{document}

\title{DisC Diversity:~Result~Diversification~based~on Dissimilarity and Coverage \titlenote{To appear at the 39th International Conference on Very Large Data Bases (VLDB), August 26-31, 2013, Riva del Garda, Trento, Italy.}}

\numberofauthors{2}

\author{
	\alignauthor
		Marina Drosou\\
    \affaddr{Computer Science Department}\\
    \affaddr{University of Ioannina, Greece}\\
    \email{mdrosou@cs.uoi.gr}
	\alignauthor
		Evaggelia Pitoura\\
    \affaddr{Computer Science Department}\\
    \affaddr{University of Ioannina, Greece}\\
    \email{pitoura@cs.uoi.gr}	
}

\maketitle

\begin{abstract}
Recently, result diversification has attracted a lot of attention as a means to improve the quality of results retrieved by user queries. In this paper, we propose a new, intuitive definition of diversity called DisC diversity. 
A DisC diverse subset of a query result contains objects such that each object in the result is represented by a similar object in the diverse subset and the objects in the diverse subset are dissimilar to each other. We show that locating a minimum DisC diverse subset is an NP-hard problem and provide heuristics for its approximation. We also propose adapting DisC diverse subsets to a different degree of diversification. We call this operation zooming. We present efficient implementations of our algorithms based on the M-tree, a spatial index structure, and experimentally evaluate their performance.
\end{abstract}

\section{Introduction}
\label{Introduction}

Result diversification has attracted considerable attention as a means of enhancing the quality of the query results presented to users (e.g.,\cite{DBLP:conf/icde/VeeSSBA08, DBLP:conf/www/ZieglerMKL05}). Consider, for example, a user who wants to buy a camera and submits a related query. A diverse result, i.e., a result containing various brands and models with different pixel counts and other technical characteristics is intuitively more informative than a homogeneous result containing only cameras with similar features.

There have been various definitions of diversity \cite{DBLP:journals/sigmod/DrosouP10}, based on (i)~\textit{content} (or \textit{similarity}), i.e., objects that are dissimilar to each other (e.g., \cite{DBLP:conf/www/ZieglerMKL05}), (ii)~\textit{novelty}, i.e., objects that contain new information when compared to what was previously presented (e.g., \cite{DBLP:conf/sigir/ClarkeKCVABM08}) and (iii)~\textit{semantic coverage}, i.e., objects that belong to different categories or topics (e.g., \cite{DBLP:conf/wsdm/AgrawalGHI09}). Most previous approaches rely on assigning a diversity score to each object and then selecting either the $k$ objects with the highest score for a given $k$ (e.g., \cite{koudas,sigmod12}) or the objects with score larger than some threshold (e.g., \cite{DBLP:conf/edbt/YuLA09}). 

In this paper, we address diversity through a different perspective. Let ${\cal P}$ be the set of objects in a query result. We consider  two objects $p_1$ and $p_2$ in  $\cal{P}$ to be similar, if $dist(p_1, p_2)$ $\leq$ $r$ for some distance function $dist$ and real number $r$, where $r$ is a tuning parameter that we call \emph{radius}. 
Given  ${\cal P}$, we select a representative subset $S$ $\subseteq$ ${\cal P}$ to be presented to the user such that: (i)~all objects in ${\cal P}$ are similar with at least one object in $S$  and (ii)~no two objects in $S$ are similar with each other. The first condition ensures that all objects in ${\cal P}$ are represented, or {\em covered}, by at least one object in the selected subset. The second condition ensures that the selected objects in $\cal{P}$ are dissimilar, or {\em independent}. We call the set $S$ {\em $r$-Dissimilar and Covering} subset or {\em $r$-DisC} diverse subset. 

In contrary to previous approaches to diversification, we aim at computing subsets of objects that contain objects that are {\em both} dissimilar with each other {\em and} cover the whole result set. Furthermore, instead of specifying a required size $k$ of the diverse set or a threshold, our tuning parameter $r$ explicitly expresses the degree of diversification and determines the size of the diverse set. Increasing $r$ results in smaller, more diverse subsets, while decreasing $r$ results in larger, less diverse subsets. We call these operations, {\em zooming-out} and {\em zooming-in} respectively. One can also zoom-in or zoom-out locally to a specific object in the presented result.  

As an example, consider searching for cities in Greece. Figure~\ref{fig:introgreece} shows the results of this query diversified based on geographical location for an initial radius (a), after zooming-in (b), zooming-out (c) and local zooming-in a specific city (d). As another example of local zooming in the case of categorical attributes, consider looking for cameras, where diversity refers to cameras with different features. Figure~\ref{fig:introcameras} depicts an initial most diverse result and the result of local zooming-in one individual camera in this result.

\begin{figure}[t]
	\vspace{-0.2cm}
	\centering
	\subfloat[Initial set.]{		
		\includegraphics[width=3.8cm]{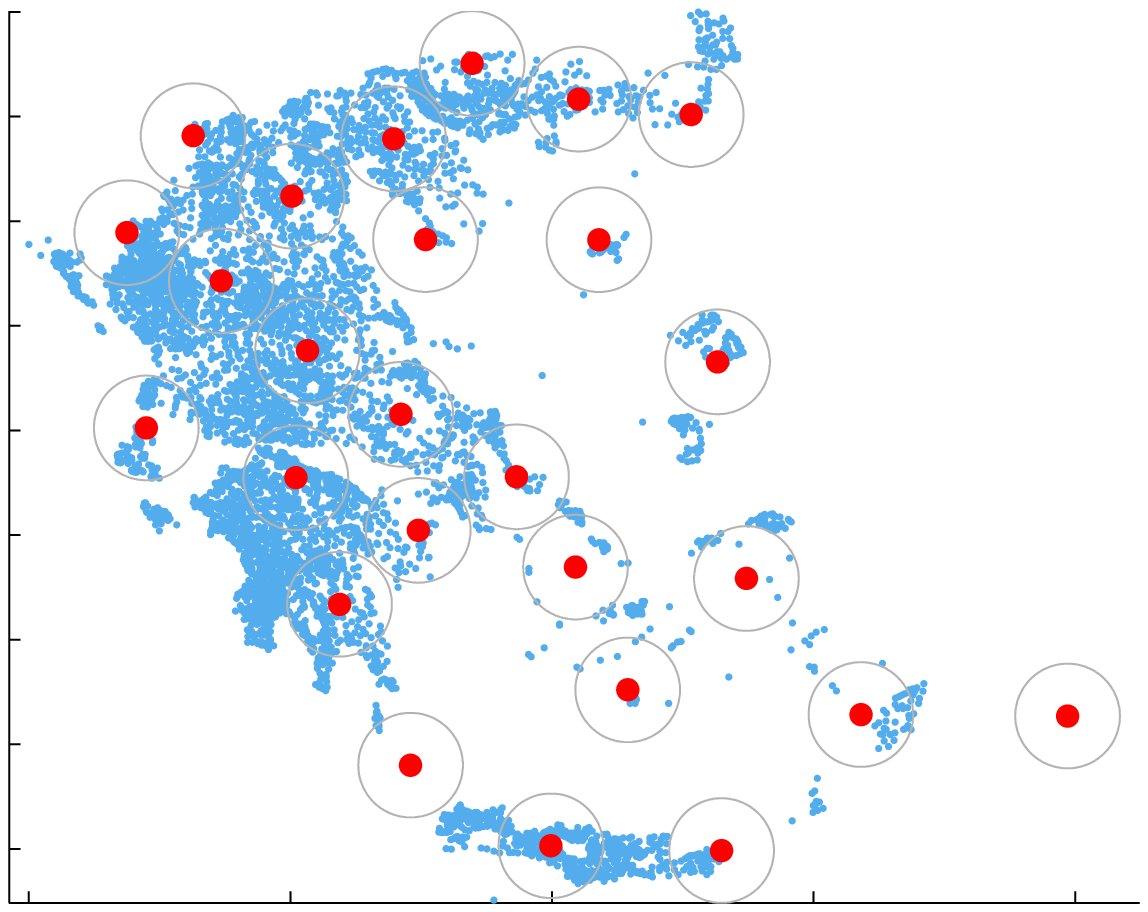}
	\label{fig:introgreece-initial}
	}
	\subfloat[Zooming-in.]{		
		\includegraphics[width=3.8cm]{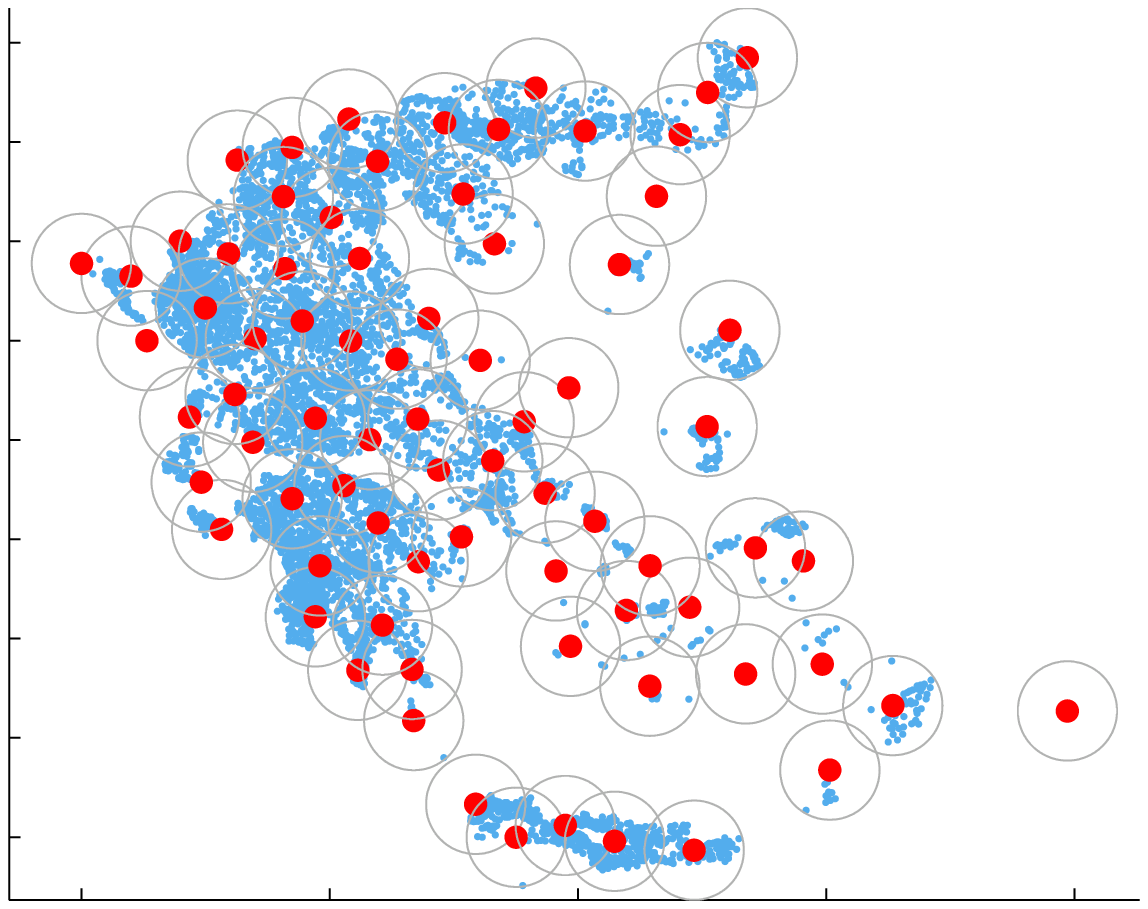}
	\label{fig:introgreece-zoomin}
	}\\
	\subfloat[Zooming-out.]{		
		\includegraphics[width=3.8cm]{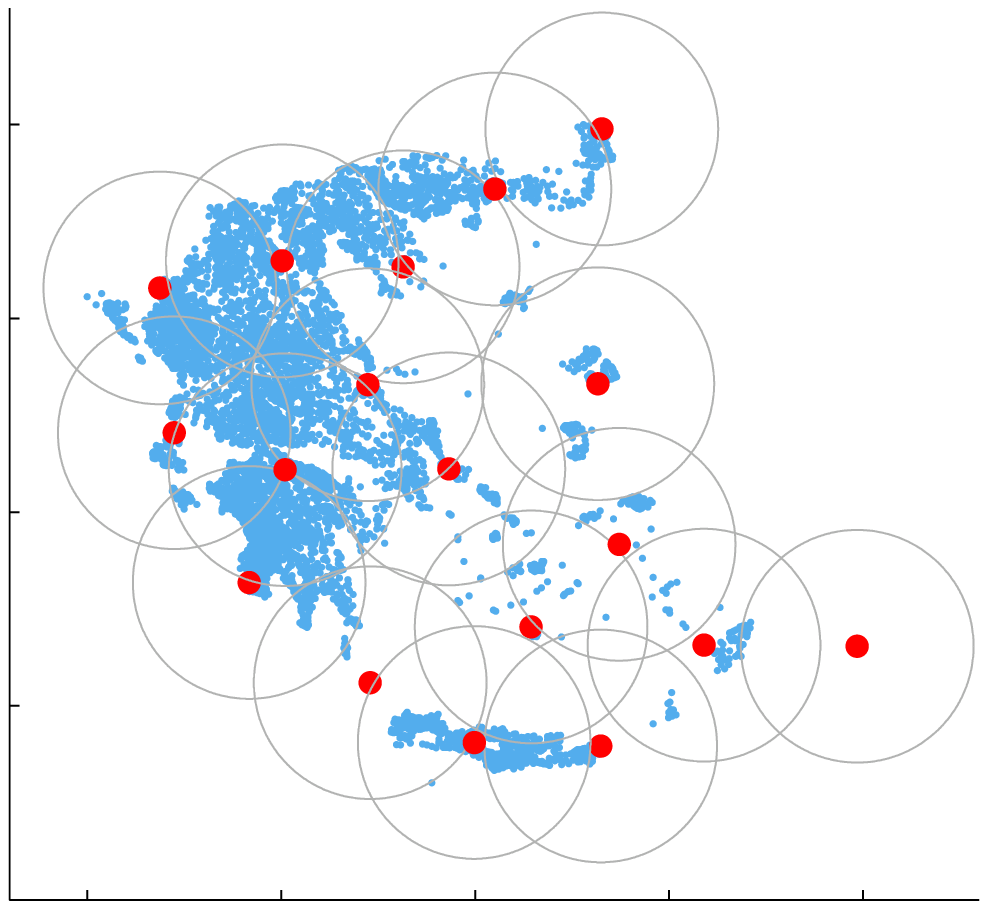}
	\label{fig:introgreece-zoomout}
	}
	\subfloat[Local zooming-in.]{		
		\includegraphics[width=3.8cm]{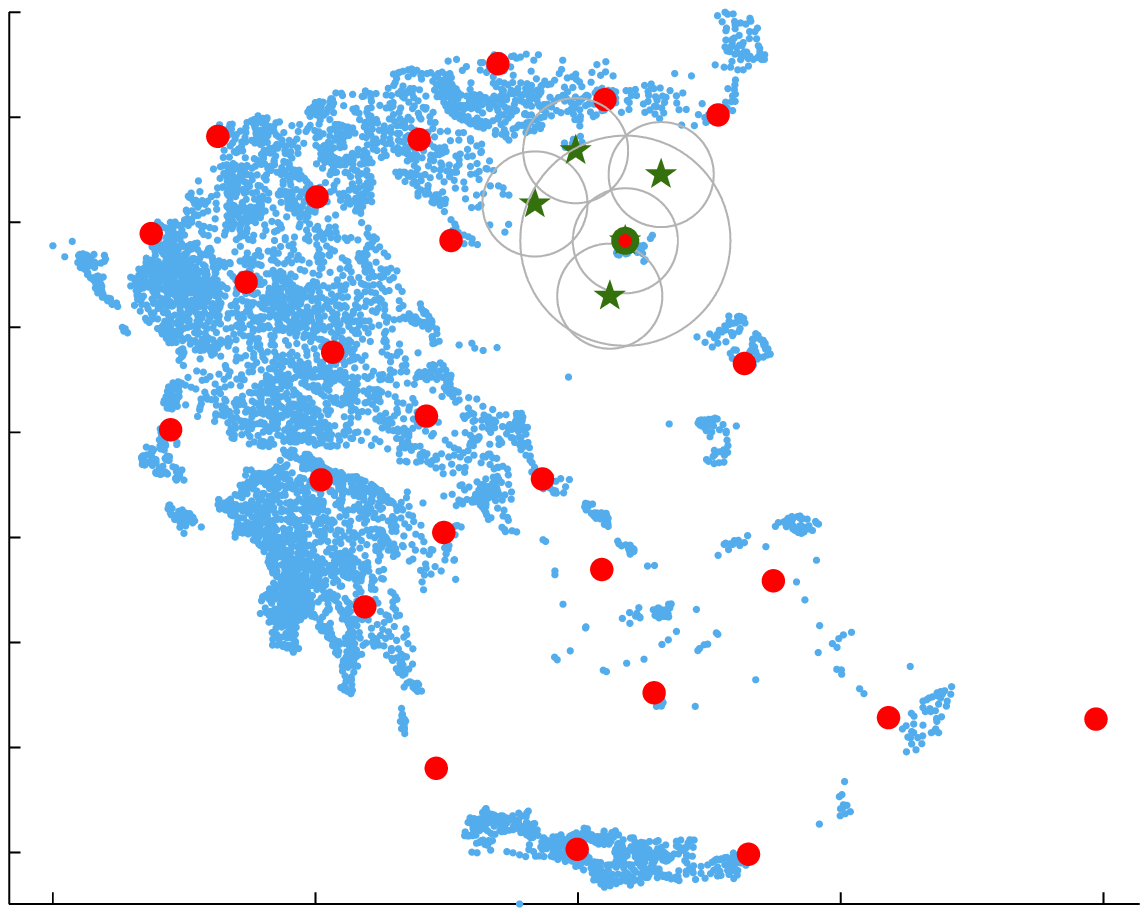}
	\label{fig:introgreece-localzoomin}
	}
	\vspace{-0.2cm}
	\caption{Zooming operations in action. Selected objects are shown in bold. Solid circles denote the radius $r$ of the selected objects.}
	\label{fig:introgreece}
	\vspace{-0.4cm}
\end{figure}

We formalize the problem of locating minimum DisC diverse subsets as an independent dominating set problem on graphs \cite{DBLP:journals/ipl/Halldorsson93a}. We provide a suite of heuristics for computing small DisC diverse subsets. We also consider the problem of adjusting the radius $r$. We explore the relation among DisC diverse subsets of different radii and provide algorithms for incrementally adapting a DisC diverse subset to a new radius. We provide theoretical upper bounds for the size of the diverse subsets produced by our algorithms for computing DisC diverse subsets as well as for their zooming counterparts. Since the crux of the efficiency of the proposed algorithms is locating neighbors, we take advantage of spatial data structures. In particular, we propose efficient algorithms based on the M-tree \cite{similaritysearch2006}.

We compare the quality of our approach to other diversification methods both analytically and qualitatively.  We also evaluate our various heuristics using both real and synthetic datasets. Our performance results show that the basic heuristic for computing dissimilar and covering subsets works faster than its greedy variation but produces larger sets. Relaxing the dissimilarity condition, although in theory could result in smaller sets, in our experiments does not reduce the size of the result considerably. Our incremental algorithms for zooming in or out to a different radius $r^{\prime}$, when compared to computing a DisC diverse subset for $r^{\prime}$ from scratch, produce sets of similar sizes and closer to what the user intuitively expects, while imposing a smaller computational cost.
Finally, we draw various conclusions for the M-tree implementation of these algorithms. 

Most often diversification is modeled as a bi-criteria problem with the dual goal of maximizing both the diversity and the relevance of the selected results. In this paper, we focus solely on diversity. Since we ``cover'' the whole dataset, each user may ``zoom-in'' to the area of the results that seems most relevant to her individual needs. Of course, many other approaches to integrating relevance with DisC diversity are possible; we discuss some of them in Section~\ref{Summary}.

In a nutshell, in this paper, we make the following contributions:
\begin{enumerate} \vspace{-0.2cm}
	\item [--] we propose a new, intuitive definition of diversity, called DisC diversity and compare it with other models,\vspace{-0.2cm}
	\item [--] we show that locating minimum DisC diverse subsets is an NP-hard problem and provide efficient heuristics
along with approximation bounds, \vspace{-0.2cm}
	\item [--] we introduce adaptive diversification through zooming-in and zooming-out and present algorithms for their incremental computation as well as corresponding theoretical bounds,\vspace{-0.2cm}
	\item [--] we provide M-tree tailored algorithms and experimentally evaluate their performance.\vspace{-0.2cm}
\end{enumerate}

The rest of the paper is structured as follows. Section~\ref{sec:DisC Diversity} introduces DisC diversity and  heuristics for computing small diverse subsets, Section~\ref{sec:Changing Focus} introduces adaptive diversification and Section~\ref{sec:Comparison} compares our approach with other diversification methods. In Section~\ref{sec:Indexing}, we employ the M-tree for the efficient implementation of our algorithms, while in Section~\ref{sec:Evaluation}, we present experimental results. Finally, Section~\ref{Related Work} presents related work and Section~\ref{Summary} concludes the paper.

\begin{figure}[t]
	\vspace{-0.1cm}
	\centering
		\includegraphics[width=8cm]{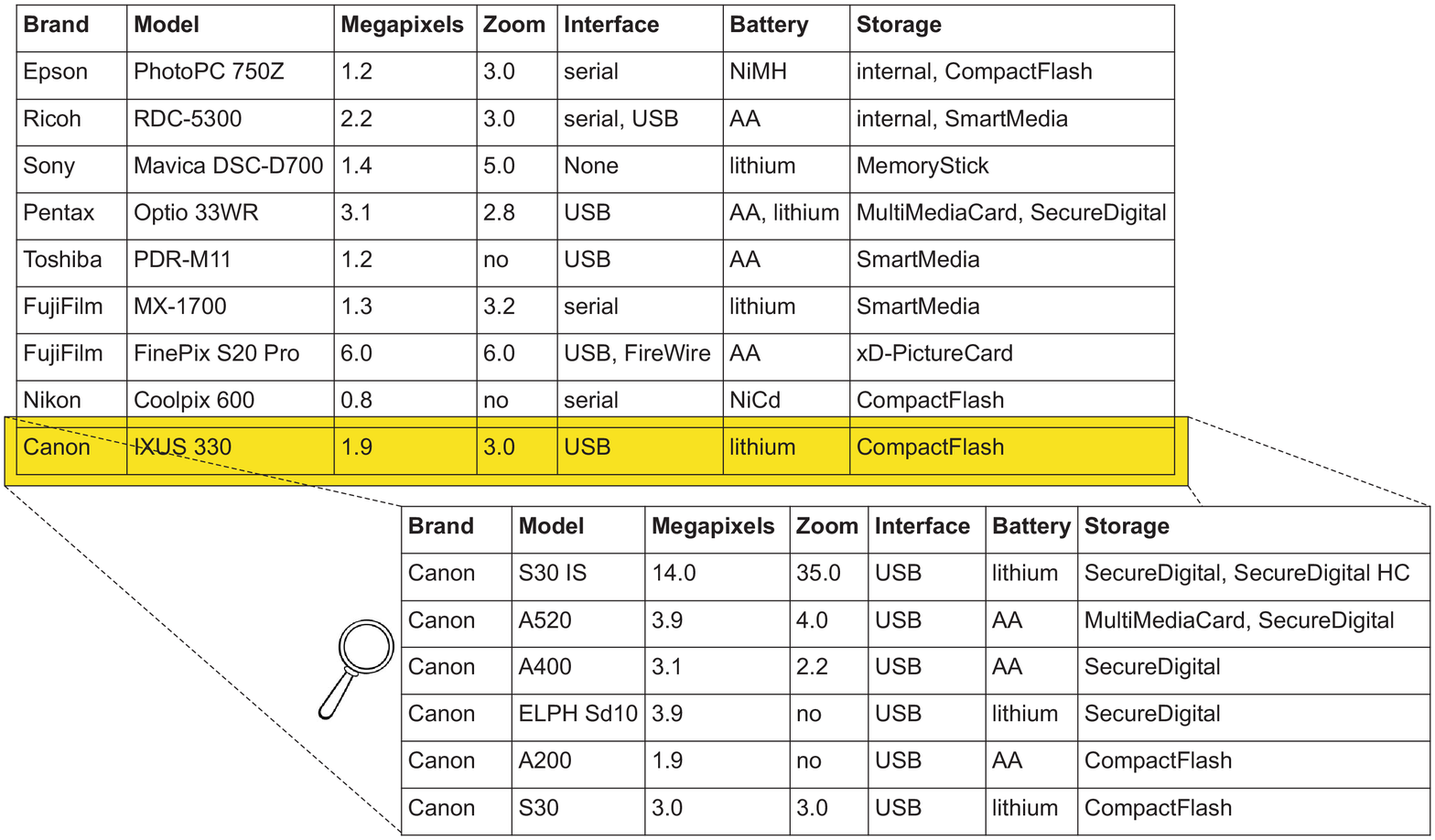}
	\vspace{-0.1cm}
	\caption{Zooming in a specific camera.}
	\label{fig:introcameras}
	\vspace{-0.2cm}
\end{figure}

\section{DisC Diversity}
\label{sec:DisC Diversity}
In this section, we first provide a formal definition of DisC diversity. We then show that locating a minimum DisC diverse set of objects is an NP-hard problem and present heuristics for locating approximate solutions.

\subsection{Definition of DisC Diversity}
Let $\mathcal{P}$ be a set of objects returned as the result of a user query. We want to select a representative subset $S$ of these objects such that each object from $\mathcal{P}$ is represented by a similar object in $S$ and the objects selected to be included in $S$ are dissimilar to each other.

We define similarity between two objects using a distance metric $dist$. For a real number $r$, $r$ $\geq$ 0, we use $N_r(p_i)$ to denote the set of \emph{neighbors} (or \emph{neighborhood}) of an object $p_i \in \mathcal{P}$, i.e., the objects lying at distance at most $r$ from $p_i$:
\vspace{-0.4cm}
\[
N_r(p_i) = \{p_j~|~p_i \neq p_j \wedge dist(p_i, p_j) \leq r\}
\vspace{-0.1cm}
\]
We use $N_r^{+}(p_i)$ to denote the set $N_r(p_i) \cup \{p_i\}$, i.e., the neighborhood of $p_i$ including $p_i$ itself. Objects in the neighborhood of $p_i$ are considered similar to $p_i$, while objects outside its neighborhood are considered dissimilar to $p_i$. We define an $r$-DisC diverse subset as follows:

\vspace{-0.1cm}
\begin{definition}
\label{def1}
{\textnormal{(}$r$-{\sc DisC Diverse Subset}\textnormal{)}}
\mbox{~} Let $\mathcal{P}$ be a set of objects and $r$, $r$ $\geq$ 0, a real number. A subset $S \subseteq \mathcal{P}$ is an r-Dissimilar-and-Covering diverse subset, or r-DisC diverse subset, of $\mathcal{P}$, if the following two conditions hold: (i)~(coverage condition) $\forall p_i \in \mathcal{P}$, $\exists p_j \in N_r^{+}(p_i)$, such that $p_j \in S$ and (ii)~(dissimilarity condition)  $\forall$  $p_i$, $p_j \in S$ with $i$ $\neq$ $j$,  it holds that $dist(p_i, p_j) > r$.
\end{definition}
\vspace{-0.1cm}

The first condition ensures that all objects in $\mathcal{P}$ are represented by at least one similar object in  $S$ and the second condition that the objects in  $S$ are dissimilar to each other. We call every object $p_i \in S$ an \emph{$r$-DisC diverse object} and $r$ the \emph{radius} of $S$. When the value of $r$ is clear from context, we simply refer to $r$-DisC diverse objects as diverse objects. Given $\mathcal{P}$, we would like to select the smallest number of diverse objects.   

\vspace{-0.1cm}
\begin{definition}
{\textnormal{(}\sc The Minimum} $r$-{\sc DisC Diverse Subset \\
Problem\textnormal{)}} Given a set $\mathcal{P}$ of objects and a radius $r$, find an $r$-DisC diverse subset  $S^*$ of $\mathcal{P}$, such that, for every $r$-DisC diverse subset $S$ of $\mathcal{P}$, it holds that $|S^*|$ $\leq$ $|S|$. 
\end{definition}
\vspace{-0.1cm}

In general, there may be more than one minimum $r$-DisC diverse subsets of $\mathcal{P}$ (see Figure~\ref{fig:no-unique}\subref{fig:geographical} for an example).

\begin{figure}
	\vspace{-0.3cm}
	\centering
	\subfloat[]{
		\includegraphics[width=2.8cm]{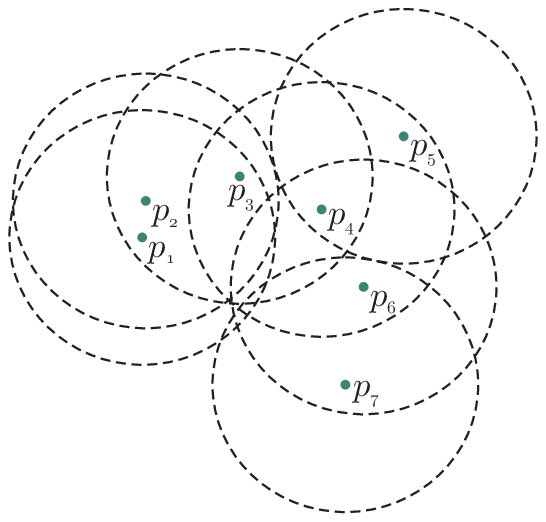}
		\label{fig:geographical}
	}
	\subfloat[]{
		\includegraphics[width=3.5cm]{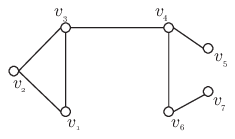}
		\label{fig:graph}
	}
	\vspace{-0.1cm}
	\caption{(a) Minimum $r$-DisC diverse subsets for the depicted objects: $\{p_1, p_4, p_7\}$, $\{p_2, p_4, p_7\}$, $\{p_3, p_5, p_6\}$, $\{p_3, p_5, p_7\}$ and (b) their graph representation.}
	\label{fig:no-unique}
	\vspace{-0.4cm}
\end{figure}

\subsection{Graph Representation and NP-hardness}
Let us consider the following graph representation of a set $\mathcal P$ of objects. Let $G_{\mathcal P, r}$ = ($V$, $E$) be an undirected graph such that there is a vertex $v_i \in V$ for each object $p_i \in \mathcal{P}$ and an edge $(v_i, v_j) \in E$, if and only if, $dist(p_i, p_j)$ $\leq$ $r$ for the corresponding objects $p_i$, $p_j$. An example is shown in Figure \ref{fig:no-unique}\subref{fig:graph}. 

Let us recall a couple of graph-related definitions. A {\em dominating set} $D$ for a graph $G$ is a subset of vertices of $G$ such that every vertex of $G$ not in $D$ is joined to at least one vertex of $D$ by some edge. An {\em independent set} $I$ for a graph $G$ is a set of vertices of $G$ such that for every two vertices in $I$, there is no edge connecting them. It is easy to see that a dominating set of $G_{\mathcal P, r}$ satisfies the covering conditions of Definition \ref{def1}, whereas an independent set of  $G_{\mathcal P, r}$ satisfies the dissimilarity conditions of Definition \ref{def1}. Thus:

\vspace{-0.1cm}
\begin{observation}
Solving the \textsc{Minimum $r$-DisC Diverse Subset Problem} for a set $\mathcal P$ is equivalent to finding a \textsc{Minimum Independent Dominating Set} of the corresponding graph $G_{\mathcal P, r}$.
\end{observation}
\vspace{-0.1cm}

The \textsc{Minimum Independent Dominating Set Problem} has been proven to be NP-hard \cite{DBLP:books/fm/GareyJ79}. 
The problem remains NP-hard even for special kinds of graphs, such as for unit disk graphs \cite{DBLP:journals/dm/ClarkCJ90}.
{\em  Unit disk graphs} are graphs whose vertices can be put in one to one correspondence with equisized circles in a plane such that two vertices are joined by an edge, if and only if, the corresponding circles intersect. 
$G_{\mathcal P, r}$ is a unit disk graph for Euclidean distances.

In the following, we use the terms dominance and coverage, as well as, independence and dissimilarity interchangeably. In particular, two objects  $p_i$ and $p_j$ are independent, if $dist(p_i, p_j) > r$. We also say that an object {\em covers} all objects in its neighborhood. We next present some useful properties that relate the coverage (i.e., dominance) and dissimilarity (i.e., independence) conditions. A {\em maximal independent set} is an independent set such that adding any other vertex to the set forces the set to contain an edge, that is, it is an independent set that is not a subset of any other independent set. It is known that:

\vspace{-0.1cm}
\begin{lemma}
\label{lem:maxind}
An independent set of a graph is maximal, if and only if, it is dominating.
\end{lemma}

\vspace{-0.2cm}
From Lemma~\ref{lem:maxind}, we conclude that:

\vspace{-0.2cm}
\begin{observation}
A minimum maximal independent set is also a minimum independent dominating set. 
\end{observation}

\vspace{-0.2cm}
However,
\begin{observation}
\label{co}
A minimum dominating set is not necessarily independent.
\end{observation}
\vspace{-0.1cm}

For example, in Figure~\ref{fig:mds-mids}, the minimum dominating set of the depicted objects is of size 2, while the minimum independent dominating set is of size 3.

\begin{figure}
	\vspace{-0.3cm}
	\centering
	\subfloat[]{
		\includegraphics[width=2.1cm]{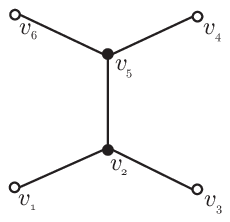}
		\label{fig:mds}
	}\hspace{1.0cm}
	\subfloat[]{
		\includegraphics[width=2.1cm]{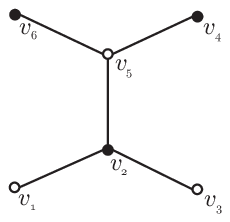}
		\label{fig:mids}
	}
	\vspace{-0.1cm}
	\caption{(a) Minimum dominating set ($\{v_2, v_5\}$) and (b) a minimum independent dominating set ($\{v_2, v_4, v_6\}$) for the depicted graph.}
	\label{fig:mds-mids}
	\vspace{-0.4cm}
\end{figure}

\subsection{Computing DisC Diverse Objects}
We consider first a baseline algorithm for computing a DisC diverse subset $S$ of $\mathcal{P}$. For presentation convenience, let us call \emph{black} the objects of $\mathcal{P}$ that are in $S$, \emph{grey} the objects covered by $S$ and \emph{white} the objects that are neither black nor grey. Initially, $S$ is empty and all objects are white. The algorithm proceeds as follows: until there are no more white objects, it selects an arbitrary white object $p_i$, colors $p_i$ black and colors all objects in $N_r(p_i)$ grey. We call this algorithm \texttt{Basic-DisC}. 

The produced set $S$ is clearly an independent set, since once an object enters $S$, all its neighbors become grey and thus are withdrawn from consideration. It is also a maximal independent set, since at the end there are only grey objects left, thus adding any of them to $S$ would violate the independence of $S$. From Lemma~\ref{lem:maxind}, the set $S$ produced by \texttt{Basic-DisC} is an $r$-DisC diverse subset. $S$ is not necessarily a minimum $r$-DisC diverse subset. However, its size is related to the size of any minimum $r$-DisC diverse subset $S^*$ as follows:

\vspace{-0.1cm}
\begin{theorem}
\label{th1}
Let $B$ be the maximum number of independent neighbors of any object in $\mathcal{P}$. Any $r$-DisC diverse subset  $S$ of $\mathcal{P}$ is at most $B$ times larger than any minimum $r$-DisC diverse subset $S^*$.
\end{theorem}
\begin{proof}
Since $S$ is an independent set, any object in $S^*$ can cover at most $B$ objects in $S$ and thus $|S|$ $\leq$ $B|S^*|$.\end{proof}
\vspace{-0.1cm}

The value of $B$ depends on the distance metric used and also on the dimensionality $d$ of the data space. For many distance metrics $B$ is a constant. Next, we show how $B$ is bounded for specific combinations of the distance metric and data dimensionality $d$.

\vspace{-0.1cm}
\begin{lemma}
\label{disk5}
If $dist$ is the Euclidean distance and $d = 2$, each object $p_i$ in $\mathcal{P}$ has at most 5 neighbors that are independent from each other.
\end{lemma}
\begin{proof}
Let $p_1$, $p_2$ be two independent neighbors of $p_i$. Then, it must hold that $\angle p_1p_ip_2$ is larger than $\frac{\pi}{3}$. Otherwise, $dist(p_1, p_2)$ $\leq$ $\max\{dist(p_i, p_1), dist(p_i, p_2)\}$ $\leq r$ which contradicts the independence of $p_1$ and $p_2$. Therefore, $p_i$ can have at most $(2\pi / \frac{\pi}{3}) - 1$ $= 5$ independent neighbors.
\end{proof}
\vspace{-0.1cm}

\vspace{-0.1cm}
\begin{lemma}
\label{lem:manhattan2}
If dist is the Manhattan distance and $d = 2$, each object $p_i$ in $\mathcal{P}$ has at most 7 neighbors that are independent from each other.
\end{lemma}
\begin{proof}
The proof can be found in the Appendix.
\end{proof}
\vspace{-0.1cm}

For $d = 3$ and the Euclidean distance, it can be shown that each object $p_i$ in $\mathcal{P}$ has at most 24 neighbors that are independent from each other. This can be shown using packing techniques and properties of solid angles \cite{tr}.

We now consider the following intuitive greedy variation of \texttt{Basic-DisC}, that we call \texttt{Greedy-DisC}. Instead of selecting white objects arbitrarily at each step, we select the white object with the largest number of white neighbors, that is, the white object that covers the largest number of uncovered objects. \texttt{Greedy-DisC} is shown in Algorithm~\ref{alg:mids}, where $N_r^W(p_i)$ is the set of the white neighbors of object $p_i$. 

\begin{algorithm}[t]
	\caption{Greedy-DisC}
	\label{alg:mids}	
	\small
	\begin{algorithmic}[1]
		\REQUIRE{A set of objects $\mathcal{P}$ and a radius $r$.}
		\ENSURE{An $r$-DisC diverse subset $S$ of $\mathcal{P}$.}
		\vspace{0.1cm}
		\hrule
		\vspace{0.1cm}
		
		\STATE $S \leftarrow \emptyset$
		
		\FORALL {$p_i \in \mathcal{P}$}
			\STATE Color $p_i$ white			
		\ENDFOR
		
		\WHILE{there exist white objects}	
			\STATE Select the white object $p_i$ with the largest $\left|N_r^W(p_i)\right|$
			\STATE $S = S \cup \{p_i\}$
			\STATE Color $p_i$ black
			\FORALL {$p_j \in N_r^W(p_i)$}
				\STATE Color $p_j$ grey			
			\ENDFOR
		\ENDWHILE
				
 		\RETURN $S$
	\end{algorithmic}
\end{algorithm}

While the size of the $r$-DisC diverse subset $S$ produced by \texttt{Greedy-DisC} is expected to be smaller than that of the subset produced by \texttt{Basic-DisC}, the fact that we consider for inclusion in $S$ only white, i.e., independent, objects may still not reduce the size of $S$ as much as expected. From Observation \ref{co}, it is possible that an independent covering set is larger than a covering set that also includes dependent objects. For example, consider the nodes (or equivalently the corresponding objects) in Figure \ref{fig:mds-mids}. Assume that object $v_2$ is inserted in $S$ first, resulting in objects $v_1$, $v_3$ and $v_5$ becoming grey. Then, we need two more objects, namely, $v_4$ and $v_6$, for covering all objects. However, if we consider for inclusion grey objects as well, then $v_5$ can join $S$, resulting in a smaller covering set.

Motivated by this observation, we also define $r$-C diverse subsets that satisfy only the coverage condition of Definition \ref{def1} and modify \texttt{Greedy-DisC} accordingly to compute $r$-C diverse sets. The only change required is that in line 6 of Algorithm \ref{alg:mids}, we select both white and grey objects. This allows us to select at each step the object that covers the largest possible number of uncovered objects, even if this object is grey. We call this variation \texttt{Greedy-C}. In the case of \texttt{Greedy-C}, we prove the following bound for the size of the produced $r$-C diverse subset $S$:

\vspace{-0.1cm}
\begin{theorem}
\label{thm:amortized}
Let $\Delta$ be the maximum number of neighbors of any object in $\mathcal{P}$. The $r$-C diverse subset produced by \textnormal{\texttt{Greedy-C}} is at most $\ln\Delta$ times larger than the minimum $r$-DisC diverse subset $S^*$.
\end{theorem}
\begin{proof}
The proof can be found in the Appendix.
\end{proof}
\vspace{-0.1cm}

\section{Adaptive Diversification}
\label{sec:Changing Focus}
The radius $r$ determines the desired degree of diversification. A large radius corresponds to fewer and less similar to each other representative objects, whereas a small radius results in more and less dissimilar representative objects. On one extreme, a radius equal to the largest distance between any two objects results in a single object being selected and on the other extreme, a zero radius results in all objects of $\mathcal{P}$ being selected.
We consider an interactive mode of operation where, after being presented with an initial set of results for some $r$, a user can see either more or less results by correspondingly decreasing or increasing $r$. 

Specifically, given a set of objects $\mathcal{P}$ and an $r$-DisC diverse subset $S^r$ of $\mathcal{P}$, we want to compute an $r^{\prime}$-DisC diverse subset $S^{r^{\prime}}$ of $\mathcal{P}$. There are two cases: (i)~$r^{\prime} < r$ and (ii)~$r^{\prime} > r$ which we call \emph{zooming-in} and \emph{zooming-out} respectively. These operations are global in the sense that the radius $r$ is modified similarly for all objects in $\mathcal{P}$. We may also modify the radius for a specific area of the data set. Consider, for example, a user that receives an $r$-DisC diverse subset $S^r$ of the results and finds some specific object $p_i$ $\in$ $S^r$ more or less interesting. Then, the user can zoom-in or zoom-out by specifying a radius $r^{\prime}$, $r^{\prime}$ $<$ $r$ or $r^{\prime}$ $>$ $r$, respectively, centered in $p_i$. We call these operations {\em local zooming-in} and {\em local zooming-out} respectively.

To study the size relationship between $S^r$ and $S^{r^{\prime}}$, we define the set $N_{r_1, r_2}^I(p_i)$, $r_2 \geq r_1$, as the set of objects at distance at most $r_2$ from $p_i$ which are at distance at least $r_1$ from each other, i.e., objects in $N_{r_2}(p_i)$ that are independent from each other considering the radius $r_1$. The following lemma bounds the size of $N_{r_1, r_2}^I(p_i)$ for specific distance metrics and dimensionality.

\begin{figure}
	\vspace{-0.3cm}
	\centering
	\subfloat[]{
		\includegraphics[width=3.5cm]{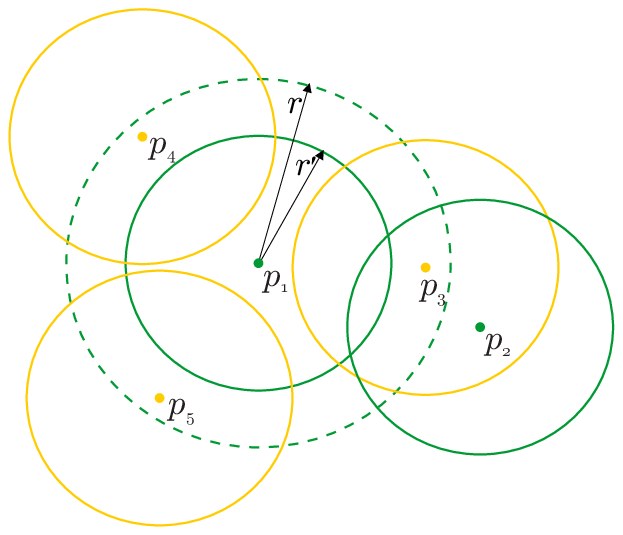}
		\label{fig:smaller-rnew}
	}
	\subfloat[]{
		\includegraphics[width=3.5cm]{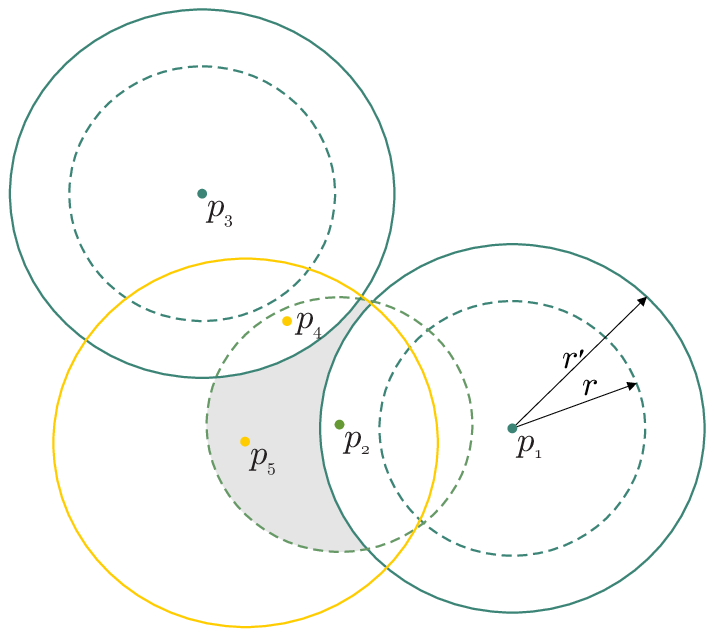}
		\label{fig:larger-rnew}
	}
	\vspace{-0.2cm}
	\caption{Zooming (a) in and (b) out. Solid and dashed circles correspond to radius $r^{\prime}$ and $r$ respectively.}
	\label{fig:zooming}
	\vspace{-0.4cm}
\end{figure}

\vspace{-0.1cm}
\begin{lemma}
\label{lem:independent-neighbors}
Let $r_1$, $r_2$ be two radii with $r_2 \geq r_1$. Then, for $d = 2$:
\begin{enumerate}
\item [(i)] if $dist$ is the Euclidean distance:
\vspace{-0.1cm}
\[
\left|N_{r_1, r_2}^I(p_i)\right| \leq 9 \left\lceil \log_\beta (r_2/r_1) \right\rceil \textnormal{, where } \beta = \frac{1 + \sqrt5}{2}
\vspace{-0.1cm}
\]
\item [(ii)] if $dist$ is the Manhattan distance:
\vspace{-0.1cm}
\[
\left|N_{r_1, r_2}^I(p_i)\right| \leq 4 \sum_{i=1}^{\gamma} (2i+1) \textnormal{, where } \gamma = \left\lceil \frac{r_2-r_1}{r_1}\right\rceil
\vspace{-0.1cm}
\]
\end{enumerate}
\end{lemma}
\begin{proof}
The proof can be found in the Appendix.
\end{proof}
\vspace{-0.1cm}

Since we want to support an incremental mode of operation, the set $S^{r^{\prime}}$ should be as close as possible to the already seen result $S^r$. Ideally, $S^{r^{\prime}}$ $\supseteq$ $S^r$, for $r^{\prime} < r$ and  $S^{r^{\prime}}$ $\subseteq$ $S^r$, for $r^{\prime} > r$. We would also like the size of $S^{r^{\prime}}$ to be as close as possible to the size of the minimum $r'$-DisC diverse subset.

If we consider only the coverage condition, that is, only $r$-C diverse subsets, then an $r$-C diverse subset of $\mathcal{P}$ is also an $r^{\prime}$-C diverse subset of $\mathcal{P}$, for any $r^{\prime}$ $\geq$  $r$. This holds because $N_r(p_i)$ $\subseteq$ $N_{r^{\prime}}(p_i)$ for any $r^{\prime}$ $\geq$ $r$. However, a similar property does not hold for the dissimilarity condition. In particular, a maximal independent diverse subset $S^r$ of $\mathcal{P}$ for $r$ is not necessarily a maximal independent  diverse subset of $\mathcal{P}$, neither for  $r^{\prime}$ $>$ $r$ nor for  $r^{\prime}$ $<$  $r$. To see this, note that for $r^{\prime}$ $>$  $r$, $S^r$ may not be independent, whereas for $r^{\prime}$ $<$  $r$,
$S^r$ may no longer be maximal. Thus, from Lemma~\ref{lem:maxind}, we reach the following conclusion:

\vspace{-0.1cm}
\begin{observation}
In general, there is no monotonic property among the $r$-DisC diverse and the  $r^{\prime}$-DisC diverse subsets of a set of objects $\mathcal{P}$, for $r$ $\neq$  $r^{\prime}$.
\end{observation}
\vspace{-0.1cm}

For zooming-in, i.e., for $r^{\prime}$ $<$  $r$, we can construct $r^{\prime}$-DisC diverse sets that are supersets of $S^r$ by adding objects to  $S^r$ to make it maximal. For zooming-out, i.e., for $r^{\prime}$ $>$  $r$, in general, there may be no subset of $S^r$ that is  $r^{\prime}$-DisC diverse. Take for example the objects of Figure~\ref{fig:zooming}\subref{fig:larger-rnew} with  $S^r$ $=$ $\{p_1, p_2, p_3$\}. No subset of $S^r$ is an  $r^{\prime}$-DisC diverse subset for this set of objects. Next, we detail the zooming-in and zooming-out operations.

\subsection{Incremental Zooming-in}
Let us first consider the case of zooming with a smaller radius, i.e., $r^{\prime} < r$. Here, we aim at producing a small independent covering solution $S^{r^{\prime}}$, such that, $S^{r^{\prime}}$ $\supseteq$ $S^r$. For this reason, we keep the objects of $S^r$ in the new $r^{\prime}$-DisC diverse subset $S^{r^{\prime}}$ and proceed as follows.

Consider an object of $S^r$, for example $p_1$ in Figure~\ref{fig:zooming}\subref{fig:smaller-rnew}. Objects at distance at most $r^{\prime}$ from $p_1$ are still covered by $p_1$ and cannot enter $S^{r^{\prime}}$. Objects at distance greater than $r^{\prime}$ and at most $r$ may be uncovered and join $S^{r^{\prime}}$. Each of these objects can enter $S^{r^{\prime}}$ as long as it is not covered by some other object of $S^r$ that lays outside the former neighborhood of $p_1$. For example, in Figure~\ref{fig:zooming}\subref{fig:smaller-rnew}, $p_4$ and $p_5$ may enter $S^{r^{\prime}}$ while $p_3$ can not, since, even with the smaller radius $r^{\prime}$, $p_3$ is covered by $p_2$.

To produce an $r^{\prime}$-DisC diverse subset based on an $r$-DisC diverse subset, we consider such objects in turn. This turn can be either arbitrary (\texttt{Zoom-In} algorithm) or proceed in a greedy way, where at each turn the object that covers the largest number of uncovered objects is selected (\texttt{Greedy-Zoom-} \texttt{In}, Algorithm~\ref{alg:zoom-in}).

\begin{algorithm}[t]
	\small
	\caption{Greedy-Zoom-In}
	\label{alg:zoom-in}	
	\begin{algorithmic}[1]
		\REQUIRE{A set of objects $\mathcal{P}$, an initial radius $r$, a solution $S^r$ and a new radius $r^{\prime} < r$.}
		\ENSURE{An ${r^{\prime}}$-DisC diverse subset of $\mathcal{P}$.}
		\vspace{0.1cm}
		\hrule
		\vspace{0.1cm}
		
		\STATE $S^{r^{\prime}} \leftarrow S^r$
		
		\FORALL {$p_i \in S^r$}
			\STATE Color objects in $\{N_r(p_i) \backslash N_{r^\prime}(p_i)\}$ white
		\ENDFOR
			
		\WHILE {there exist white objects}		
			\STATE Select the white object $p_i$ with the largest $\left|N_r^W(p_i)\right|$
			\STATE Color $p_i$ black
			\STATE $S^{r^{\prime}} = S^{r^{\prime}} \cup \{p_i\}$			
			\FORALL {$p_j \in N_r^W(p_i)$}
				\STATE Color $p_j$ grey			
			\ENDFOR
		\ENDWHILE
				
 		\RETURN $S^{r^{\prime}}$
	\end{algorithmic}
\end{algorithm}

\vspace{-0.1cm}
\begin{lemma}
For the set $S^{r^{\prime}}$ generated by the \texttt{Zoom-In} and \texttt{Greedy-Zoom-In} algorithms, it holds that:
\begin{enumerate}[(i)]\vspace{-0.2cm}
	\item $S^r \subseteq S^{r^{\prime}}$ and\vspace{-0.2cm}
	\item $|S^{r^{\prime}}| \leq N_{r^{\prime}, r}^I(p_i)|S^r|$\vspace{-0.1cm}
\end{enumerate}
\end{lemma}
\begin{proof}
Condition (i) trivially holds from step 1 of the algorithm. Condition (ii) holds since for each object in $S^r$ there are at most $N_{r^{\prime}, r}^I(p_i)$ independent objects at distance greater than $r^{\prime}$ from each other that can enter  $S^{r^{\prime}}$.
\end{proof}
\vspace{-0.1cm}

In practice, objects selected to enter $S^{r^{\prime}}$, such as $p_4$ and $p_5$ in Figure~\ref{fig:zooming}\subref{fig:smaller-rnew}, are likely to cover other objects left uncovered by the same or similar objects in $S^r$. Therefore, the size difference between $S^r$ and $S^{r^{\prime}}$ is expected to be smaller than this theoretical upper bound.

\subsection{Incremental Zooming-out}
Next, we consider zooming with a larger radius, i.e., $r^{\prime} > r$. In this case, the user is interested in seeing less and more dissimilar objects, ideally  a subset of the already seen results for $r$, that is, $S^{r^{\prime}} \subseteq S^r$. However, in this case, as discussed, in contrast to zooming-in, it may not be possible to construct a diverse subset $S^{r^{\prime}}$ that is a subset of $S^r$.

Thus, we focus on the following sets of objects: (i)~$S^r \backslash S^{r^{\prime}}$ and (ii)~$S^{r^{\prime}} \backslash S^r$. The first set consists of the objects that belong to the previous diverse subset but are removed from the new one, while the second set consists of the new objects added to $S^{r^{\prime}}$. To illustrate, let us consider for example the objects of Figure~\ref{fig:zooming}\subref{fig:larger-rnew} and that $p_1$, $p_2$, $p_3$ $\in S^r$. Since the radius becomes larger, $p_1$ now covers all objects at distance at most $r^{\prime}$ from it. This may include a number of objects that also belonged to $S^r$, such as $p_2$. These objects have to be removed from the solution, since they are no longer dissimilar to $p_1$. However, removing such an object, say $p_2$ in our example, can potentially leave uncovered a number of objects that were previously covered by $p_2$ (these objects lie in the shaded area of Figure~\ref{fig:zooming}\subref{fig:larger-rnew}). In our example, requiring $p_1$ to remain in $S^{r^{\prime}}$ means than $p_5$ should be now added to $S^{r^{\prime}}$.

To produce an $r^{\prime}$-DisC diverse subset based on an $r$-DisC diverse subset, we proceed in two passes. In the first pass, we examine all objects of $S^r$ in some order and remove their diverse neighbors that are now covered by them. At the second pass, objects from any uncovered areas are added to $S^{r^{\prime}}$. Again, we have an arbitrary and a greedy variation, denoted \texttt{Zoom-Out} and \texttt{Greedy-Zoom-Out} respectively. Algorithm~\ref{alg:zoom-out} shows the greedy variation; the first pass (lines 4-11) considers $S^r \backslash S^{r^{\prime}}$, while the second pass (lines 12-19) considers $S^{r^{\prime}} \backslash S^r$. Initially, we color all previously black objects red. All other objects are colored white. We consider three variations for the first pass of the greedy algorithm: selecting the red objects with (a)~the largest number of red neighbors, (b)~the smallest number of red neighbors and (c)~the largest number of white neighbors. Variations (a) and (c) aim at minimizing the objects to be added in the second pass, that is, $S^{r^{\prime}} \backslash S^r$, while variation (b) aims at maximizing $S^r$ $\cap$ $S^{r^{\prime}}$. Algorithm~\ref{alg:zoom-out} depicts variation (a), where $N_{r^{\prime}}^R(p_i)$ denotes the red neighbors of object $p_i$.

\begin{algorithm}[t]
	\caption{Greedy-Zoom-Out(a)}
	\small
	\label{alg:zoom-out}	
	\begin{algorithmic}[1]
		\REQUIRE{A set of objects $\mathcal{P}$, an initial radius $r$, a solution $S^r$ and a new radius $r^{\prime} > r$.}
		\ENSURE{An ${r^{\prime}}$-DisC diverse subset of $\mathcal{P}$.}		\vspace{0.1cm}
		\hrule
		\vspace{0.1cm}
		
		\STATE $S^{r^{\prime}} \leftarrow \emptyset$
		
		\STATE Color all black objects red
		\STATE Color all grey objects white
		
		\WHILE{there exist red objects}
			\STATE Select the red object $p_i$ with the largest $|N_{r^{\prime}}^R(p_i)|$
			\STATE Color $p_i$ black
			\STATE $S^{r^{\prime}} = S^{r^{\prime}} \cup \{p_i\}$
			
			\FORALL {$p_j \in N_{r^{\prime}}(p_i)$}
				\STATE Color $p_j$ grey			
			\ENDFOR
		\ENDWHILE
		
		\WHILE{there exist white objects}
			\STATE Select the white object $p_i$ with the larger $|N_{r^{\prime}}^W(p_i)|$
			\STATE Color $p_i$ black
			\STATE $S^{r^{\prime}} = S^{r^{\prime}} \cup \{p_i\}$
			
			\FORALL {$p_j \in N_{r^{\prime}}^W(p_i)$}
				\STATE Color $p_j$ grey			
			\ENDFOR
		\ENDWHILE
			
 		\RETURN $S^{r^{\prime}}$
	\end{algorithmic}
\end{algorithm}

\begin{figure*}
	\vspace{-0.2cm}
	\centering
	\subfloat[$r$-DisC.]{		
		\includegraphics[width=3.3cm]{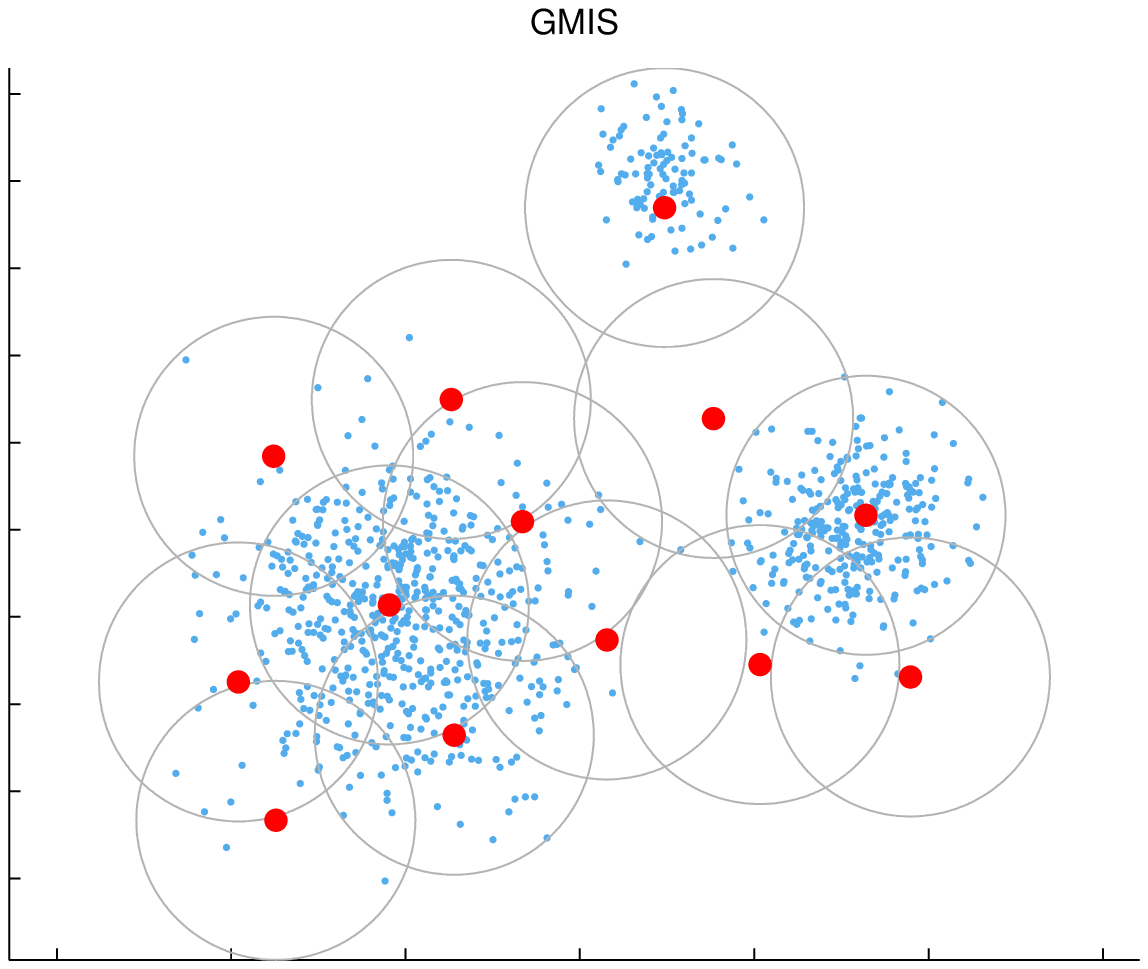}
	}
	\subfloat[MaxSum.]{		
		\includegraphics[width=3.3cm]{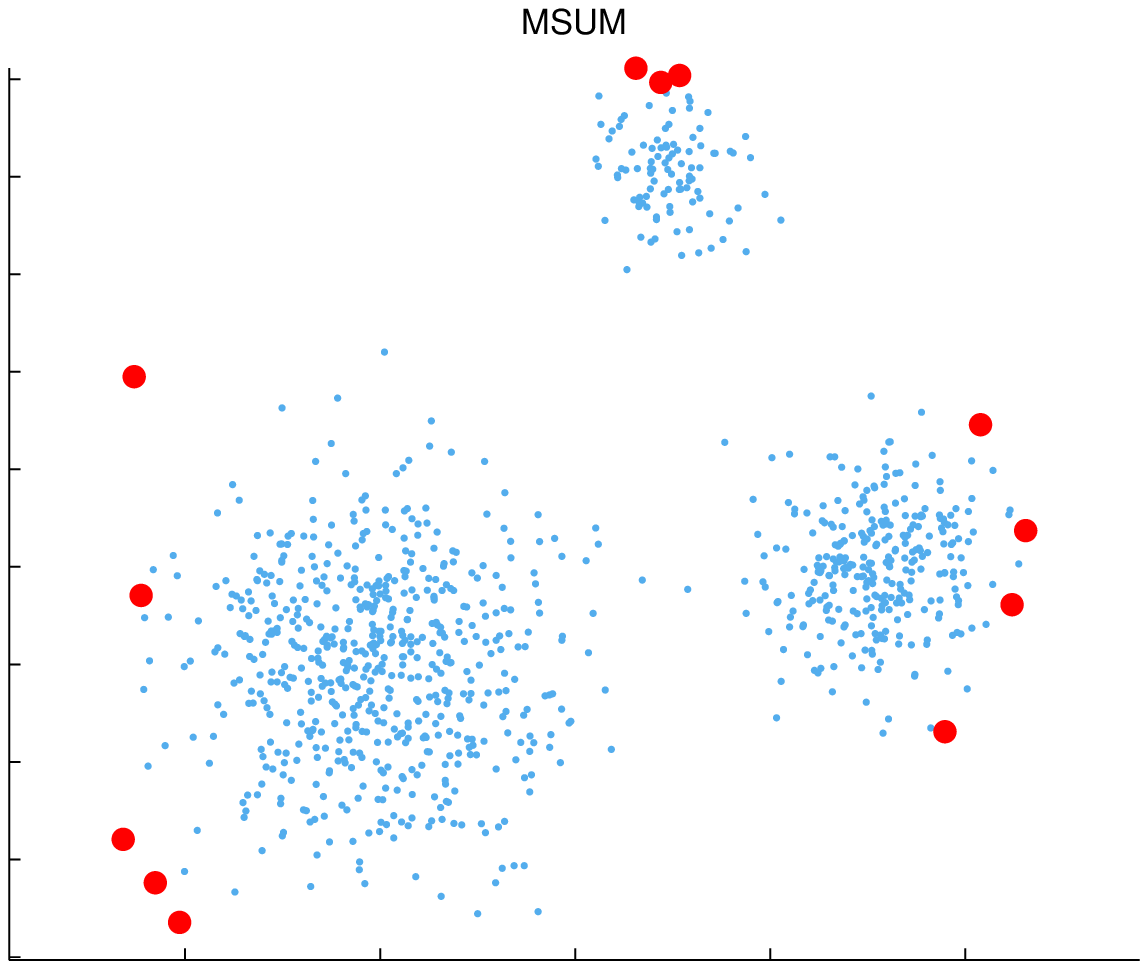}
		\label{subfig:fig:methods-clustered-MSUM}
	}
	\subfloat[MaxMin.]{		
		\includegraphics[width=3.3cm]{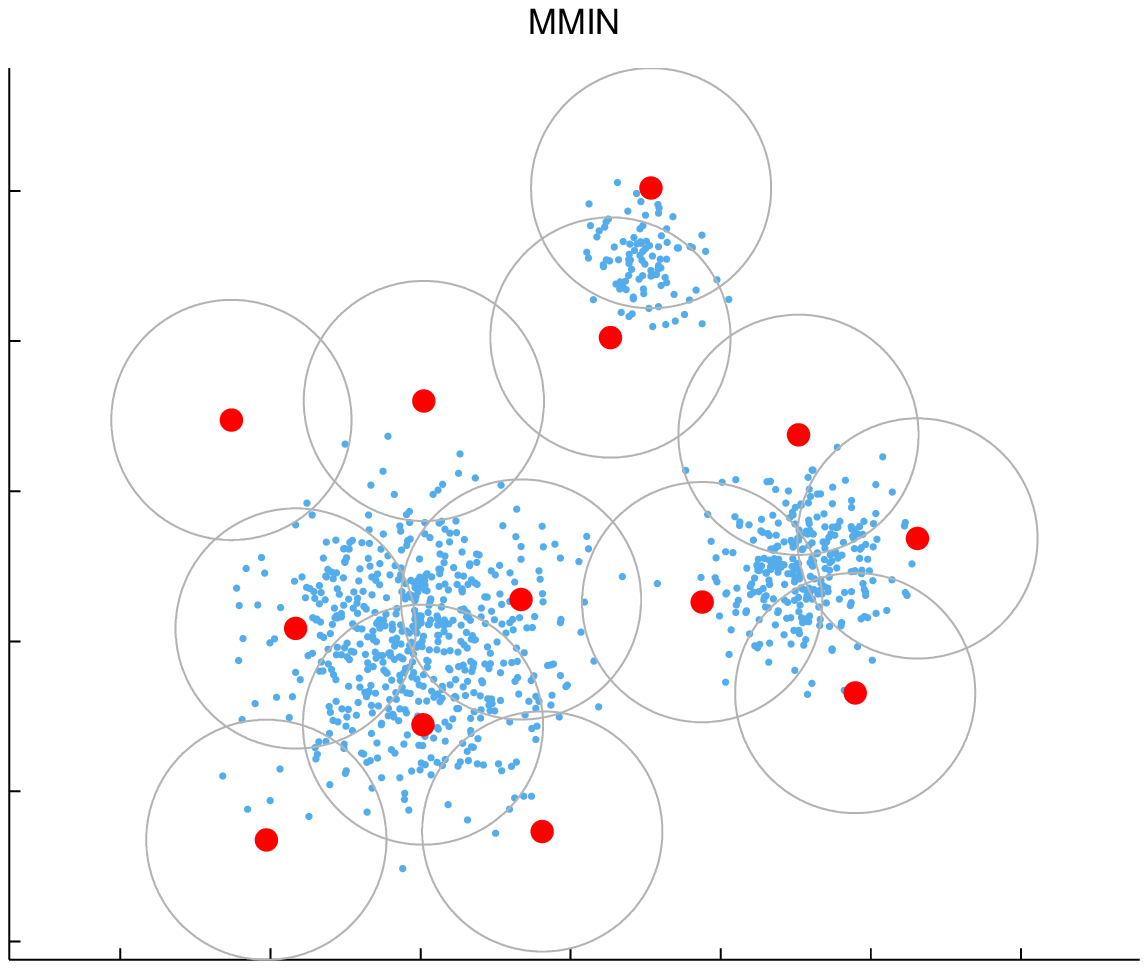}
	}
	\subfloat[$k$-medoids.]{		
		\includegraphics[width=3.3cm]{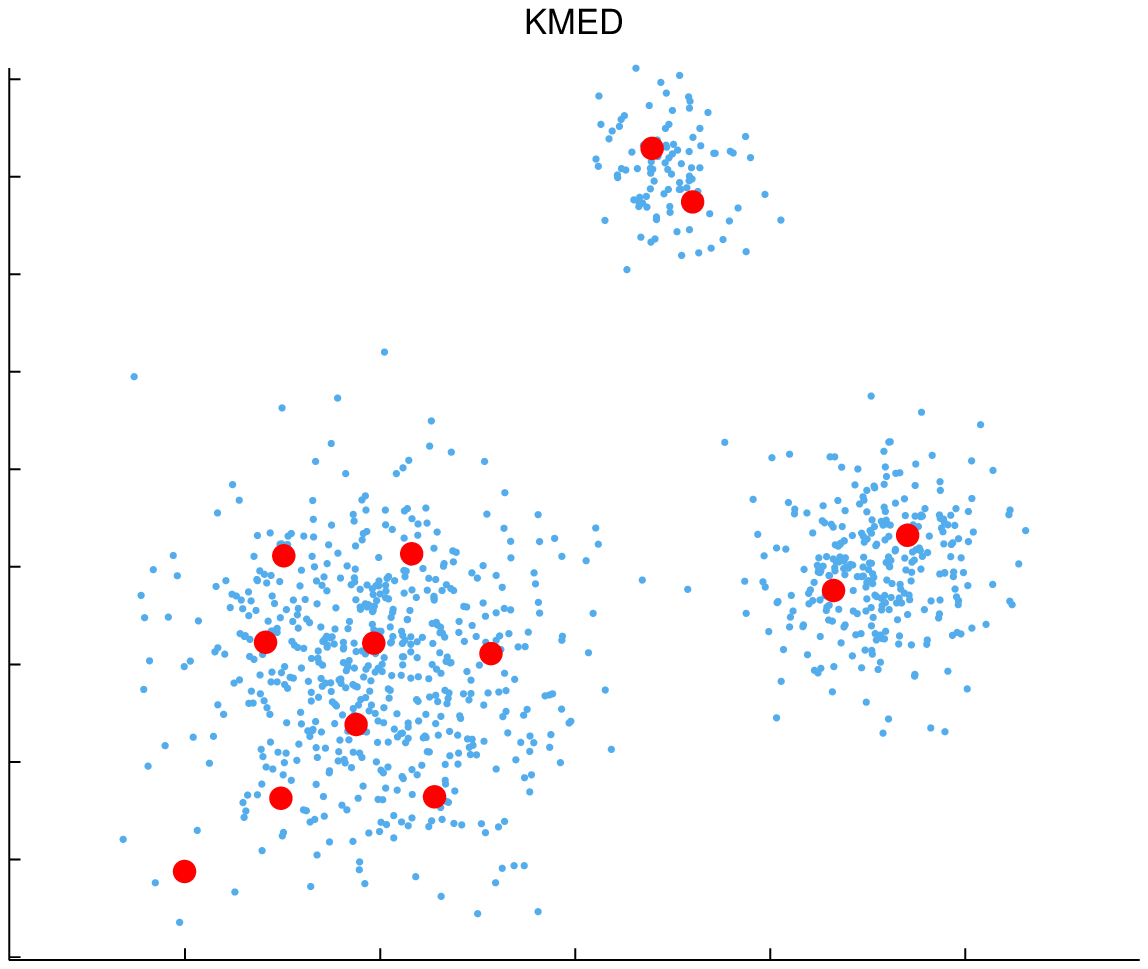}
	}
	\subfloat[$r$-C.]{		
		\includegraphics[width=3.3cm]{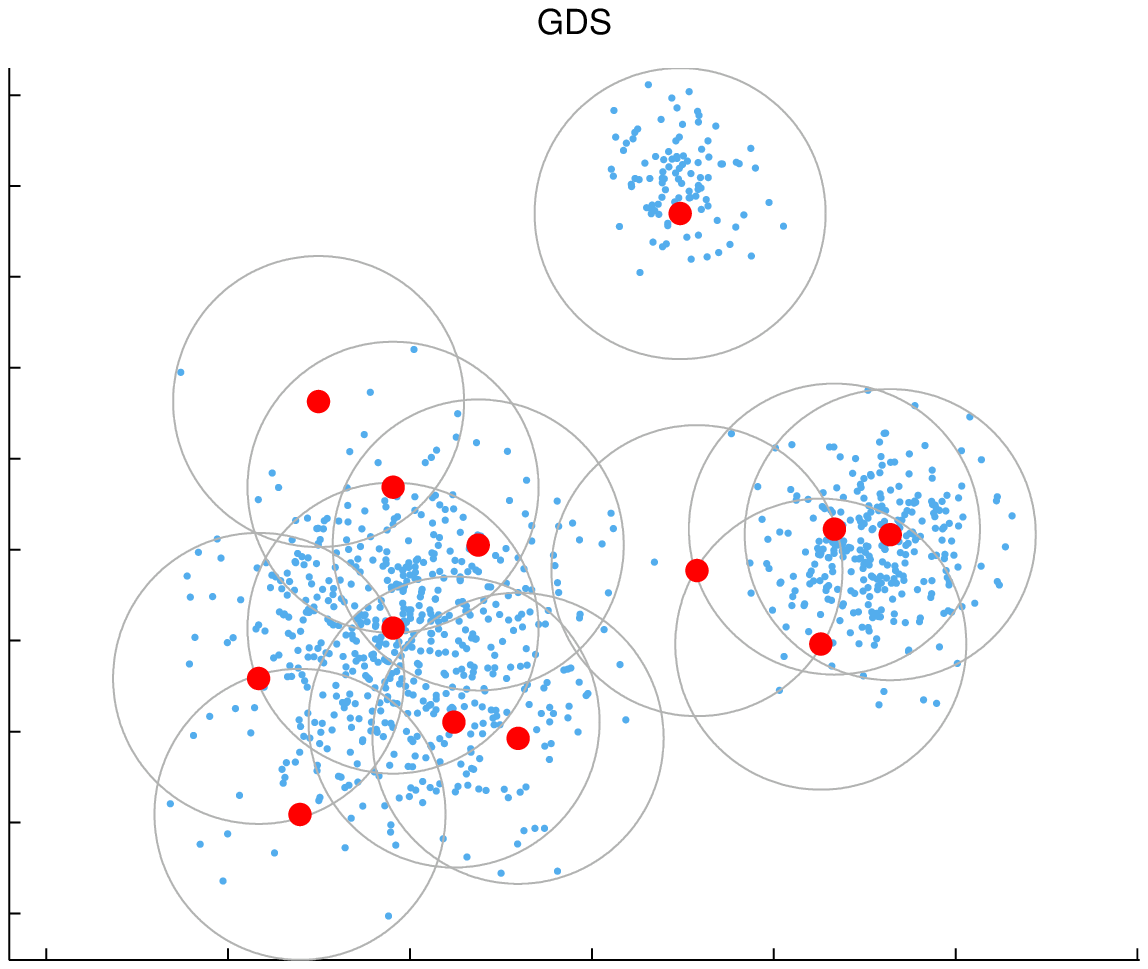}
	}
	\vspace{-0.2cm}
	\caption{Solutions by the various diversification methods for a clustered dataset. Selected objects are shown in bold. Solid circles denote the radius $r$ of the selected objects.}
	\vspace{-0.2cm}
	\label{fig:methods-clustered}
\end{figure*}

\vspace{-0.1cm}
\begin{lemma}
For the solution $S^{r^{\prime}}$ generated by the \texttt{Zoom-Out} and \texttt{Greedy-Zoom-Out} algorithms, it holds that:
\begin{enumerate}\vspace{-0.2cm}
	\item [(i)]There are at most $N_{r, r^{\prime}}^I(p_i)$ objects in $S^r \backslash S^{r^{\prime}}$.\vspace{-0.2cm}
	\item [(ii)]For each object of $S^r$ not included in $S^{r^{\prime}}$, at most $B - 1$ objects are added to $S^{r^{\prime}}$.\vspace{-0.1cm}
\end{enumerate}
\end{lemma}
\begin{proof}
Condition (i) is a direct consequence of the definition of $N_{r, r^{\prime}}^I(p_i)$. Concerning condition (ii), recall that each removed object $p_i$ has at most $B$ independent neighbors for $r^{\prime}$. Since $p_i$ is covered by some neighbor, there are at most $B-1$ other independent objects that can potentially enter $S^{r^{\prime}}$.
\end{proof}
\vspace{-0.1cm}

As before, objects left uncovered by objects such as $p_2$ in Figure~\ref{fig:zooming}\subref{fig:larger-rnew} may already be covered by other objects in the new solution (consider $p_4$ in our example which is now covered by $p_3$). However, when trying to adapt a DisC diverse subset to a larger radius, i.e., maintain some common objects between the two subsets, there is no theoretical guarantee that the size of the new solution will be reduced.

\section{Comparison with Other Models}
\label{sec:Comparison}
The most widely used diversification models 
are \textsc{MaxMin} and \textsc{MaxSum} that aim at selecting objects that maximize $f_{\textsc{Min}}$ $=$ $\min_{\substack{p_i, p_j \in S\\p_i \neq p_j}}$ $dist(p_i, p_j)$ and $f_{\textsc{Sum}}$ $=$ $\sum_{\substack{p_i, p_j \in S\\p_i \neq p_j}} dist(p_i, p_j)$ respectively (e.g. \cite{DBLP:conf/www/GollapudiS09, DBLP:conf/icde/VieiraRBHSTT11, pods12}). Concerning the quality of an $r$-DisC solution to the optimal value of the \textsc{MaxMin} metric, the following lemma holds.

\begin{lemma}
\label{maxmin}
Let $\cal{P}$ be a set of objects, $S$ be an $r$-DisC diverse subset of $\cal{P}$ and $\lambda$ $\geq$ $r$ be the $f_{\textsc{Min}}$ distance between objects of $S$. Let $S^*$ be an optimal \textsc{MaxMin} subset of $\cal{P}$ for $k$ = $|S|$ and $\lambda^*$ be the $f_{\textsc{Min}}$ distance for $S^*$. Then, $\lambda^*$ $\leq$ 3 $\lambda$.
\end{lemma}

\begin{proof}
Each object in $S^*$ is covered by (at least) one object in $S$. There are two cases, either (i) all objects $p_1^*$, $p_2^*$ $\in$ $S^*$, $p_1^*$ $\neq$ $p_2^*$, are covered by  different objects is $S$, or (ii) there are at least two objects in $S^*$, $p_1^*$, $p_2^*$, $p_1^*$ $\neq$ $p_2^*$ that are both covered by the same object $p$ in $S$. Case (i): Let $p_1$ and $p_2$ be two objects in $S$ such that $dist(p_1, p_2)$ = $d$ and $p_1^*$ and $p_2^*$ respectively be the object in $S^*$ that each covers. Then, by applying the triangle inequality twice, we get: $dist(p_1^*,  p_2^*)$ $\leq$ $dist(p_1^*,  p_1)$ + $dist(p_1,  p_2^*)$ $\leq$ $dist(p_1^*,  p_1)$ + $dist(p_1,  p_2)$ + $dist(p_2,  p_2^*)$. By coverage, we get: $dist(p_1^*,  p_2^*)$ $\leq$ $r$ + $\lambda$ + $r$ $\leq$ 3 $\lambda$, thus $\lambda^*$ $\leq$ 3 $\lambda$. Case (b): Let $p_1^*$ and $p_2^*$ be two objects in $S^*$ that are covered by the same object $p$ in $S$. Then, by coverage and the triangle inequality, we get $dist(p_1^*,  p_2^*)$ $\leq$ $dist(p_1^*,  p)$ + $dist(p,  p_2^*)$ $\leq$ 2 $r$, thus $\lambda^*$ $\leq$ 2 $\lambda$.
\end{proof}

%

Next, we show qualitative results of applying \textsc{MaxMin} and \textsc{MaxSum} to a 2-dimensional ``Clustered'' dataset. To implement \textsc{MaxMin} and \textsc{MaxSum}, we used greedy heuristics which have been shown to achieve good solutions \cite{DBLP:journals/sigmod/DrosouP10}. In addition to \textsc{MaxMin} and \textsc{MaxSum}, we also show results for $r$-C-diversity (i.e., covering but not necessarily independent subsets for the given $r$) and $k$-medoids, a widespread clustering algorithm that seeks to minimize $\frac{1}{|\mathcal{P}|} \sum_{p_i \in \mathcal{P}} dist(p_i, c(p_i))$, where $c(p_i)$ is the closest object of $p_i$ in the selected subset, since the located medoids can be viewed as a representative subset of the dataset. To allow for a comparison, we first run our algorithms for a given $r$ and then use as $k$ the size of the produced diverse subset. In this example, $k$ = 15 for $r$ = 0.7.

\textsc{MaxSum} diversification and $k$-medoids fail to cover all areas of the dataset; \textsc{MaxSum} tends to  focus on the outskirts of the dataset, whereas $k$-medoids reports only central points, ignoring outliers. \textsc{MaxMin} performs better in this aspect. However, since \textsc{MaxMin} seeks to retrieve objects that are as far apart as possible, it fails to retrieve objects from dense areas; see, for example, the center areas of the clusters in Figure~\ref{fig:methods-clustered}. DisC gives priority to such areas and, thus, such areas are better represented in the solution. Note also that \textsc{MaxSum} and $k$-medoids may select duplicate objects while DisC and \textsc{MaxMin} do not. We also experimented with variations of \textsc{MaxSum} proposed in \cite{DBLP:conf/icde/VieiraRBHSTT11} but the results did not differ substantially from the ones in Figure~\ref{fig:methods-clustered}\subref{subfig:fig:methods-clustered-MSUM}. For $r$-C diversity, the resulting selected set is one object smaller, however, the selected objects are less widely spread than in DisC. Finally, note that, we are not interested in retrieving as representatives subsets that follow the same distribution as the input dataset, as in the case of sampling, since such subsets will tend to ignore outliers. Instead, we want to cover the whole dataset and provide a complete view of all its objects, including the distant ones.

\section{Implementation}
\label{sec:Indexing}
Since a central operation in computing DisC diverse subsets is locating neighbors,
we introduce algorithms that exploit a spatial index structure, 
namely,  the M-tree \cite{similaritysearch2006}. An M-tree is a balanced tree index that can handle large volumes of dynamic data of any dimensionality in general metric spaces. In particular, an M-tree partitions space around some of the indexed objects, called {\em pivots}, by forming a bounding ball region of some {\em covering radius} around them. Let $c$ be the maximum node capacity of the tree. Internal nodes have at most $c$ entries, each containing a pivot object $p_v$, the covering radius $r_v$ around $p_v$, the distance of $p_v$ from its parent pivot and a pointer to the subtree $t_v$. All objects in the subtree $t_v$ rooted at $p_v$ are within distance at most equal to the covering radius $r_v$ from $p_v$. Leaf nodes have entries containing the indexed objects and their distance from their parent pivot. 

The construction of an M-tree is influenced by the \emph{splitting policy} that determines how nodes are split when they exceed their maximum capacity $c$. Splitting policies indicate (i)~which two of the $c+1$ available pivots will be promoted to the parent node to index the two new nodes (\emph{promote policy}) and (ii)~how the rest of the pivots will be assigned to the two new nodes (\emph{partition policy}). These policies affect the overlap among nodes. For computing diverse subsets:
\begin{enumerate}[(i)]\vspace{-0.2cm}
	\item We link together all leaf nodes.
This allows us to visit all objects in a single left-to-right traversal of the leaf nodes and exploit some degree of locality in covering the objects.\vspace{-0.2cm}
	\item To compute the neighbors $N_r(p_i)$ of an object $p_i$ at radius $r$, we perform a range query centered around $p_i$ with distance $r$, denoted $Q(p_i, r)$. Range queries can be performed either in a top-down fashion starting from the root node or in a bottom-up fashion starting from $p_i$. We consider both variations. \vspace{-0.2cm}
	\item We build trees using splitting policies that minimize overlap. In most cases, the policy that resulted in the lowest overlap was (a) promoting as new pivots the pivot $p_i$ of the overflowed node and the object $p_j$ with the maximum distance from $p_i$ and (b) partitioning the objects by assigning each object to the node whose pivot has the closest distance with the object. We call this policy ``MinOverlap''.
\end{enumerate}

\subsection{Computing Diverse Subsets}
The \texttt{Basic-DisC} algorithm selects white objects in random order. The M-tree implementation of \texttt{Basic-DisC} allows us to consider objects in the order they appear in the leaves of the M-tree, thus taking advantage of locality. Upon encountering a white object $p_i$, the algorithm colors it black and executes a range query $Q(p_i, r)$ to retrieve the neighbors of $p_i$ and color them grey. If the overlap among nodes is small, the neighbors of an indexed object are expected to reside in nearby leaf nodes, thus such range queries are in general efficient. We can visualize the progress of \texttt{Basic-DisC} as gradually coloring all objects in the leaf nodes from left-to-right until all objects are either grey or black. 

We make the following observation that allows us to further \emph{prune} subtrees while executing range queries. Objects that are already grey do not need to be colored grey again when some other of their neighbors is colored black. 

\vspace{0.1cm}
\noindent {\sc Pruning Rule:}
A leaf node that contains no white objects is colored grey. When all its children become grey, an internal node is colored grey. While executing range queries, the top-down search of the tree does not need to follow subtrees rooted at grey nodes.
\vspace{0.1cm}

As the algorithm progresses, more and more nodes become grey, and thus, the cost of range queries reduces over time. We call this variation \texttt{Basic-DisC (Pruned)}. We can visualize its progress as gradually coloring all tree nodes grey in a post-order manner.

The \texttt{Greedy-DisC} algorithm selects at each iteration the white object with the largest white neighborhood. To efficiently implement \texttt{Greedy-DisC}, we maintain all white objects in a structure $L^{\prime}$ sorted by the size of their white neighborhood. For initializing $L^{\prime}$, we need to compute the size of the white neighborhoods of all objects. Initially, $N_r^W(p_i) = N_r(p_i)$ for every object $p_i$. We opt to compute the neighborhood size of each object as we build the M-tree. When an object $p_i$ in inserted, a range query $Q(p_i, r)$ is executed. The white neighborhood of $p_i$ is initialized to $|Q(p_i, r)|$ and the white neighborhoods of all objects retrieved by the range query are incremented by one. We found that computing the size of neighborhoods while building the tree reduces node accesses up to 45\%. 

At each iteration of the algorithm, the first element $p_i$ of $L^{\prime}$ is selected and colored black
and a range query $Q(p_i, r)$ is used to retrieve and color grey the neighbors of $p_i$. 
We also need to update the size of the white neighborhoods of all affected objects, i.e., all objects in
the  neighborhood of  each $p_j$ in $N_r(p_i)$.  
We consider two variations.
The first variation, termed \texttt{Grey-Greedy-DisC}, executes an additional range query $Q(p_j, r)$ for each of the
newly colored grey neighbors $p_j$ of $p_i$  
for locating the neighbors of $p_j$ and reducing by one the size of their white neighborhood. 
The second variation, termed \texttt{White-Greedy-DisC}, executes one range query for all remaining white objects with distance less than or equal to $2r$ from $p_i$. These are the only white objects for which the size of their white neighborhood may have changed. As before, we can reduce the cost using the pruning rule. We call these variations \texttt{Grey-Greedy-DisC (Pruned)} and \texttt{White-Greedy-DisC} \texttt{(Pruned)}.

Finally, \texttt{Greedy-C} considers both grey and white objects 
as candidates. A sorted structure $L^{\prime}$ has to be maintained 
as well, which now includes both white and grey objects and is substantially larger. Furthermore, the pruning rule is no longer useful, since grey objects and nodes need to be accessed again for  updating the size of their white neighborhood. We use the pruning rule to introduce a faster heuristic for computing $r$-C diverse subsets, called \texttt{Fast-C}. Whenever an object $p_i$ is colored black, \texttt{Fast-C} executes a range query to retrieve its neighbors by traversing the tree bottom-up. The query stops climbing up the tree when the first grey internal node is met. This may lead to failing to reach and thus color grey some of the neighbors of $p_i$ that reside in distant leaf nodes and thus produce 
larger results. 
However, in a tree with small overlap, we expect the number of such neighbors to be small.

\subsection{Adapting the Radius}
For zooming-in, given an $r$-DisC diverse subset $S^r$ of $\mathcal{P}$, we would like to compute an $r^{\prime}$-DisC diverse subset $S^{r^{\prime}}$ of $\mathcal{P}$, $r^{\prime}$ $<$ $r$, such that, $S^{r^{\prime}}$ $\supseteq$ $S^r$. A naive implementation would require two range queries per object in $S^r$ plus any additional range queries required to cover newly uncovered areas. During the construction of $S^r$, however, objects in the corresponding M-tree are already colored black or grey. We can exploit this information to efficiently implement the \texttt{Zoom-In} and \texttt{Greedy-Zoom-In} algorithms.

\vspace{0.1cm}
\noindent {\sc Zooming Rule:}
Black objects of $S^r$ maintain their color in $S^{r^{\prime}}$. Grey objects maintain their color as long as there exists a black object at distance at most $r^{\prime}$ from them. Therefore, only grey nodes with no black neighbors at distance $r^{\prime}$ may turn black and enter $S^{r^{\prime}}$. 
\vspace{0.1cm}

\begin{table}
	\vspace{-0.3cm}
	\scriptsize
	\renewcommand{\arraystretch}{1.1}
	\caption{Input parameters.}
	\centering
		\begin{tabular}{|l|l|l|}
			\hline
			\textbf{Parameter} & \textbf{Default value} & \textbf{Range}\\ \hline			
			M-tree node capacity & 50 & 25 - 100\\ \hline
			M-tree splitting policy & MinOverlap & various\\ \hline
			Dataset cardinality & 10000 & 579 - 50000\\ \hline
			Dataset dimensionality & 2 & 2 - 10\\ \hline
			Dataset distribution & normal & uniform, normal\\ \hline
			Distance metric & Euclidean & Euclidean, Hamming\\ \hline
	\end{tabular}
	\label{tab:parameters}
	\vspace{-0.3cm}
\end{table}

To take advantage of the {\sc Zooming Rule}, 
we extend the leaf nodes of the M-tree to include the distance of the indexed object $p_i$ to its \emph{closest} black neighbor $p_j$, since $p_i$ will continue to be covered by $p_j$ for all $r^{\prime} \leq dist(p_i, p_j)$.

The \texttt{Zoom-In} algorithm requires one pass of the leaf nodes. Each time a grey object $p_i$ is encountered, its distance from its closest black neighbor is compared against the new radius $r^{\prime}$. In case this distance is larger than $r^{\prime}$, $p_i$ is colored black and a range query $Q(p_i, r^{\prime})$ is executed to locate any objects for which $p_i$ is now the closest black neighbor and color them grey. At the end of the pass, the black objects of the leaves form $S^{r^{\prime}}$. \texttt{Greedy-Zoom-In} also requires maintenance of a sorted structure $L^{\prime}$. First, the leaf nodes are traversed, grey objects that are now uncovered are colored white and inserted into $L^{\prime}$. Then, the white neighborhoods of objects in $L^{\prime}$ are computed and $L^{\prime}$ is sorted accordingly. Finally, the first object $p_i$ of $L^{\prime}$ is retrieved and colored black; a range query $Q(p_i, r^{\prime})$ is executed and retrieved objects are colored grey. Black and grey objects are removed from $L^{\prime}$. This process is repeated until $L^{\prime}$ is empty.

Note that, our pruning technique during the construction of $S^r$ interferes with the correct computation of the distances to the closest back neighbor of the objects. Therefore, an additional post-processing step is required after the construction of $S^r$ to compute these distances.

Zooming-out algorithms are implemented similarly. A sorted structure $L^{\prime}$ is employed at the first pass of the greedy variations to process red objects in the desired order, while the same structure is used at the second step to process white objects. Finally, for local zooming in an object $p_i$, the only difference is that instead of all objects in $\mathcal{P}$, the algorithm receives as input only the objects in $N_{r}(p_i)$.

\vspace{0.2cm}
\section{Experimental Evaluation}
\label{sec:Evaluation}

\begin{table}
	\vspace{-0.3cm}
	\scriptsize
	\renewcommand{\arraystretch}{1.1}
	\caption{Input parameters.}
	\centering
		\begin{tabular}{|l|l|l|}
			\hline
			\textbf{Parameter} & \textbf{Default value} & \textbf{Range}\\ \hline			
			M-tree node capacity & 50 & 25 - 100\\ \hline
			M-tree splitting policy & MinOverlap & various\\ \hline
			Dataset cardinality & 10000 & 579 - 50000\\ \hline
			Dataset dimensionality & 2 & 2 - 10\\ \hline
			Dataset distribution & normal & uniform, normal\\ \hline
			Distance metric & Euclidean & Euclidean, Hamming\\ \hline
	\end{tabular}
	\label{tab:parameters}
	\vspace{-0.3cm}
\end{table}

In this section, to evaluate our approach, we use both synthetic and real datasets. Our synthetic datasets consist of multi dimensional objects, where values at each dimension are in $\left[0, 1\right]$. Objects are either uniformly distributed in space (``Uniform'') or form (hyper) spherical clusters of different sizes (``Clustered''). We also employ two real datasets. The first one (``Cities'') is a collection of 2-dimensional points representing geographic information about 5922 cities and villages in Greece \cite{data:cities}. We normalized the values of this dataset in $\left[0, 1\right]$. The second real dataset (``Cameras'') consists of 7 characteristics for 579 digital cameras from \cite{acme}, such as brand and storage type.

We use the Euclidean distance for the synthetic datasets and ``Cities'', while for ''Cameras'', whose attributes are categorical, we use $dist(p_i, p_j)$ $=$ $\sum_i \delta^i(p_i, p_j)$, where $\delta^i(p_i, p_j)$ is equal to 1, if $p_i$ and $p_j$ differ in the $i^{\textnormal{th}}$ dimension and 0 otherwise, i.e., the Hamming distance. Note that the choice of an appropriate distance metric is an important but orthogonal to our approach issue.

Table~\ref{tab:parameters} summarizes the values of the input parameters used in our experiments.

\begin{table}
\centering
\scriptsize
\caption{Solution size for the \texttt{Basic-DisC}, \texttt{Greedy-C} and the variations of \texttt{Greedy-DisC} algorithms.}
\label{tab:size-comparison}
\renewcommand{\arraystretch}{1.05}
\subfloat[Uniform (2D - 10000 objects).]{
\begin{tabular}{|l|l|l|l|l|l|l|l|}
		\hline
		& \multicolumn{7}{c|}{$\mathbi{r}$}\\
		\cline{2-8}
		~	& \textbf{0.01}	& \textbf{0.02}	& \textbf{0.03}	& \textbf{0.04}	& \textbf{0.05}	& \textbf{0.06}	& \textbf{0.07}\\\hline
		\textbf{B-DisC}									& 3839	& 1360	& 676	& 411	& 269	& 192	& 145	\\
		\textbf{G-DisC}	& 3260	& 1120	& 561	& 352	& 239	& 176	& 130	\\
		\textbf{L-Gr-G-DisC}				& 3384	& 1254	& 630	& 378	& 253	& 184	& 137	\\
		\textbf{L-Wh-G-DisC}			& 3293	& 1152	& 589	& 352	& 240	& 170	& 130	\\
		\textbf{G-C}									& 3427	& 1104	& 541	& 338	& 230	& 170	& 126	\\
		\hline
	\end{tabular}
}\\\vspace{-0.1cm}
\subfloat[Clustered (2D - 10000 objects).]{
	\begin{tabular}{|l|l|l|l|l|l|l|l|}
		\hline
		& \multicolumn{7}{c|}{$\mathbi{r}$}\\
		\cline{2-8}
		~	& \textbf{0.01}	& \textbf{0.02}	& \textbf{0.03}	& \textbf{0.04}	& \textbf{0.05}	& \textbf{0.06}	& \textbf{0.07}\\\hline
		\textbf{B-DisC}									& 1018	& 370	& 193	& 121	& 80	& 61	& 48	\\
		\textbf{G-DisC}	& 892		& 326	& 162	& 102	& 69	& 52	& 43	\\
		\textbf{L-Gr-G-DisC}				& 680		& 394	& 218	& 133	& 87	& 64	& 49	\\
		\textbf{L-Wh-G-DisC}			& 906		& 313	& 168	& 104	& 70	& 52	& 41	\\
		\textbf{G-C}									& 895		& 322	& 166	& 102	& 71	& 50	& 43	\\
		\hline
	\end{tabular}
}\\\vspace{-0.1cm}
\subfloat[Cities.]{
	\begin{tabular}{|l|l|l|l|l|l|l|l|}
		\hline
		& \multicolumn{7}{c|}{$\mathbi{r} \cdot 10^{-2}$}\\
		\cline{2-8}
		~	& \textbf{0.10}	& \textbf{0.25}	& \textbf{0.50}	& \textbf{0.75}	& \textbf{1.00}	& \textbf{1.25}	& \textbf{1.50}\\\hline
		\textbf{B-DisC}									& 2534	& 743	& 269	& 144	& 96	& 68	& 11	\\
		\textbf{G-DisC}	& 2088		& 569	& 209	& 123	& 76	& 54	& 8	\\
		\textbf{L-Gr-G-DisC}				& 2166		& 638	& 239	& 124	& 90	& 62	& 9	\\
		\textbf{L-Wh-G-DisC}			& 2106		& 587	& 221	& 115	& 79	& 56	& 8	\\
		\textbf{G-C}									& 2188		& 571	& 205	& 117	& 79	& 52	& 7	\\
		\hline
	\end{tabular}
}\\\vspace{-0.1cm}
\subfloat[Cameras.]{
	\begin{tabular}{|l|l|l|l|l|l|l|}
		\hline
		& \multicolumn{6}{c|}{$\mathbi{r}$}\\
		\cline{2-7}
		~	& \textbf{1}	& \textbf{2}	& \textbf{3}	& \textbf{4}	& \textbf{5}	& \textbf{6}\\\hline
		\textbf{B-DisC}									& 461	& 237	& 103	& 34	& 9	& 4		\\
		\textbf{G-DisC}			& 461	& 212	& 78	& 28	& 9 & 2		\\
		\textbf{L-Gr-G-DisC}				& 461		& 216	& 80	& 31	& 9	& 2		\\
		\textbf{L-Wh-G-DisC}			& 461		& 212	& 81	& 28	& 9	& 2		\\
		\textbf{G-C}									& 461		& 218	& 74	& 25	& 6	& 2		\\
		\hline
	\end{tabular}
}
\vspace{-0.6cm}
\end{table}

\begin{figure*}
	\vspace{-0.2cm}
	\centering
	\subfloat[Uniform.]{		
		\includegraphics[width=4.0cm]{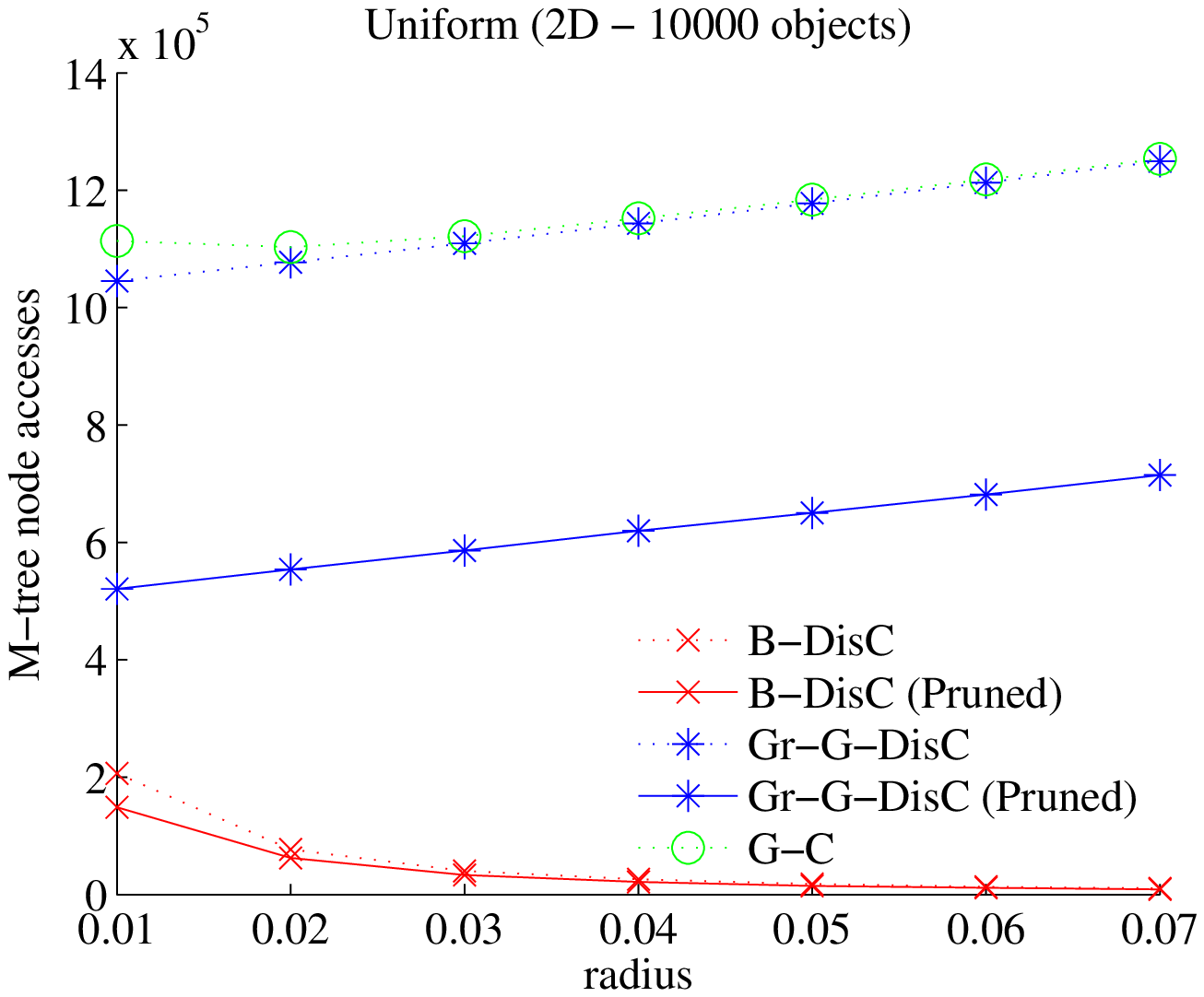}
	\label{fig:active-pruning1-uniform}
	}
	\subfloat[Clustered.]{		
		\includegraphics[width=4.0cm]{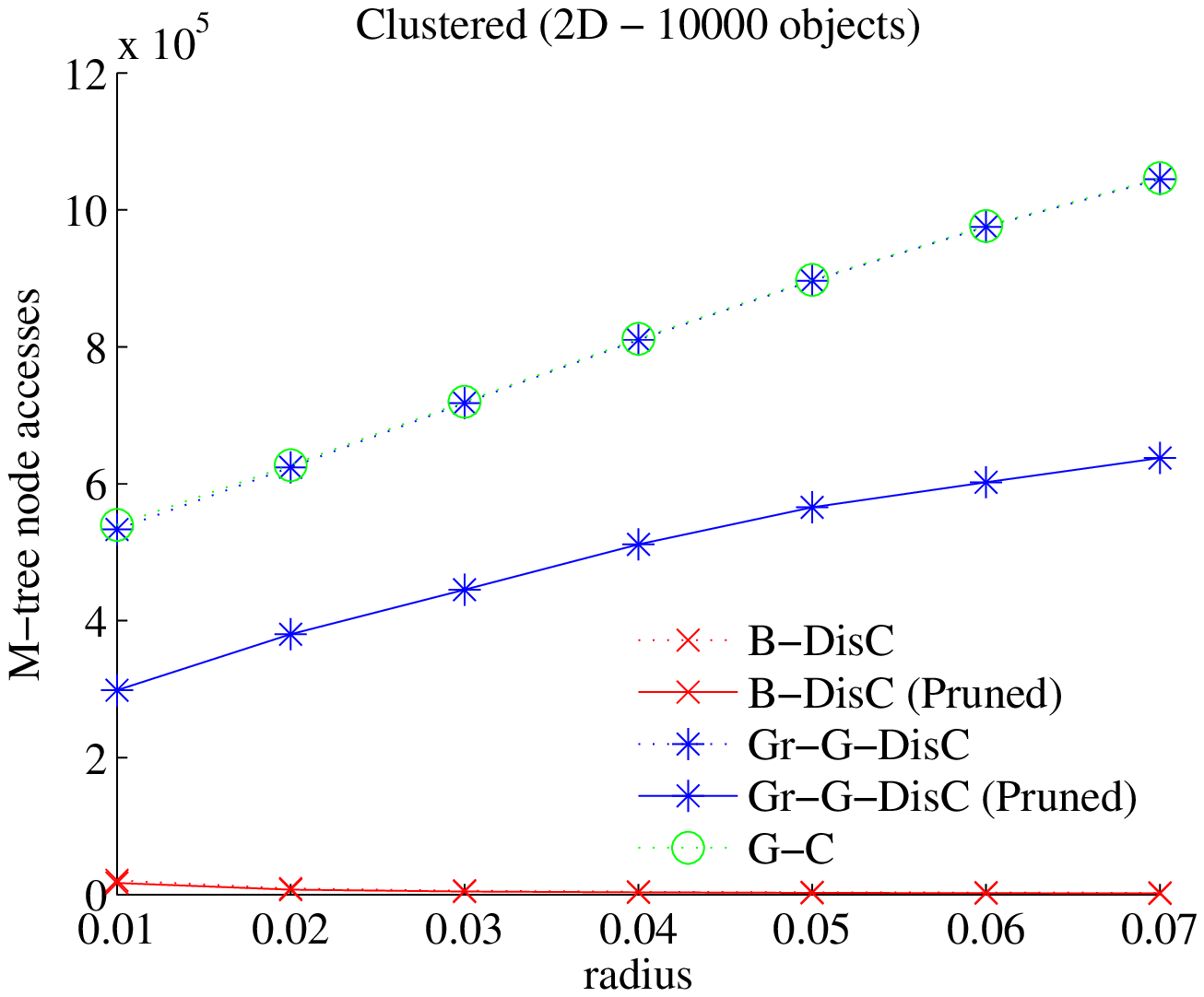}
	\label{fig:active-pruning1-clustered}
	}
	\subfloat[Cities.]{		
		\includegraphics[width=4.0cm]{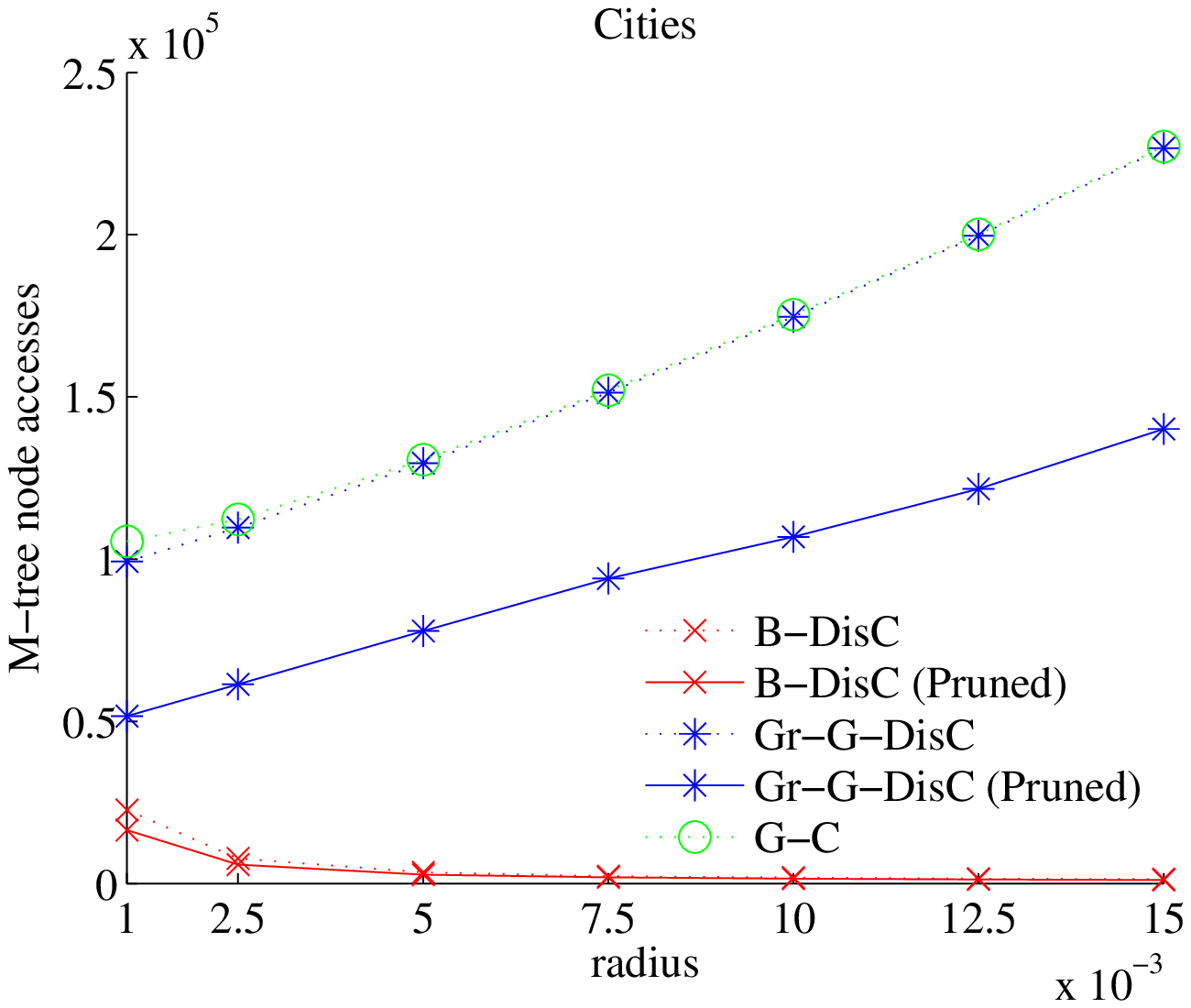}
	\label{fig:active-pruning1-cities}
	}
	\subfloat[Cameras.]{		
		\includegraphics[width=4.0cm]{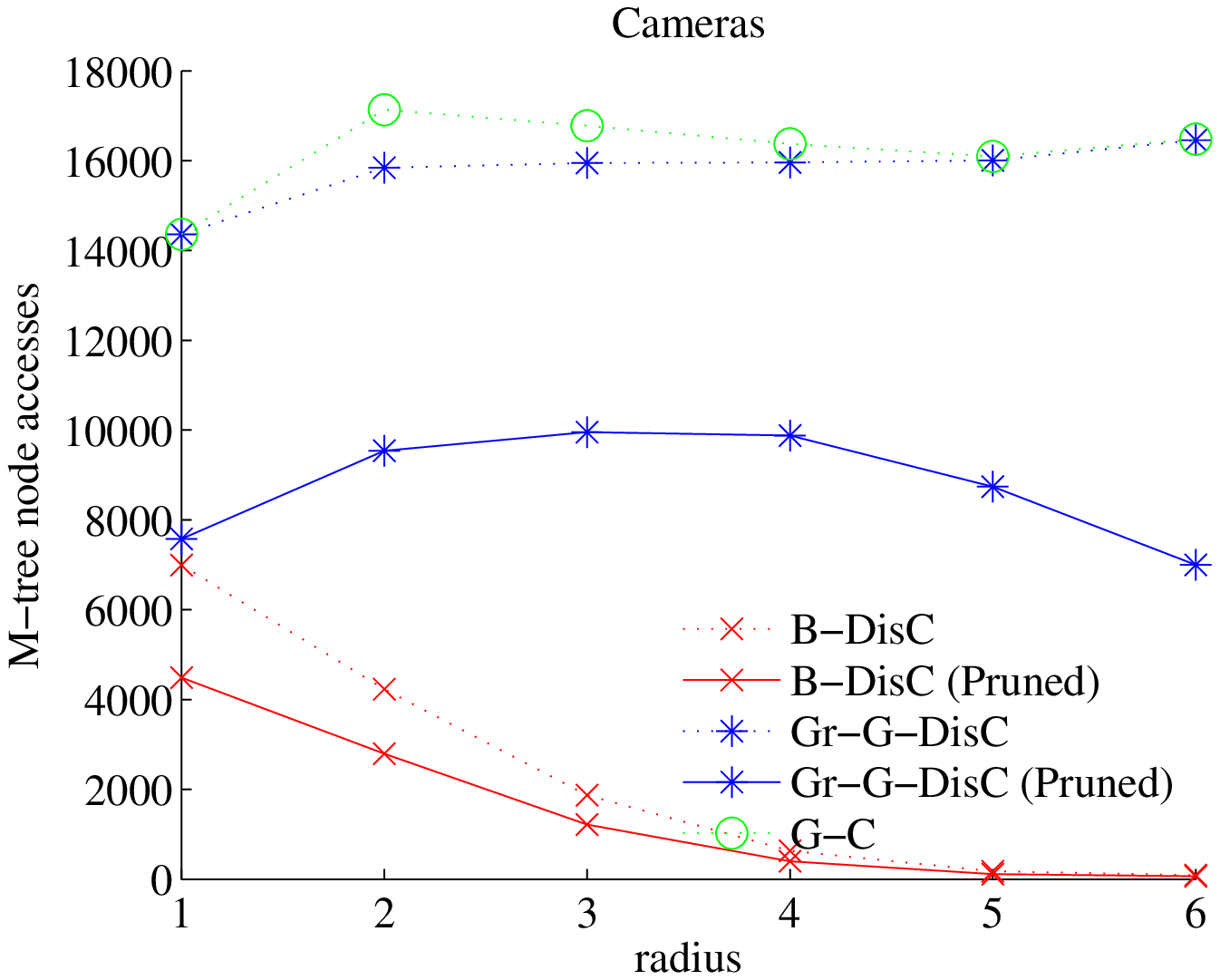}
	\label{fig:active-pruning1-cameras}
	}	
	\vspace{-0.2cm}
	\caption{Node accesses for \texttt{Basic-DisC}, \texttt{Grey-Greedy-DisC} and \texttt{Greedy-DS} with and without pruning.}
	\label{fig:active-pruning1}
	\vspace{-0.2cm}
\end{figure*}

\begin{figure*}
	\vspace{-0.2cm}
	\centering
	\subfloat[Uniform.]{		
		\includegraphics[width=4.0cm]{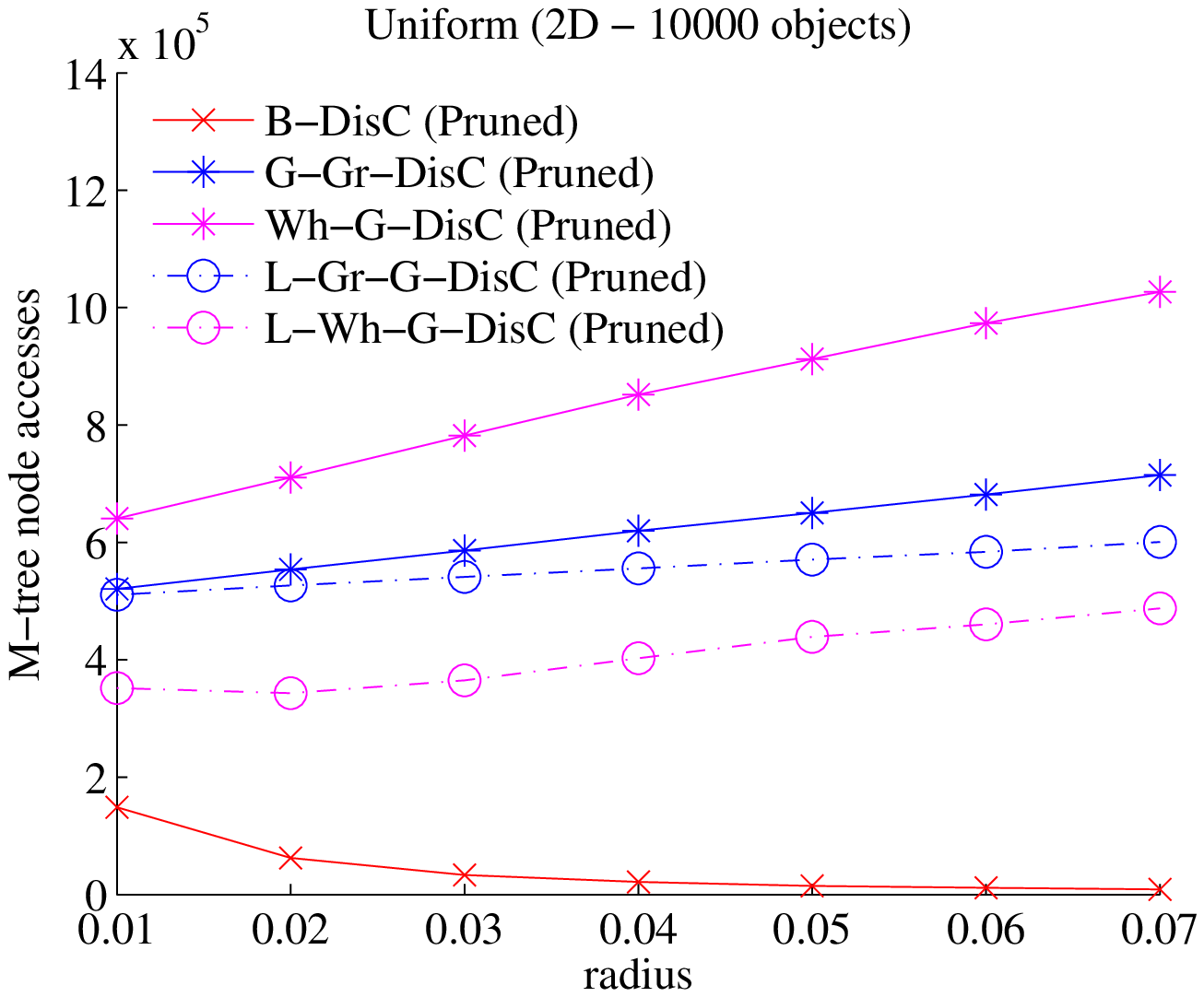}
	\label{fig:active-pruning2-uniform}
	}
	\subfloat[Clustered.]{		
		\includegraphics[width=4.0cm]{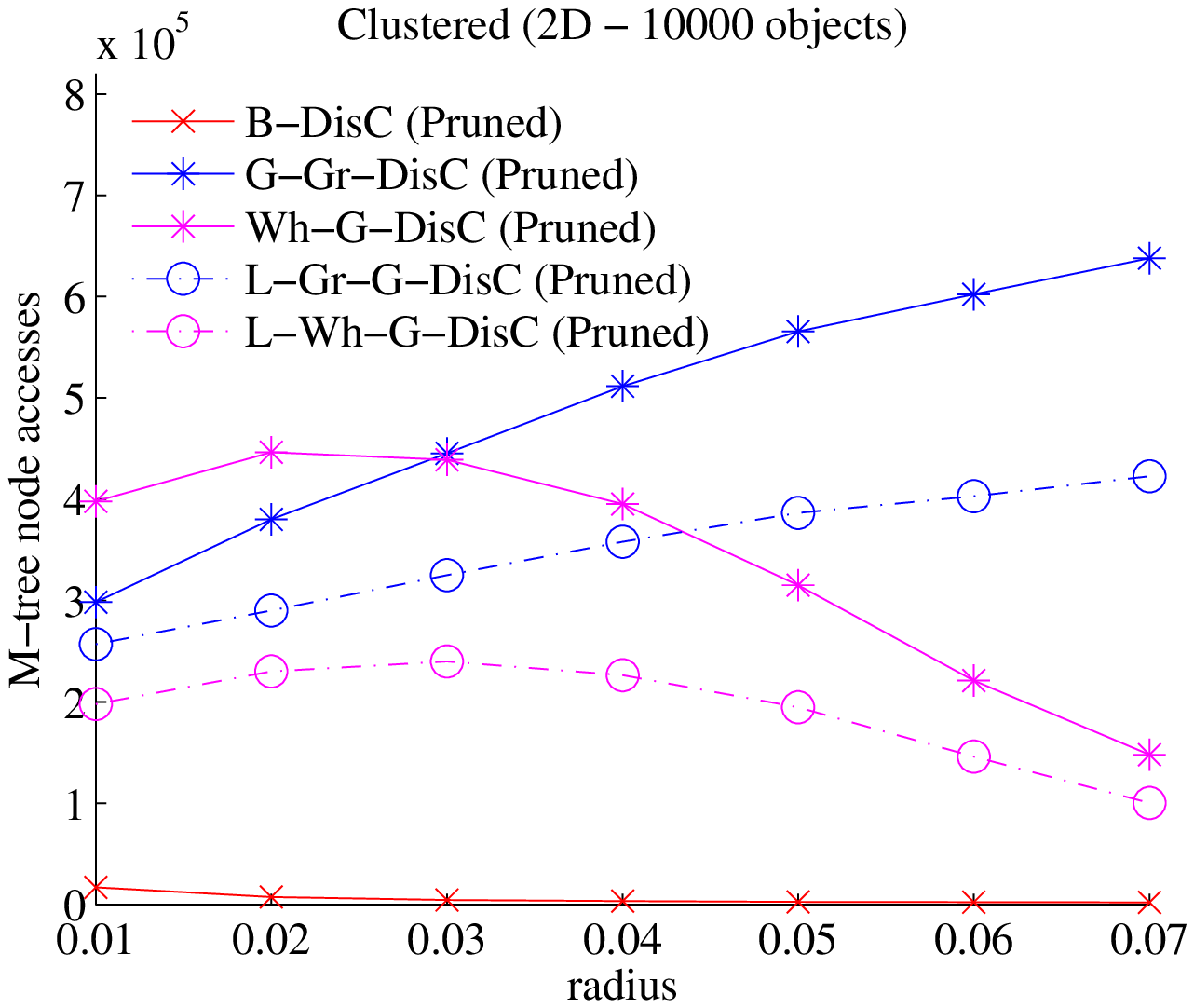}
	\label{fig:active-pruning2-clustered}
	}
	\subfloat[Cities.]{		
		\includegraphics[width=4.0cm]{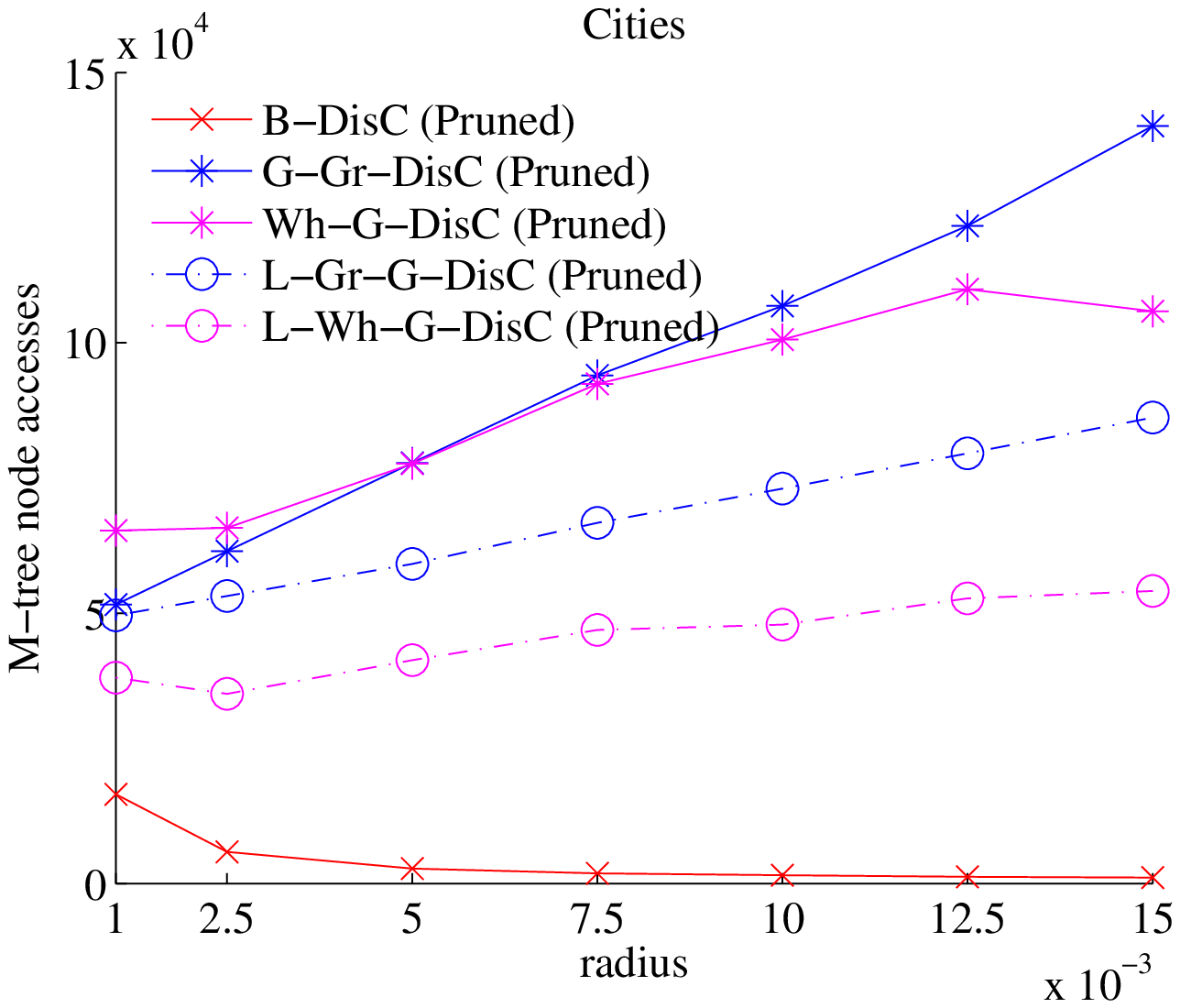}
	\label{fig:active-pruning2-cities}
	}	
	\subfloat[Cameras.]{		
		\includegraphics[width=4.0cm]{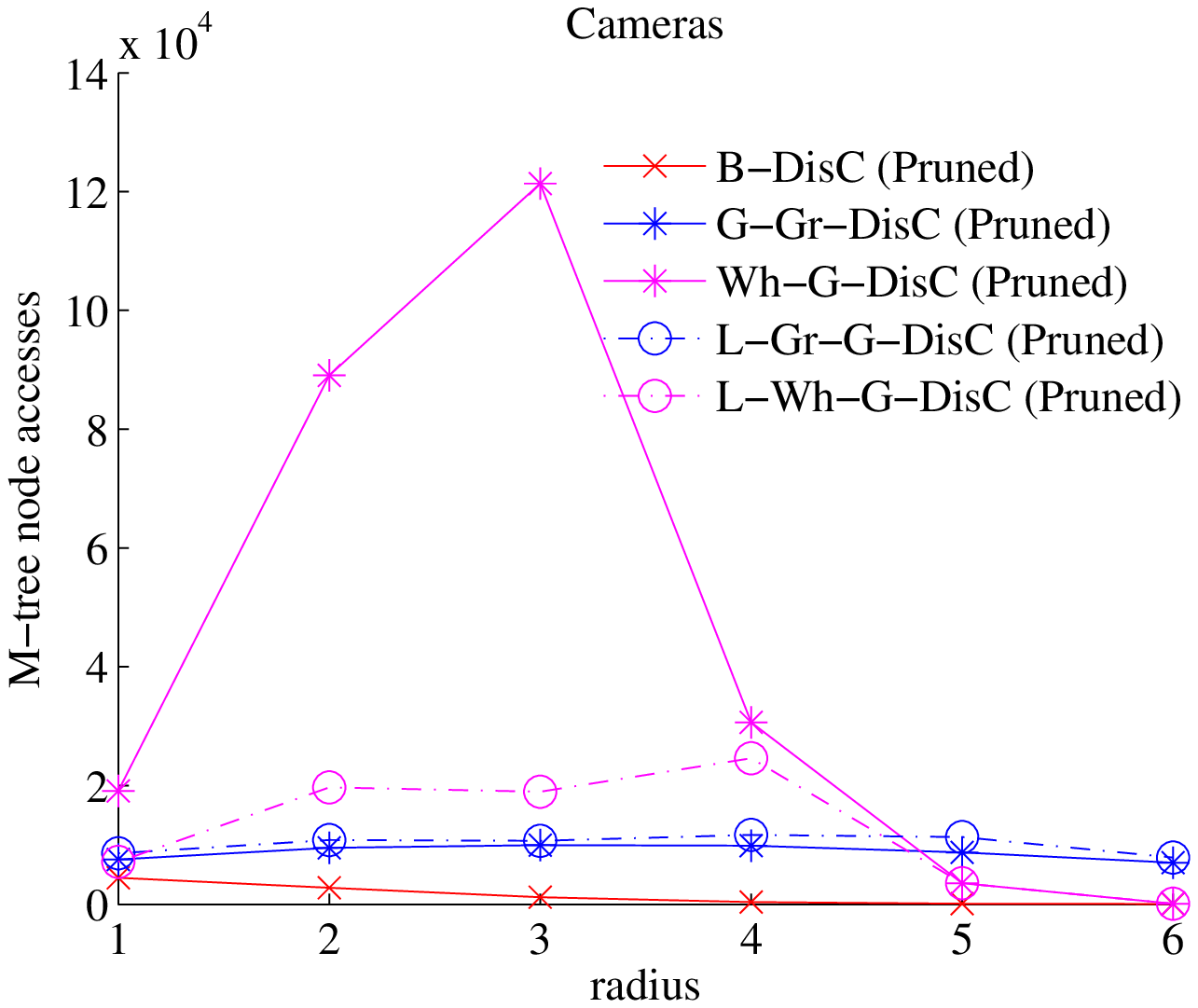}
	\label{fig:active-pruning2-cameras}
	}	
	\vspace{-0.2cm}
	\caption{Node accesses for for \texttt{Basic-DisC} and all variations of \texttt{Greedy-DisC} with pruning.}
	\label{fig:active-pruning2}
	\vspace{-0.2cm}
\end{figure*}

\begin{figure*}
	\vspace{-0.3cm}
	\centering
	\subfloat[Solution size.]{		
		\includegraphics[width=3.9cm]{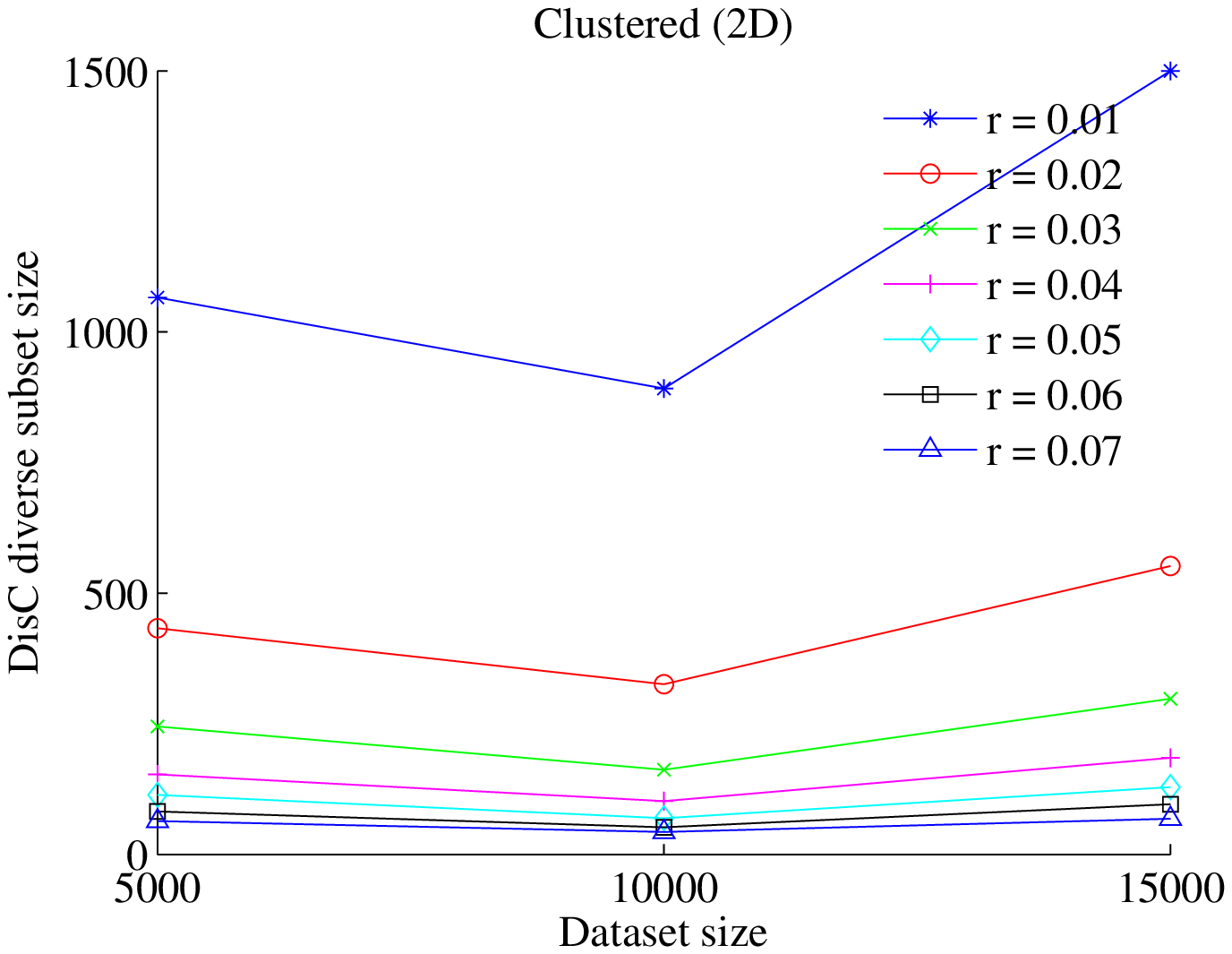}
	\label{fig:datasetsize-size}
	}
	\subfloat[Node accesses.]{		
		\includegraphics[width=3.9cm]{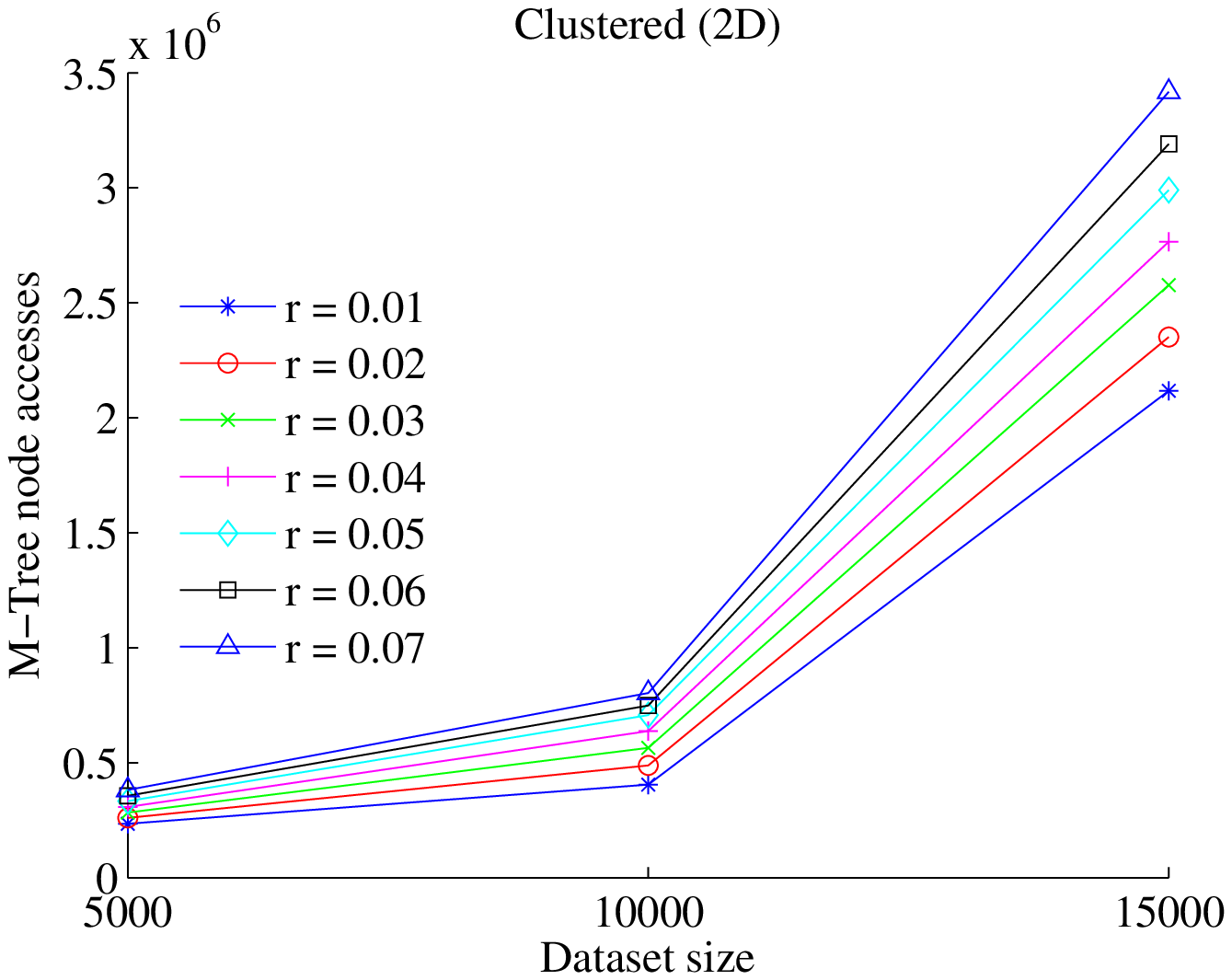}
	\label{fig:datasetsize-accesses}
	}
	\subfloat[Solution size.]{		
		\includegraphics[width=3.9cm]{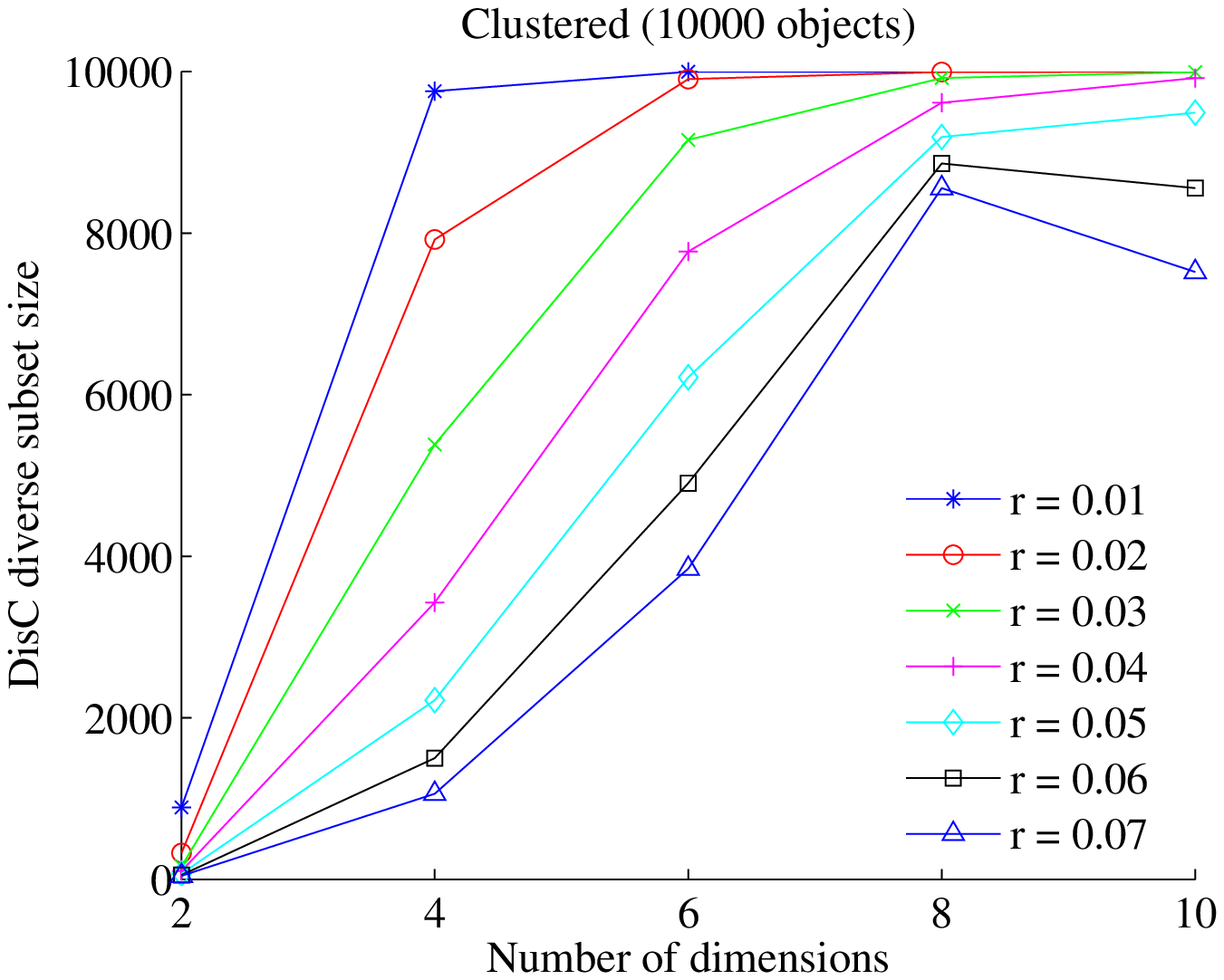}
	\label{fig:dimensionsize-size}
	}
	\subfloat[Node accesses.]{		
		\includegraphics[width=3.9cm]{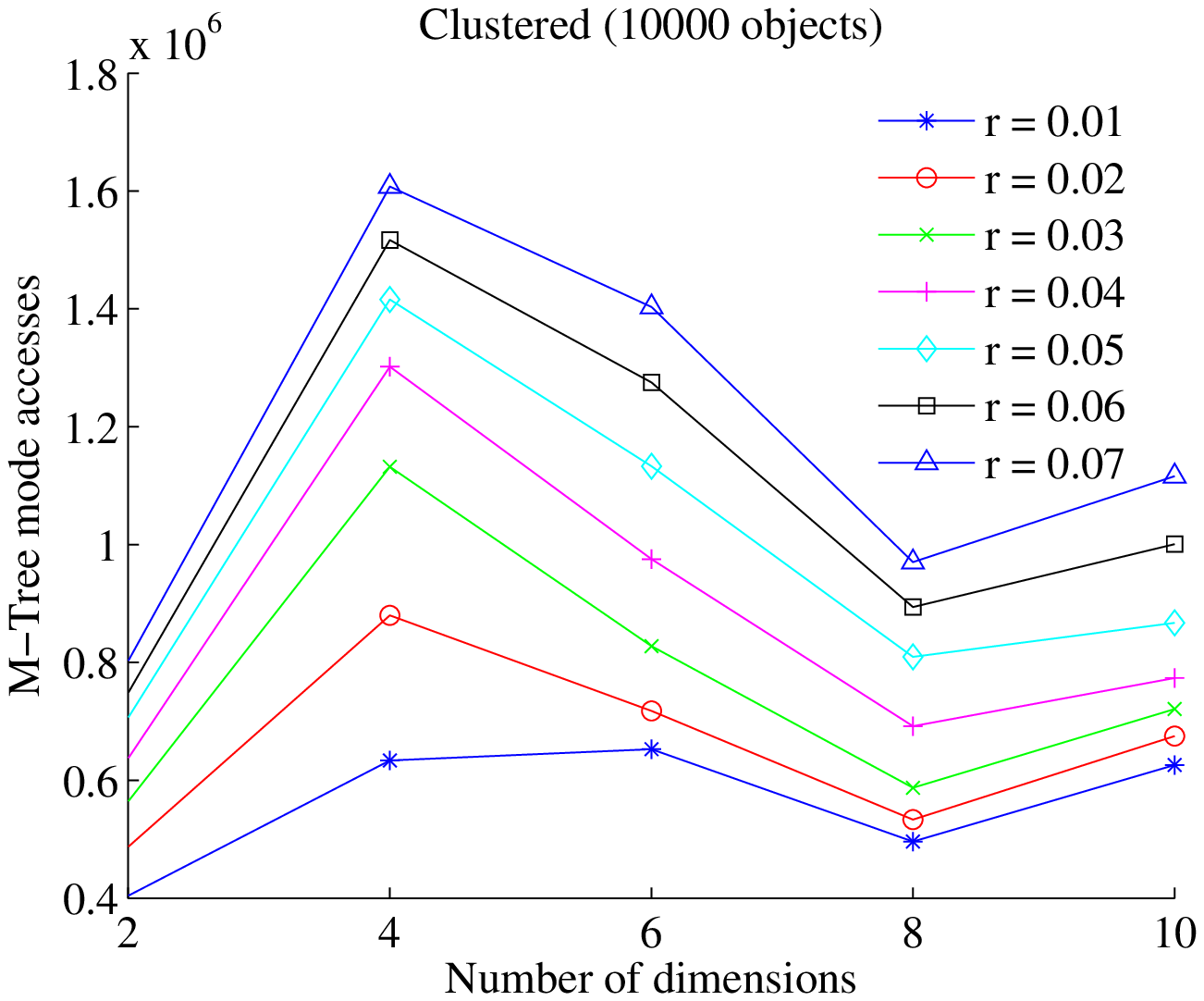}
	\label{fig:dimensionsize-accesses}
	}
	\vspace{-0.3cm}
	\caption{Varying (a)-(b) cardinality and (c)-(d) dimensionality.}
	\label{fig:dimensioncarinality}
	\vspace{-0.3cm}
\end{figure*}

\vspace{0.2cm}
\noindent\textbf{Solution Size and Computational Cost:} 
We first compare our various heuristics in terms of the size of the computed diverse subset and the cost of its computation. The computational cost is measured in terms of the required accesses to the nodes of the M-tree which indexes the objects. Table~\ref{tab:size-comparison} shows the solution size for different radii, while Figure~\ref{fig:active-pruning1} reports the computational cost for \texttt{Basic-DisC}, \texttt{Grey-Greedy-DisC} and \texttt{Greedy-C} (or \texttt{B-DisC}, \texttt{Gr-G-DisC} and \texttt{G-C} respectively) and also shows the savings in accesses when the pruning rule of Section~\ref{sec:Indexing} is employed for \texttt{Basic-DisC} and \texttt{Grey-Greedy-DisC} (as previously detailed, this pruning cannot be applied to \texttt{Greedy-C}).

\texttt{Grey-Greedy-DisC} locates a smaller DisC diverse subset than \texttt{Basic-DisC} in all cases. This, however, has the trade-off of increased computational cost. The additional computational cost becomes more significant as the radius increases. The reason for this is that \texttt{Grey-Greedy-DisC} performs significantly more range queries than \texttt{Basic-DisC}. As the radius increases, objects have more neighbors and, thus, more M-tree nodes need to be accessed in order to retrieve them. On the contrary, the cost of \texttt{Basic-DisC} is reduced when the radius increases. This happens because this heuristic performs a single pass of the leaves of the M-tree. For larger radii, more objects are colored grey by each selected (black) object and, therefore, less range queries are performed. \texttt{Greedy-C} has similar behavior with \texttt{Grey-Greedy-DisC} in terms of solution size. This means that raising the independence assumption does not always lead to smaller diverse subsets as one might expect. Note that, the computed diverse subsets by all heuristics for the ``Clustered'' dataset are smaller than for the ``Uniform'' one, since objects are generally more similar to each other. Both heuristics benefit from pruning (up to 50\% for small radii). We also experimented with employing bottom-up rather than top-down range queries. At most cases, the benefit in node accesses was less than 5\%. The \texttt{Fast-C} heuristic described in Section~\ref{sec:Indexing} required up to 30\% less node accesses than \texttt{Greedy-C}, while computing similar sized solutions. However, the solutions had a larger percentage of independent objects than those of \texttt{Greedy-C}. We omit the relative figures due to space limitations.

Figure~\ref{fig:active-pruning2} compares \texttt{Grey-Greedy-DisC} with \texttt{White-Greedy-} \texttt{DisC} and also two ``Lazy'' variations of them, where we perform range queries with a smaller radius than $r$ and $2r$ for each grey or white object respectively; instead, we use the values $r/2$ and $3r/2$, aiming at reducing the computational cost by not updating the white neighborhoods of objects that are further away from the newly colored objects. We call these two variations \texttt{Lazy-Grey-Greedy-DisC} and \texttt{Lazy-White-Greedy} \texttt{-DisC} (or \texttt{L-Gr-G-DisC} and \texttt{L-Wh-G-DisC} respectively). We see that \texttt{White-Greedy-DisC} performs very well for the clustered dataset as $r$ increases, since in that case objects are generally closer to each other and, therefore, at each iteration many neighbors of a white object may turn grey at the same time. The lazy variations can further reduce the computational cost of the heuristics with the trade-off of slightly larger solution sizes (Table~\ref{tab:size-comparison}).

In the rest of this section, unless otherwise noted, we use the \texttt{(Grey-)Greedy-DisC (Pruned)} heuristic.

\vspace{0.2cm}
\noindent\textbf{Impact of Dataset Cardinality and Dimensionality:} 
For this experiment, we employ the ``Clustered'' dataset and vary its cardinality from 5000 to 15000 objects and its dimensionality from 2 to 10 dimensions. Figure~\ref{fig:dimensioncarinality} shows the corresponding solution size and computational cost as computed by the \texttt{Greedy-DisC} heuristic. We observe that the solution size is more sensitive to changes in cardinality when the radius is small. The reason for this is that for large radii, a selected object covers a large area in space. Therefore, even when the cardinality increases and there are many available objects to choose from, these objects are quickly covered by the selected ones. In Figure~\ref{fig:dimensioncarinality}\subref{fig:datasetsize-accesses}, the increase in the computational cost is due to the increase of range queries required to maintain correct information about the size of the white neighborhoods.

Increasing the dimensionality of the dataset causes more objects to be selected as diverse as shown in Figure~\ref{fig:dimensioncarinality}\subref{fig:dimensionsize-size}. This is due to the ``curse of dimensionality'' effect, since space becomes sparser at higher dimensions. The computational cost may however vary as dimensionality increases, since it is influenced by the cost of computing the neighborhood size of the objects that are colored grey.

\vspace{0.2cm}
\noindent\textbf{Impact of M-tree Characteristics:} 
Next, we evaluate how the characteristics of the employed M-trees affect the computational cost of computed DisC diverse subsets. Note that, different tree characteristics do not have an impact on which objects are selected as diverse.

Different degree of overlap among the nodes of an M-tree may affect its efficiency for executing range queries. To quantify such overlap, we employ the \emph{fat-factor} \cite{DBLP:journals/tkde/TrainaTFS02} of the tree defined as:
\vspace{-0.3cm}
\[
f(T) = \frac{Z - nh}{n} \cdot \frac{1}{m - h}
\vspace{-0.2cm}
\]
where $Z$ denotes the total number of node accesses required to answer point queries for all objects stored in the tree, $n$ the number of these objects, $h$ the height of the tree and $m$ the number of nodes in the tree. Ideally, the tree would require accessing one node per level for each point query which yields a fat-factor of zero. The worst tree would visit all nodes for every point query and its fat-factor would be equal to one.

We created various M-trees using different splitting policies which result in different fat-factors. We present results for four different policies. The lowest fat-factor was acquired by employing the ``MinOverlap'' policy. Selecting as new pivots the two objects with the greatest distance from each other resulted in increased fat-factor. Even higher fat-factors were observed when assigning an equal number of objects to each new node (instead of assigning each object to the node with the closest pivot) and, finally, selecting the new pivots randomly produced trees with the highest fat-factor among all policies.

Figure~\ref{fig:fftc} reports our results for our uniform and clustered 2-dimensional datasets with cardinality equal to 10000. For the uniform dataset, we see that a high fat-factor leads to more node accesses being performed for the \emph{same} solution. This is not the case for the clustered dataset, where objects are gathered in dense areas and thus increasing the fat-factor does not have the same impact as in the uniform case, due to pruning and locality. As the radius of the computed subset becomes very large, the solution size becomes very small, since a single object covers almost the entire dataset, this is why all lines of Figure~\ref{fig:fftc} begin to converge for $r > 0.7$. 

We also experimented with varying the capacity of the nodes of the M-tree. Trees with smaller capacity require more node accesses since more nodes need to be recovered to locate the same objects; when doubling the node capacity, the computational cost was reduced by almost 45\%.

\begin{figure}
	\vspace{-0.3cm}
	\centering
	\subfloat[Uniform.]{		
		\includegraphics[width=3.9cm]{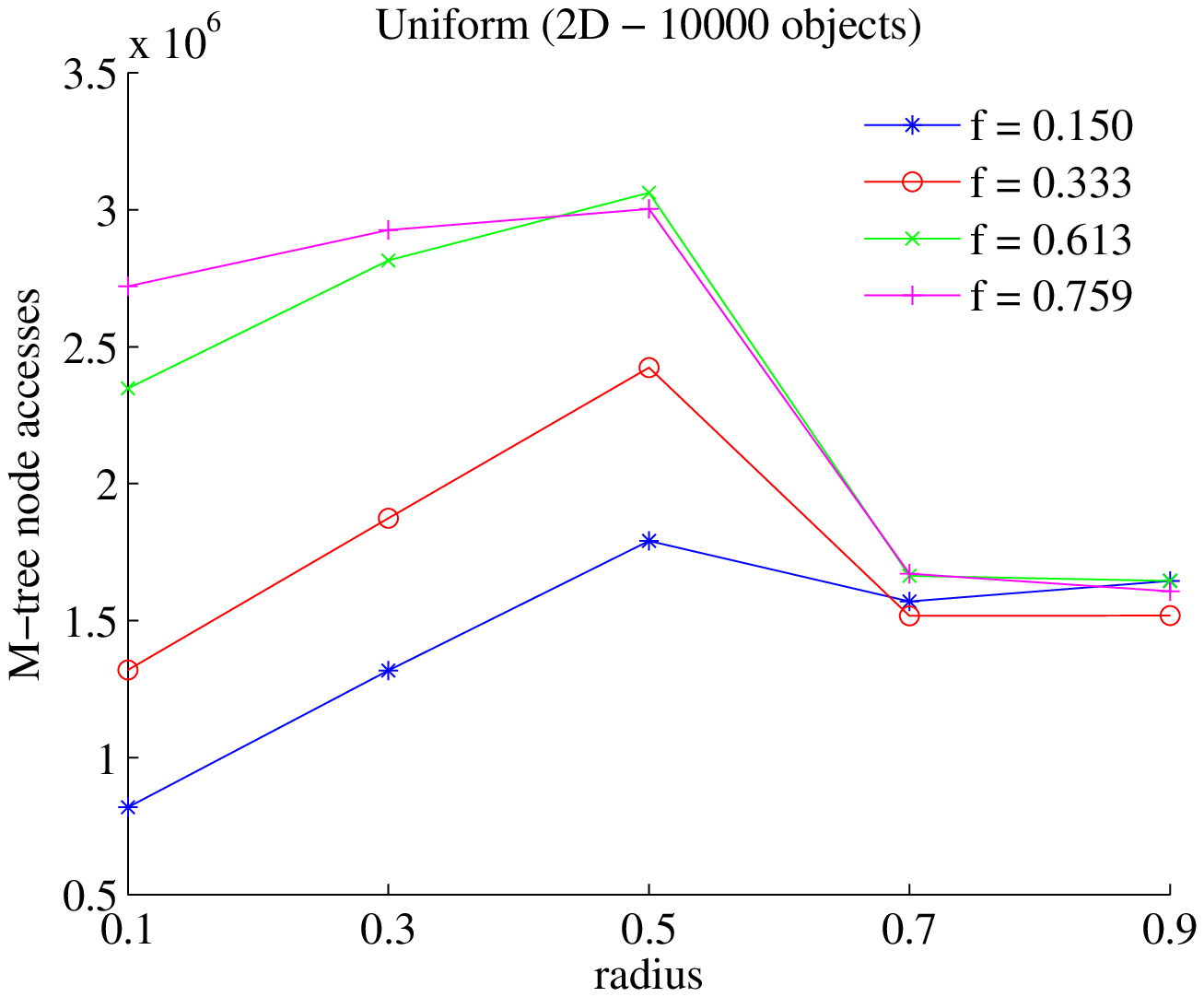}
		\label{fig:fatfactor-uniform}
	}
	\subfloat[Clustered.]{		
		\includegraphics[width=3.9cm]{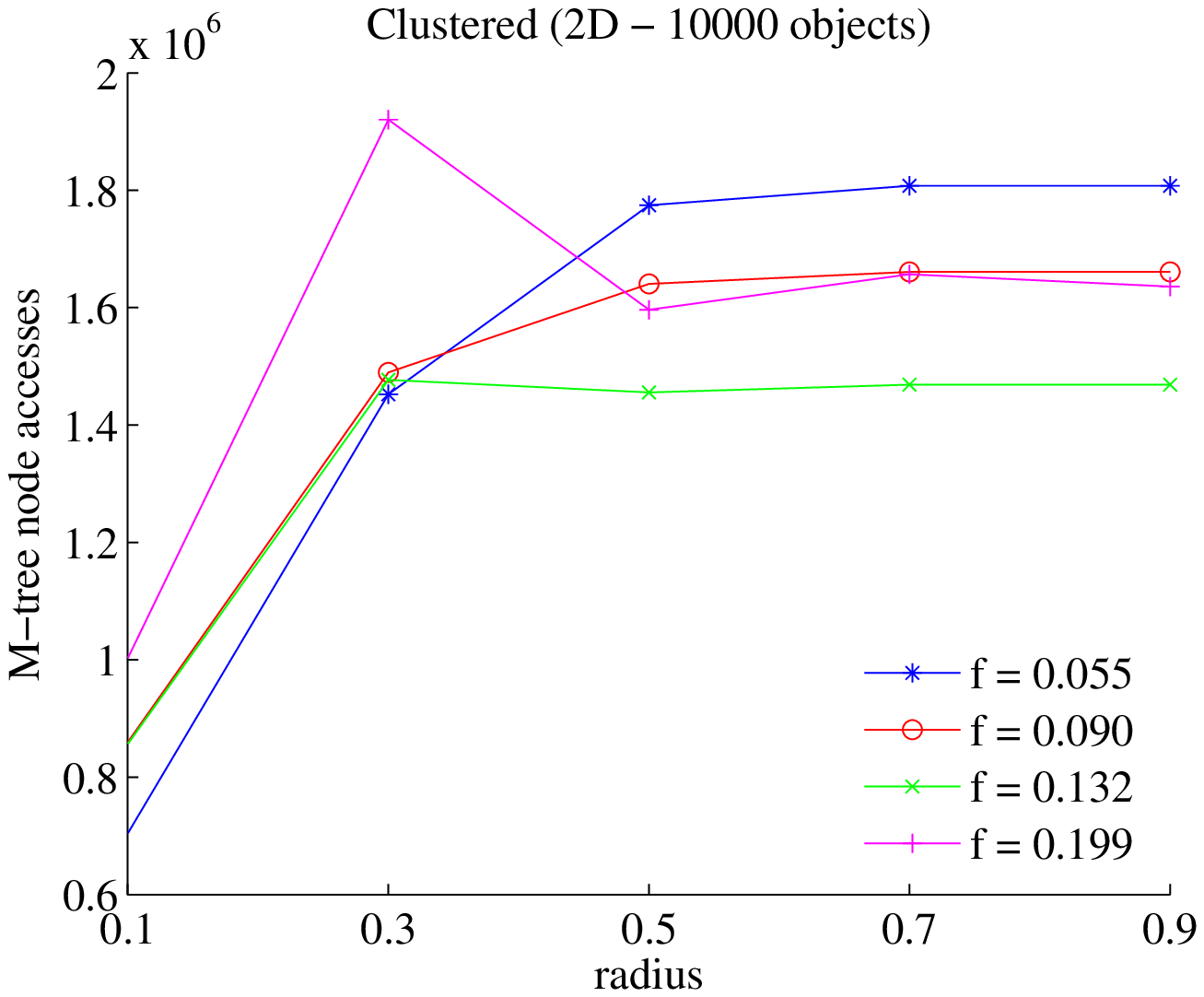}
		\label{fig:fatfactor-clustered}
	}
	\vspace{-0.3cm}
	\caption{Varying M-tree fat-factor.}
	\label{fig:fftc}
	\vspace{-0.3cm}
\end{figure}

\vspace{0.2cm}
\noindent\textbf{Zooming:} 
In the following, we evaluate the performance of our zooming algorithms. We begin with the zooming-in heuristics. To do this, we first generate solutions with \texttt{Greedy-DisC} for a specific radius $r$ and then adapt these solutions for radius $r^{\prime}$. We use \texttt{Greedy-DisC} because it gives the smallest sized solutions. We compare the results to the solutions generated from scratch by \texttt{Greedy-DisC} for the new radius. The comparison is made in terms of solution size, computational cost and also the relation of the three solutions as measured by the Jaccard distance.
Figure~\ref{fig:zoomin-size} and Figure~\ref{fig:zoomin-accesses} report the corresponding results for different radii. Due to space limitations, we report results for the ``Clustered'' and ``Cities'' datasets. Similar results are obtained for the other datasets as well. Each solution reported for the zooming-in algorithms is adapted from the \texttt{Greedy-DisC} solution for the immediately larger radius and, thus, the x-axis is reversed for clarity; e.g., the zooming solutions for $r = 0.02$ in Figure~\ref{fig:zoomin-size}\subref{fig:zoomin-size-clustered} and Figure~\ref{fig:zoomin-accesses}\subref{fig:zoomin-accesses-clustered} are adapted from the \texttt{Greedy-DisC} solution for $r = 0.03$.

We observe that the zooming-in heuristics provide similar solution sizes with \texttt{Greedy-DisC} in most cases, while their computational cost is smaller, even for \texttt{Greedy-Zoom-In}. More importantly, the Jaccard distance of the adapted solutions for $r^{\prime}$ to the \texttt{Greedy-DisC} solution for $r$ is much smaller than the corresponding distance of the \texttt{Greedy-DisC} solution for $r^{\prime}$ (Figure~\ref{fig:zoomin-jaccard}). This means that computing a new solution for $r^{\prime}$ from scratch changes most of the objects returned to the user, while a solution computed by a zooming-in heuristic maintains many common objects in the new solution. Therefore, the new diverse subset is intuitively closer to what the user expects to receive.

Figure~\ref{fig:zoomout-size} and Figure~\ref{fig:zoomout-accesses} show corresponding results for the zooming-out heuristics. The \texttt{Greedy-Zoom-Out(c)} heuristic achieves the smallest adapted DisC diverse subsets. However, its computational cost is very high and generally exceeds the cost of computing a new solution from scratch. \texttt{Greedy-Zoom-Out(a)} also achieves similar solution sizes with \texttt{Greedy-Zoom-Out(c)}, while its computational cost is much lower. The non-greedy heuristic has the lowest computational cost. Again, all the Jaccard distances of the zooming-out heuristics to the previously computed solution are smaller than that of \texttt{Greedy-DisC} (Figure~\ref{fig:zoomout-jaccard}), which indicates that a solution computed from scratch has only a few objects in common from the initial DisC diverse set.


\begin{figure}
	\vspace{-0.2cm}
	\centering
	\subfloat[Clustered.]{		
		\includegraphics[width=3.9cm]{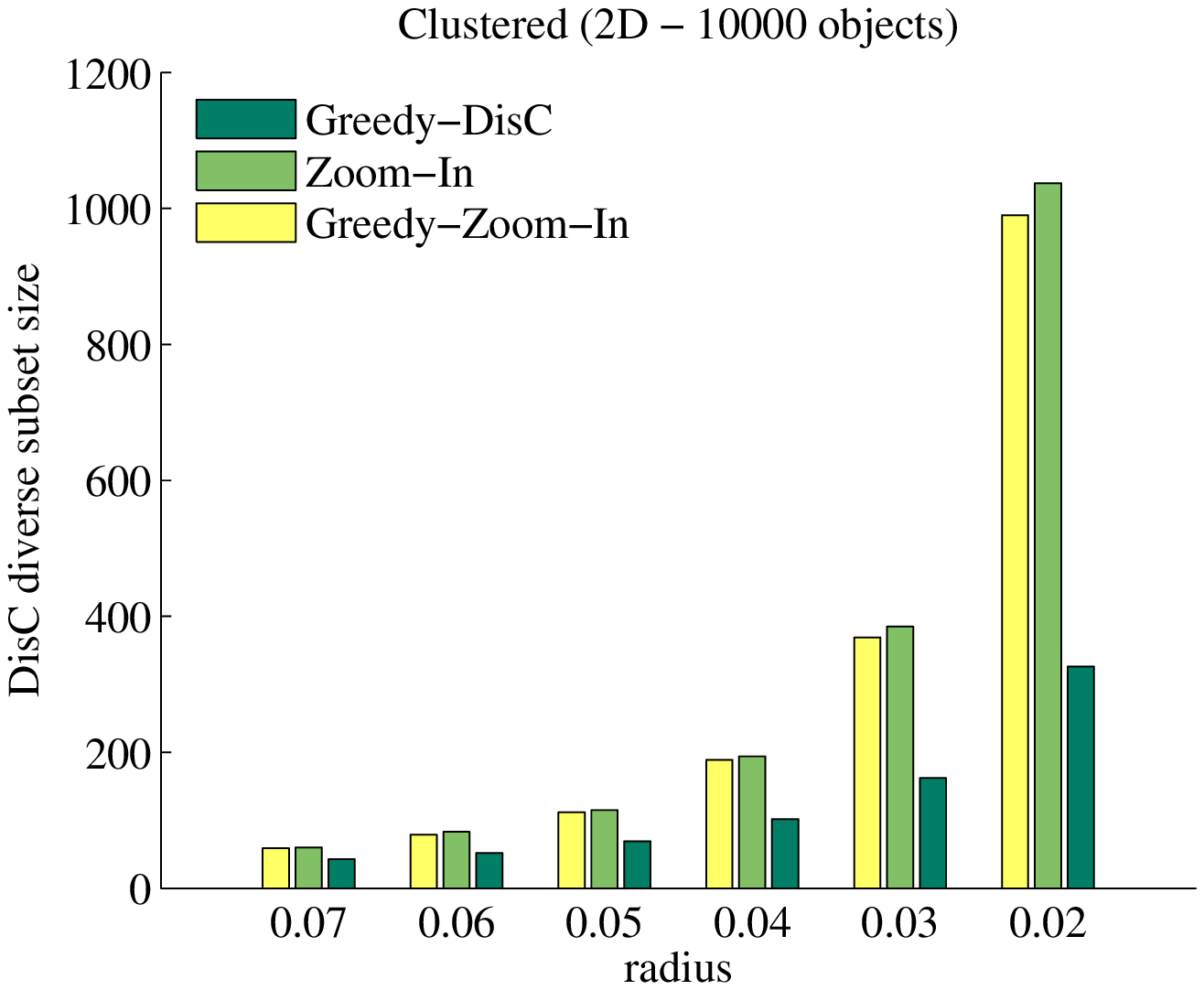}
		\label{fig:zoomin-size-clustered}
	}
	\subfloat[Cities.]{		
		\includegraphics[width=3.9cm]{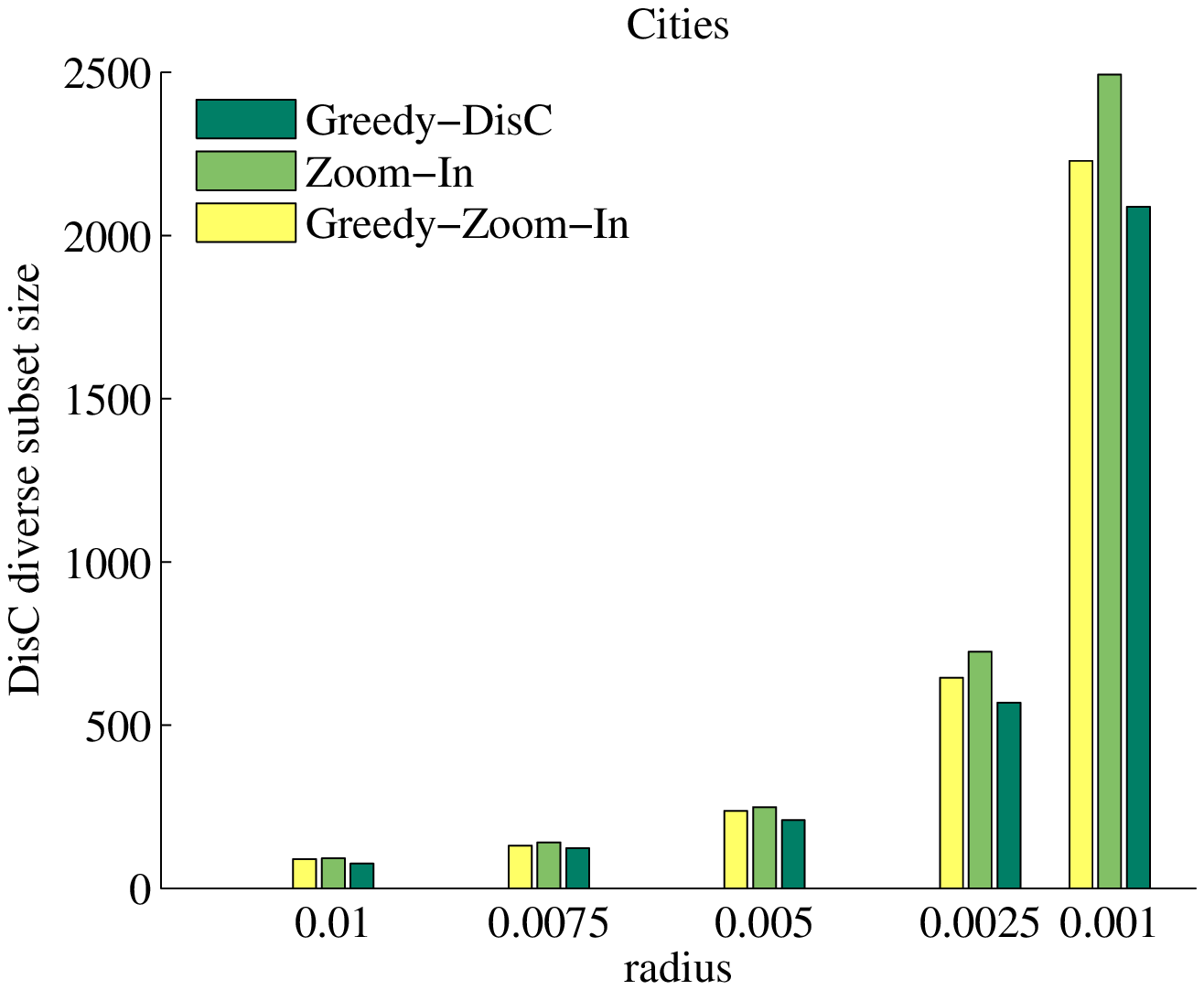}
	}
	\vspace{-0.2cm}
	\caption{Solution size for zooming-in.}
	\label{fig:zoomin-size}
	\vspace{-0.4cm}
\end{figure}

\begin{figure}
	\vspace{-0.1cm}
	\centering
	\subfloat[Clustered.]{		
		\includegraphics[width=3.9cm]{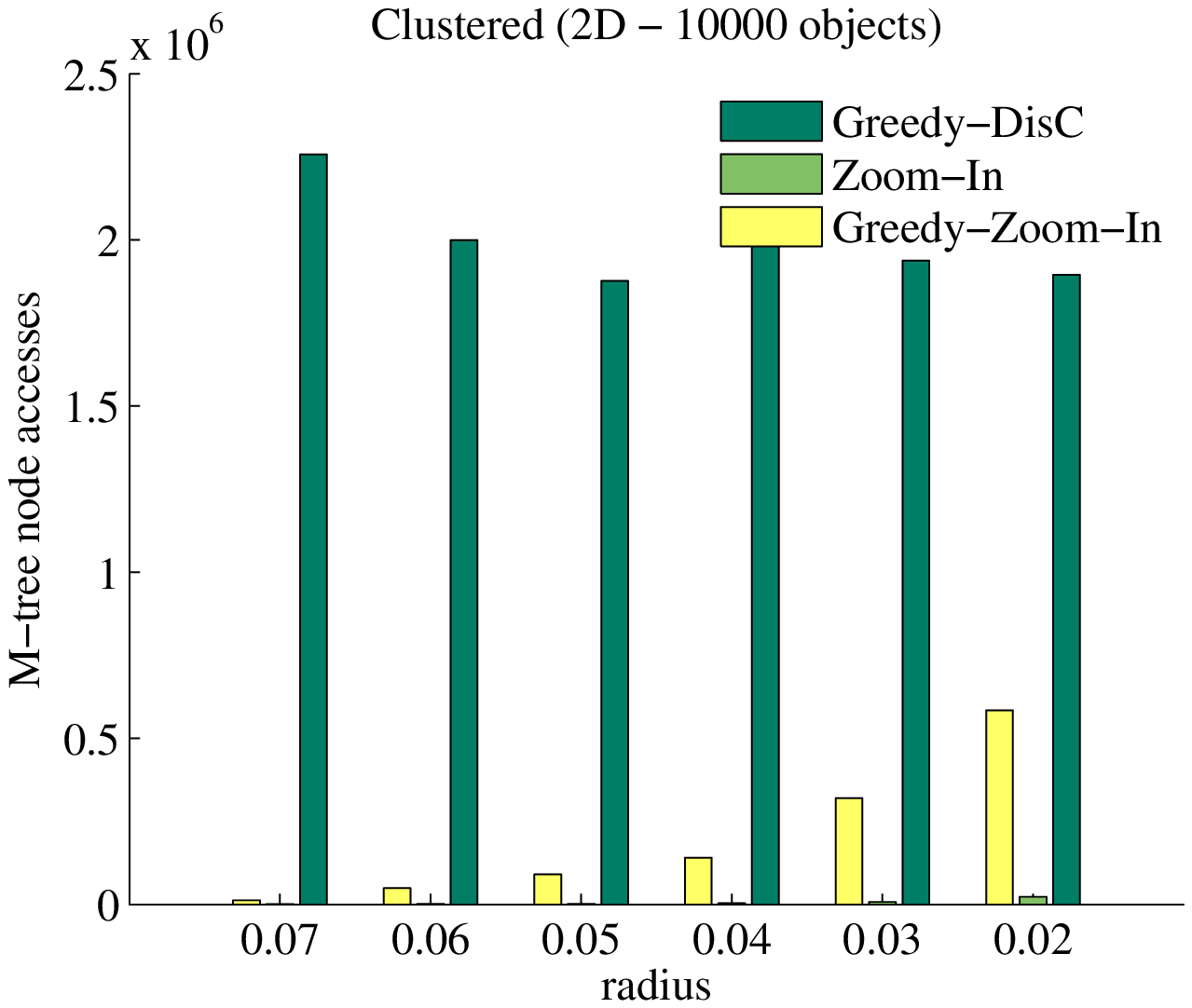}
	\label{fig:zoomin-accesses-clustered}
	}
	\subfloat[Cities.]{		
		\includegraphics[width=3.9cm]{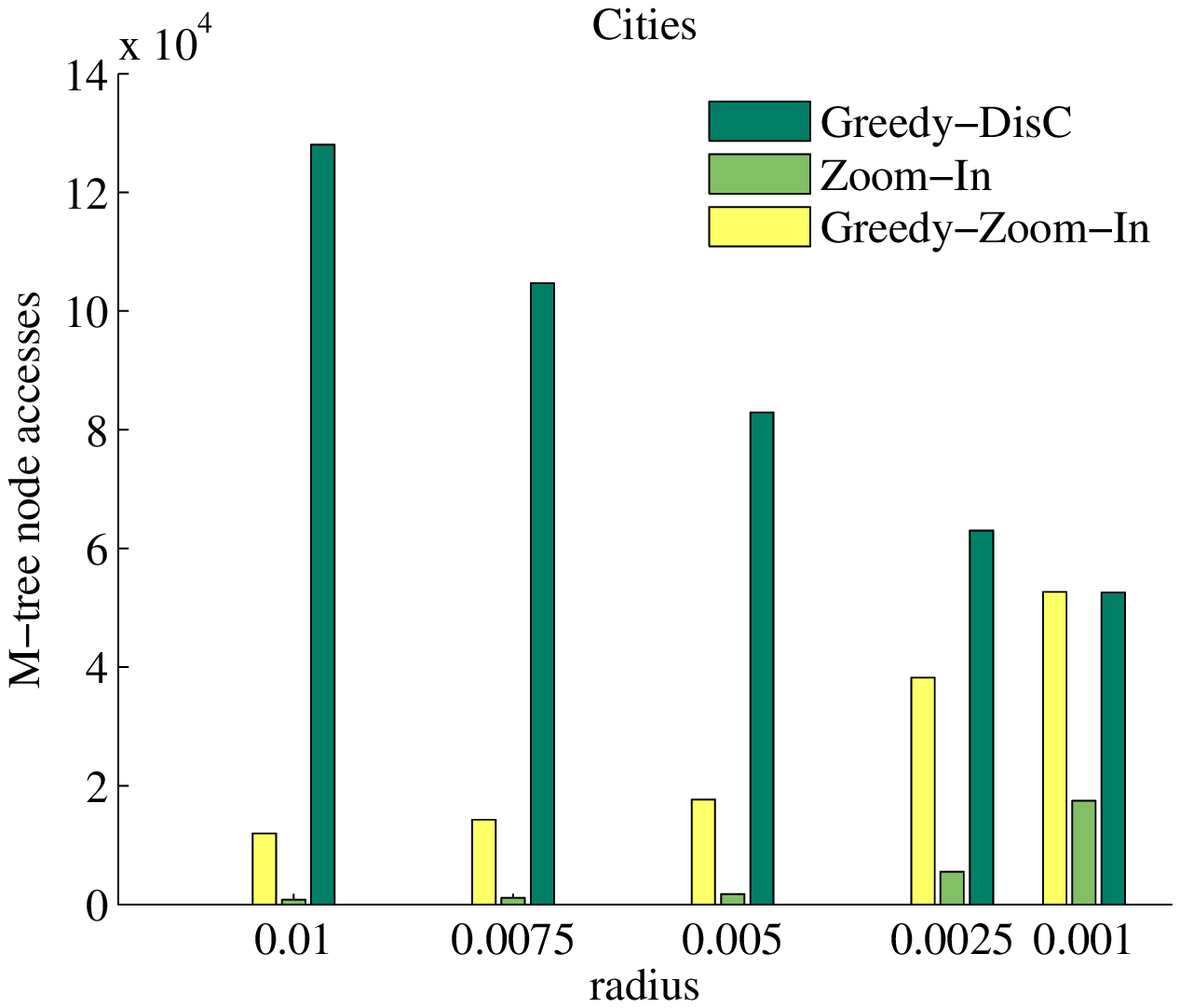}
	}
	\vspace{-0.2cm}
	\caption{Node accesses for zooming-in.}
	\label{fig:zoomin-accesses}
	\vspace{-0.4cm}
\end{figure}

\begin{figure}
	\vspace{-0.1cm}
	\centering
	\subfloat[Clustered.]{		
		\includegraphics[width=3.9cm]{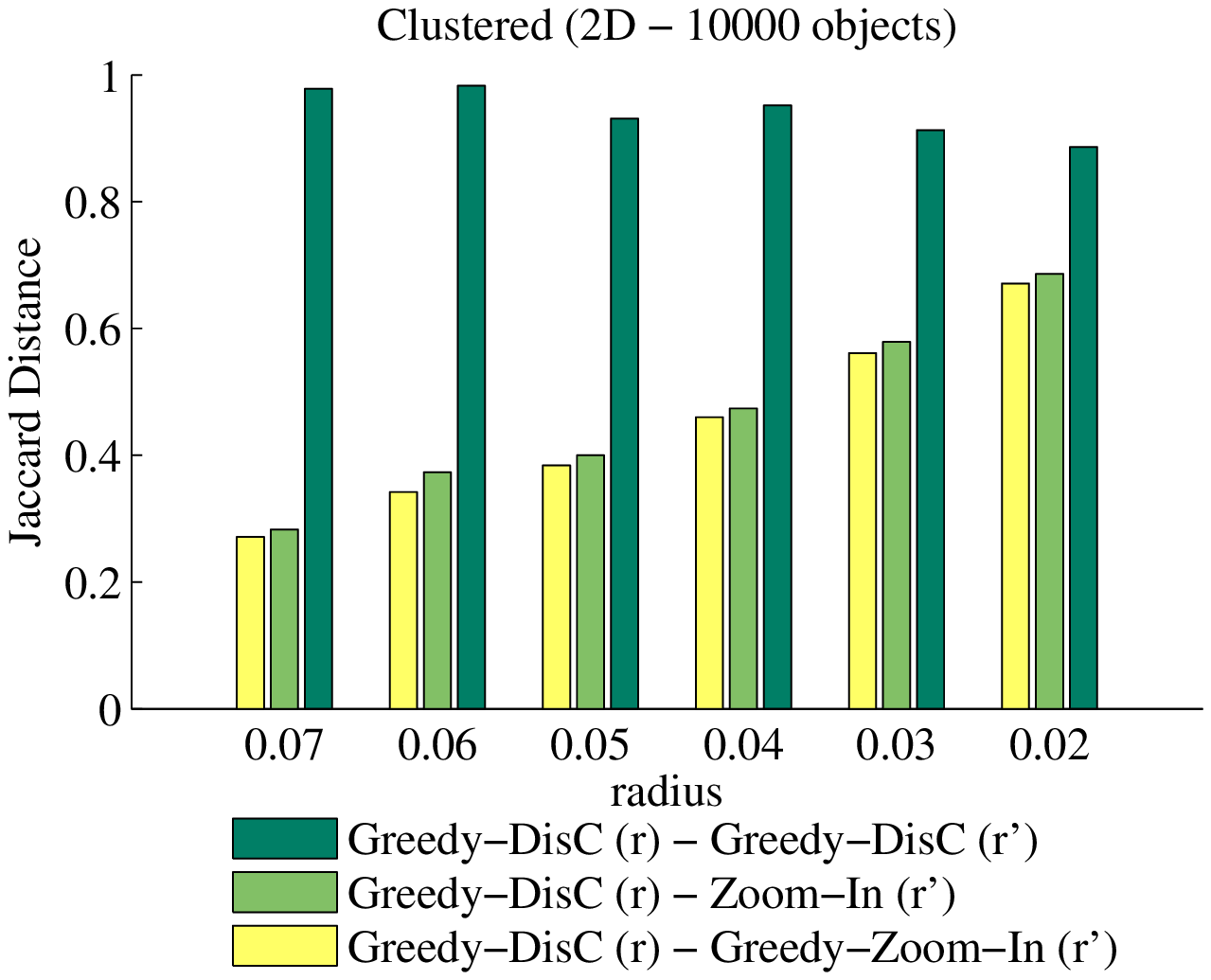}
	}
	\subfloat[Cities.]{		
		\includegraphics[width=3.9cm]{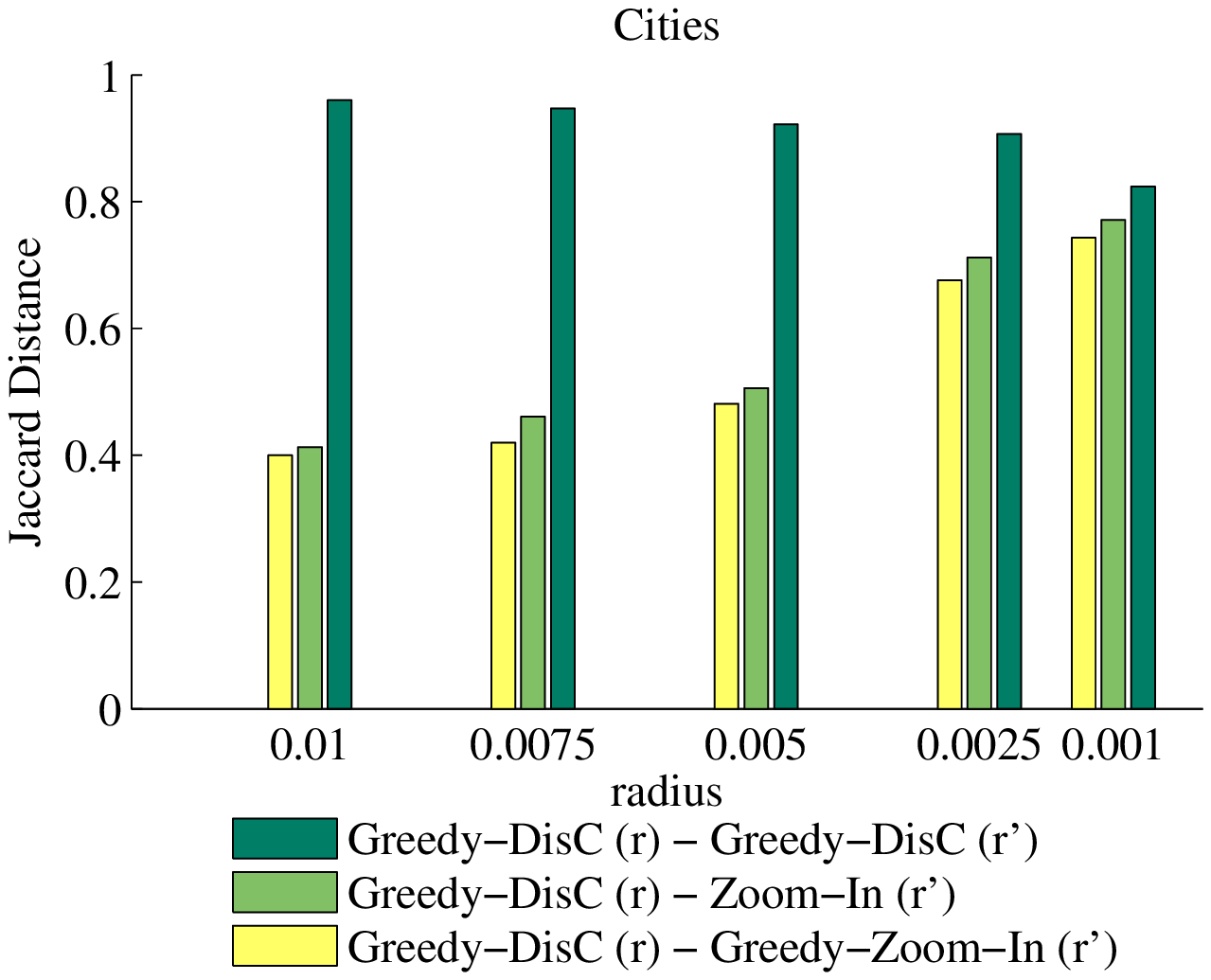}
	}
	\vspace{-0.2cm}
	\caption{Jaccard distance for zooming-in.}
	\label{fig:zoomin-jaccard}
	\vspace{-0.4cm}
\end{figure}

\begin{figure}
	\vspace{-0.1cm}
	\centering
	\subfloat[Clustered.]{		
		\includegraphics[width=3.9cm]{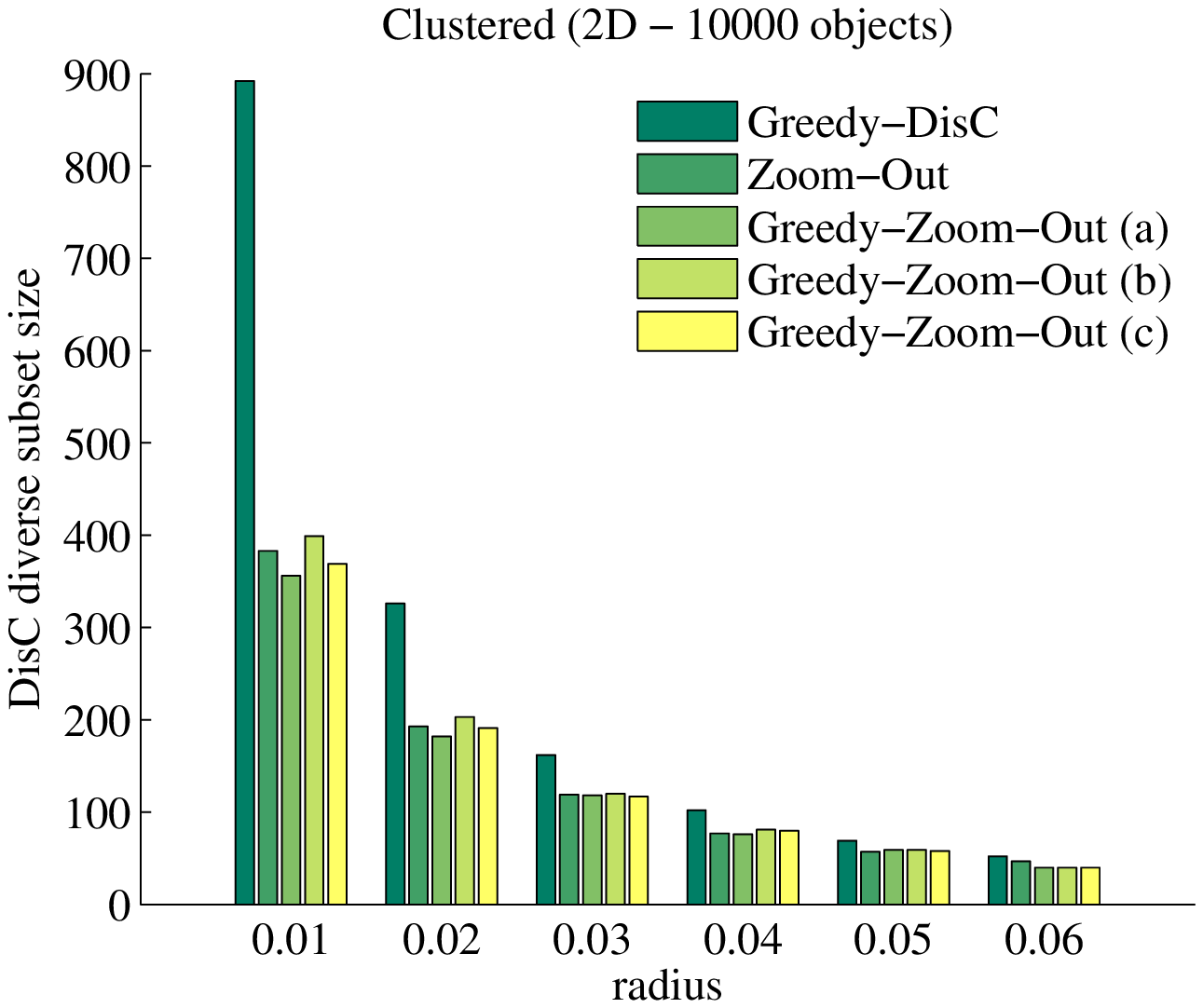}
		\label{fig:zoomout-size-clustered}
	}
	\subfloat[Cities.]{		
		\includegraphics[width=3.9cm]{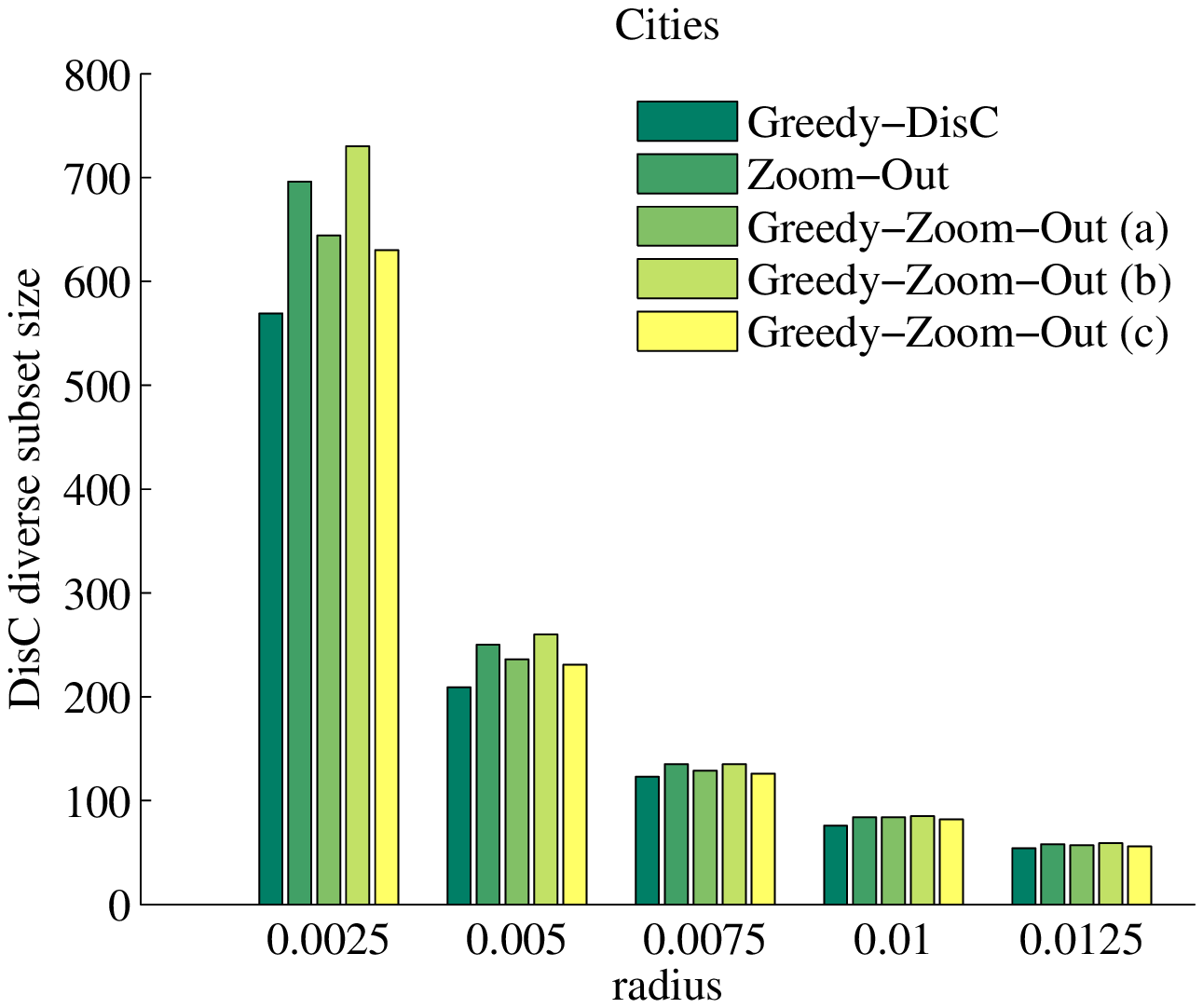}
	}
	\vspace{-0.2cm}
	\caption{Solution size for zooming-out.}
	\label{fig:zoomout-size}
	\vspace{-0.4cm}
\end{figure}

\begin{figure}
	\vspace{-0.1cm}
	\centering
	\subfloat[Clustered.]{		
		\includegraphics[width=3.9cm]{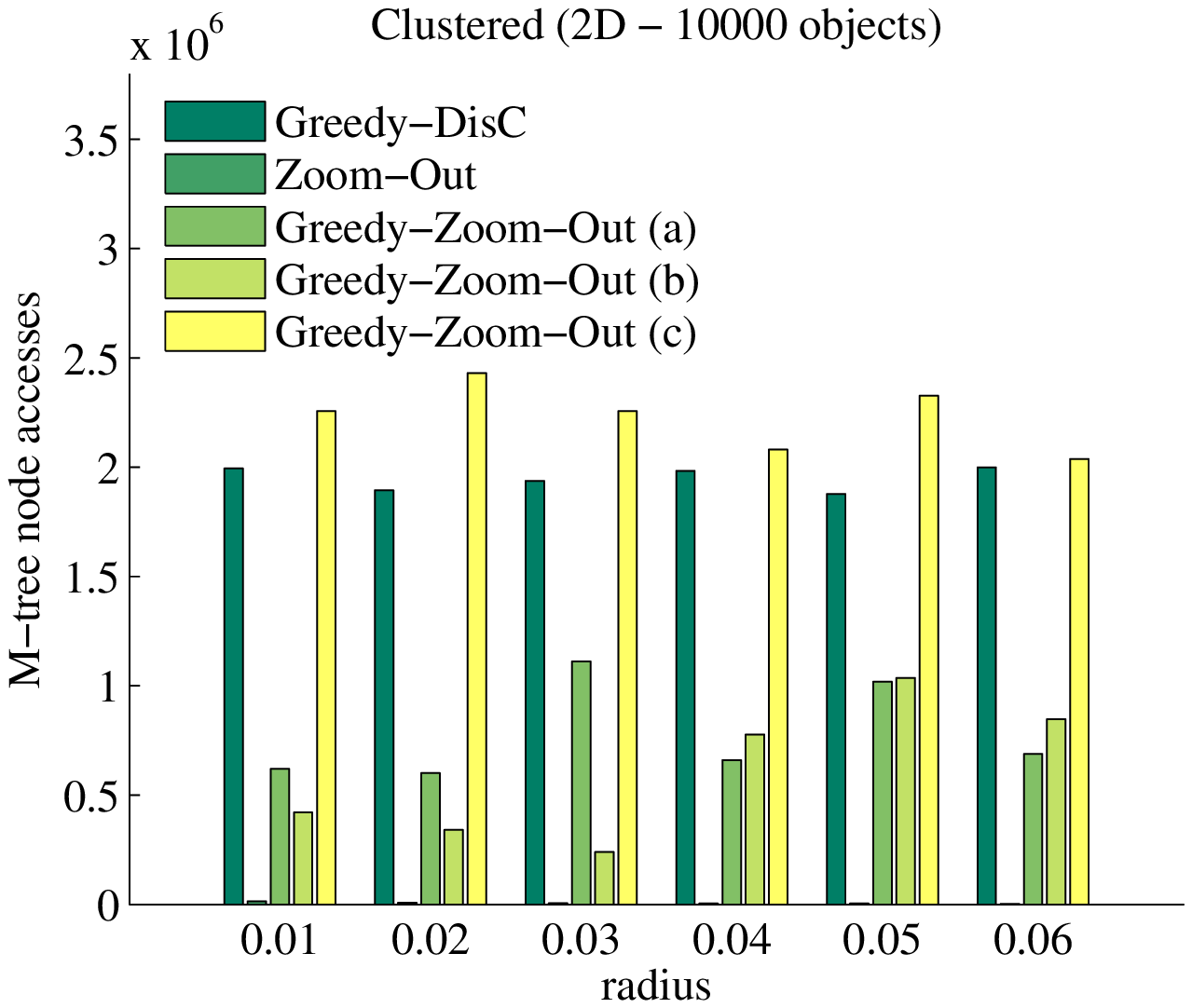}
	\label{fig:zoomout-accesses-clustered}
	}
	\subfloat[Cities.]{		
		\includegraphics[width=3.9cm]{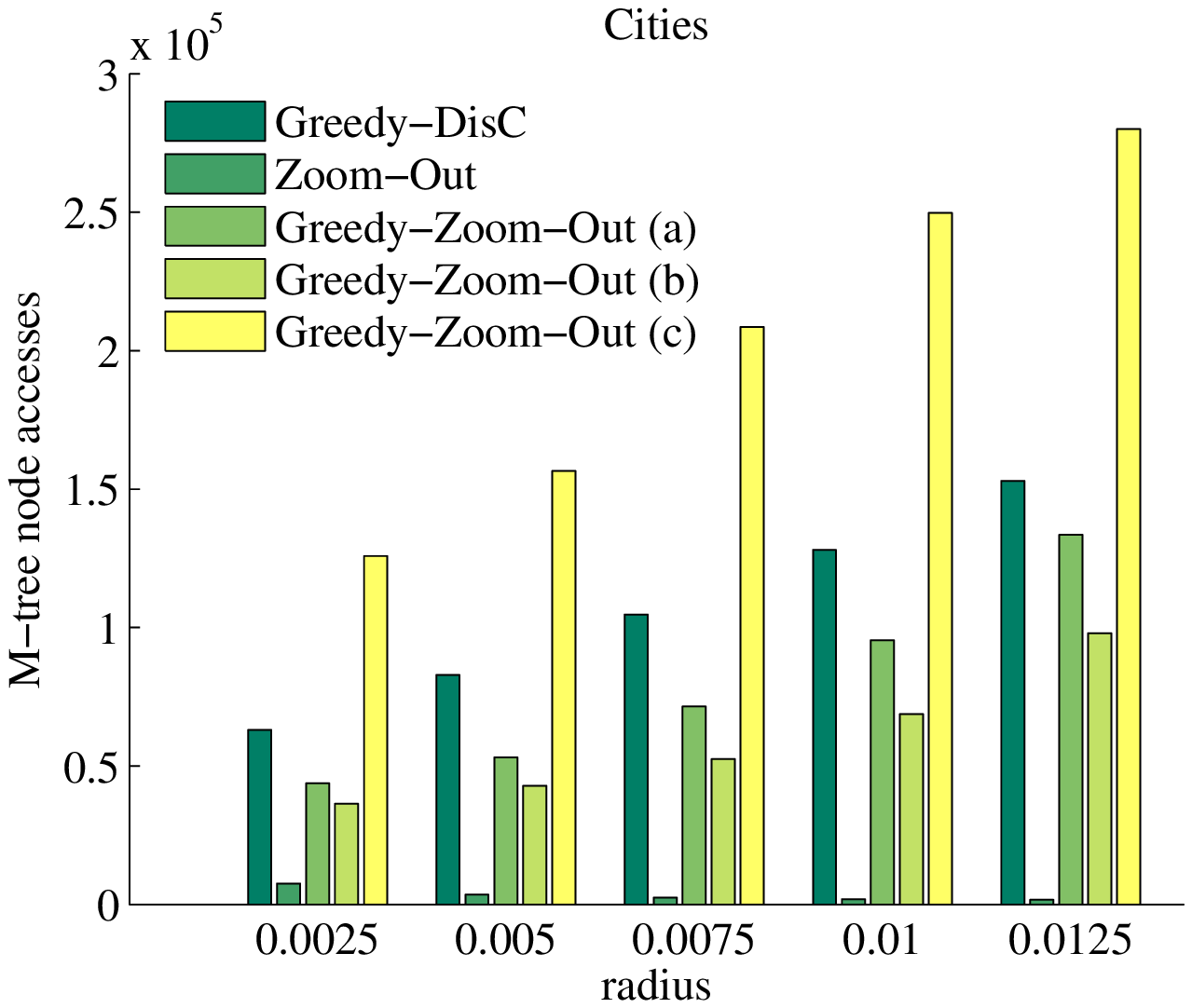}
	}
	\vspace{-0.2cm}
	\caption{Node accesses for zooming-out.}
	\label{fig:zoomout-accesses}
	\vspace{-0.4cm}
\end{figure}

\begin{figure}
	\vspace{-0.1cm}
	\centering
	\subfloat[Clustered.]{		
		\includegraphics[width=3.9cm]{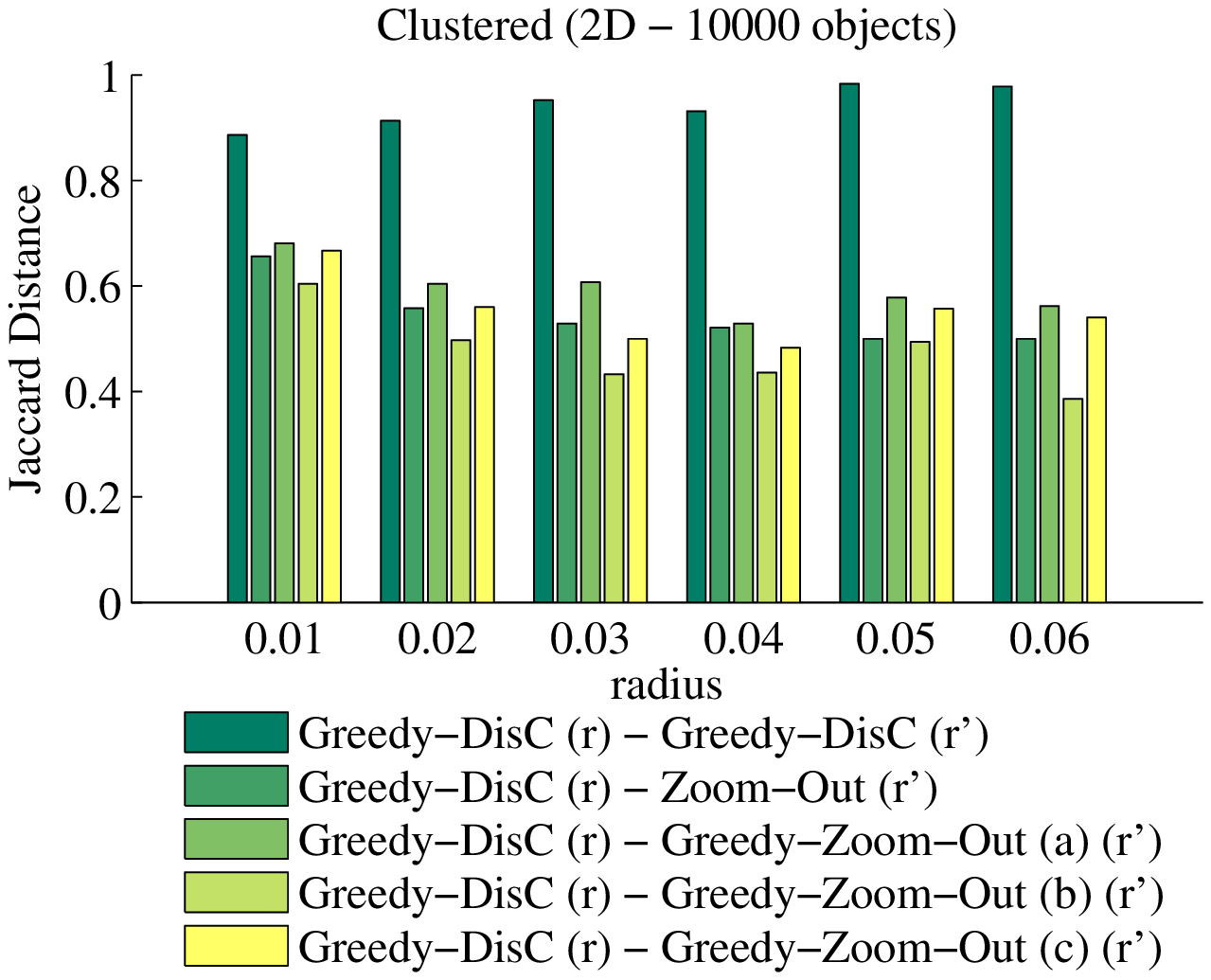}
	}
	\subfloat[Cities.]{		
		\includegraphics[width=3.9cm]{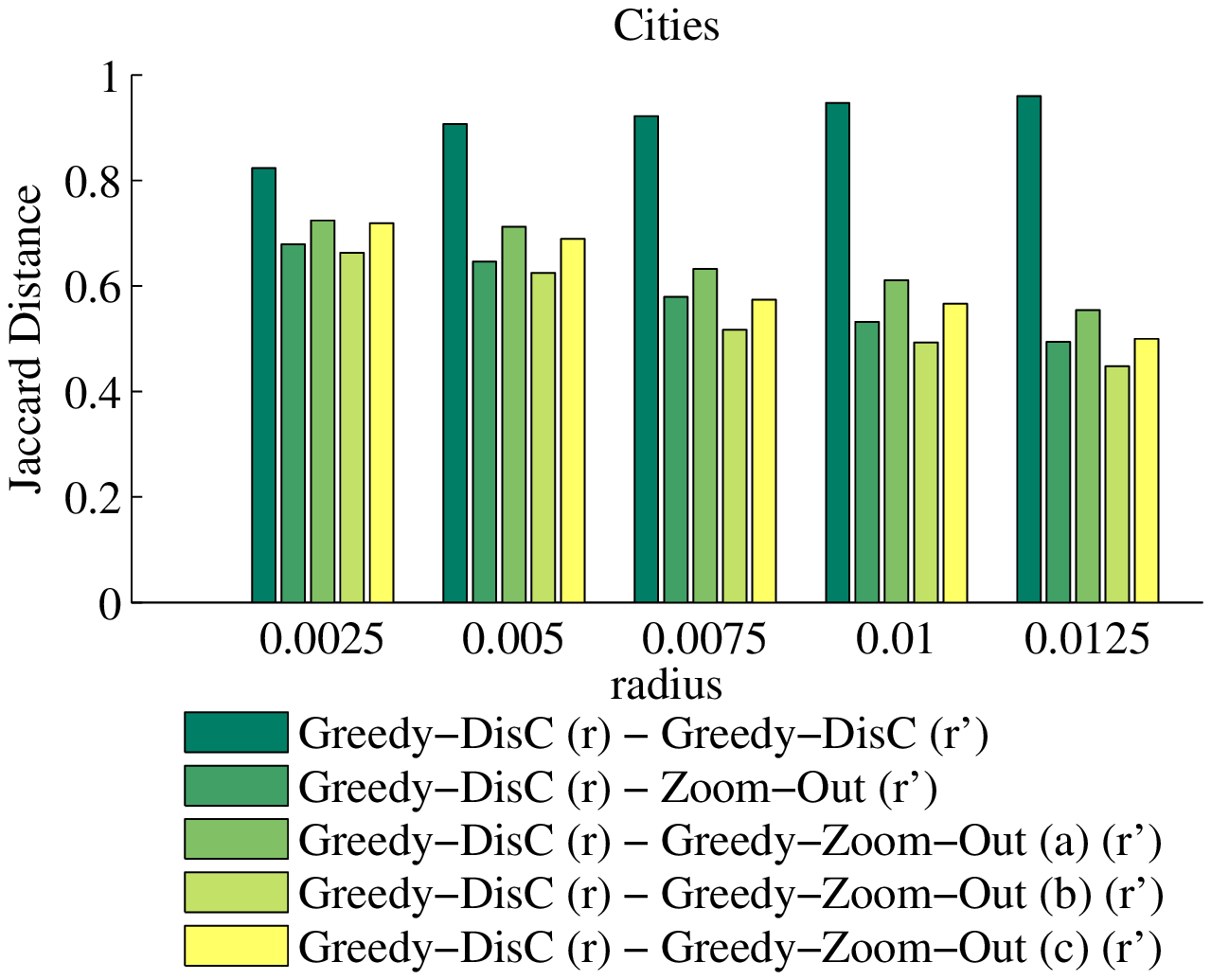}
	}
	\vspace{-0.3cm}
	\caption{Jaccard distance for zooming-out.}
	\label{fig:zoomout-jaccard}
	\vspace{-0.4cm}
\end{figure}

\section{Related Work}
\label{Related Work}

%
%
%
\noindent\textbf{Other Diversity Definitions:} 
Diversity has recently attracted a lot of attention as a means of enhancing user satisfaction \cite{DBLP:conf/icde/VieiraRBHSTT11, koudas, DBLP:conf/www/GollapudiS09, pods12}. Diverse results have been defined in various ways \cite{DBLP:journals/sigmod/DrosouP10}, namely in terms of \textit{content} (or \textit{similarity}), \textit{novelty} and \textit{semantic coverage}. Similarity definitions (e.g., \cite{DBLP:conf/recsys/ZhangH08}) interpret diversity as an instance of the \textit{$p$-dispersion problem} \cite{DBLP:journals/cor/ErkutUY94} whose objective is to choose $p$ out of $n$ given points, so that the minimum distance between any pair of chosen points is maximized. Our approach differs in that the size of the diverse subset is not an input parameter. Most current novelty and semantic coverage approaches to diversification (e.g., \cite{DBLP:conf/sigir/ClarkeKCVABM08, DBLP:conf/wsdm/AgrawalGHI09, DBLP:conf/edbt/YuLA09, sigmod12}) rely on associating a diversity score with each object in the result and then either selecting the top-$k$ highest ranked objects or those objects whose score is above some threshold. Such diversity scores are hard to interpret, since they do not depend solely on the object. Instead, the score of each object is relative to which objects precede it in the rank. Our approach is fundamentally different in that we treat the result as a whole and select DisC diverse subsets of it that fully cover it.

Another related work is that of \cite{DBLP:conf/pakdd/JainSH04} that extends nearest neighbor search to select $k$ neighbors that are not only spatially close to the query object but also differ on a set of predefined attributes above a specific threshold. Our work is different since our goal is not to locate the nearest and most diverse neighbors of a single object but rather to locate an independent and covering subset of the whole dataset.

The problem of diversifying continuous data has been recently considered in \cite{edbt12, wsdm12, DBLP:conf/sigir/MinackSN11} using a number of variations of the \textsc{MaxMin} and \textsc{MaxSum} diversification models.

Finally, another related method for selecting representative results, besides diversity-based ones, is $k$-medoids, since medoids can be viewed as representative objects (e.g., \cite{DBLP:journals/pvldb/LiuJ09}). However, medoids may not cover all the available space. Medoids were extended in \cite{DBLP:conf/cikm/BoimMN11} to include some sense of relevance (priority medoids).


\vspace{0.2cm}
\noindent\textbf{Results from Graph Theory:} 
The properties of independent and dominating (or covering) subsets have been extensively studied. A number of different variations exist. Among these, the \textsc{Minimum Independent Dominating Set Problem} (which is equivalent to the $r$-DisC diverse problem) has been shown to have some of the strongest negative approximation results: in the general case, it cannot be approximated in polynomial time within a factor of $n^{1-\epsilon}$ for any $\epsilon > 0$ unless $P = NP$ \cite{DBLP:journals/ipl/Halldorsson93a}. However, some approximation results have been found for special graph cases, such as bounded degree graphs \cite{DBLP:journals/iandc/ChlebikC08}. In our work, rather than providing polynomial approximation bounds for DisC diversity, we focus on the efficient computation of non-minimum but small DisC diverse subsets. There is a substantial amount of related work in the field of wireless networks research, since a \textsc{Minimum Connected Dominating Set} of wireless nodes can be used as a backbone for the entire network \cite{DBLP:journals/tcs/ThaiZTX07}. Allowing the dominating set to be connected has an impact on the complexity of the problem and allows different algorithms to be designed.

\section{Summary and Future Work}
\label{Summary}

In this paper, we proposed a novel, intuitive definition of diversity as the problem of selecting a minimum representative subset $S$ of a result $\mathcal{P}$, such that each object in $\mathcal{P}$ is represented by a similar object in $S$ and that the objects  included in $S$ are not similar to each other.
Similarity is modeled by a radius $r$ around each object.
We call such subsets $r$-DisC diverse subsets of $\mathcal{P}$.
We introduced adaptive diversification through decreasing $r$, termed zooming-in, and
increasing $r$, called zooming-out.
Since locating minimum $r$-DisC diverse subsets is an NP-hard problem, we
introduced heuristics for computing approximate solutions, including incremental ones for zooming,
and provided corresponding theoretical bounds.
We also presented an efficient implementation based on spatial indexing.

There are many directions for future work. We are currently looking into two different ways of
integrating relevance with DisC diversity. The first approach is by a ``weighted'' variation of the
DisC set, where each object has an associated weight based on its relevance.
Now the goal is to select a DisC subset having the maximum sum of weights.
The other approach is by allowing multiple radii per object, so that relevant objects get a smaller radius
than the radius of less relevant ones. Other potential future directions is implementations using different data structures
and designing algorithms for the online version of the problem.

\normalsize{
\vspace{-0.1cm}
\section*{Appendix}

\noindent\textbf{{Proof of Lemma~\ref{lem:manhattan2}}:}
Let $p_1$, $p_2$ be two independent neighbors of $p$. Then, it must hold that $\angle p_1pp_2$ (in the Euclidean space) is larger than $\frac{\pi}{4}$. We will prove this using contradiction. $p_1$, $p_2$ are neighbors of $p$ so they must reside in the shaded area of Figure~\ref{fig:proofs}\subref{fig:proofs-manhattan}. Without loss of generality, assume that one of them, say $p_1$, is aligned to the vertical axis. Assume that $\angle p_1pp_2$ $\leq$ $\frac{\pi}{4}$. Then $\cos(\angle p_1p_ip_2)$ $\geq$ $\frac{\sqrt{2}}{2}$. It holds that $b \leq r$ and $c \leq r$, thus, using the cosine law we get that $a^2 \leq r^2(2-\sqrt{2})$ (1). The Manhattan distance of $p_1$,$p_2$ is equal to $x + y$ $=$ $\sqrt{a^2 + 2xy}$ (2). Also, the following hold: $x = \sqrt{b^2 - z^2}$, $y = c - z$ and $z$ $=$ $b\cos(\angle p_1p_ip_2)$ $\geq$ $\frac{b\sqrt{2}}{2}$. Substituting $z$ and $c$ in the first two equations, we get $x \leq \frac{b}{\sqrt{2}}$ and $y$ $\leq$ $r - \frac{b\sqrt{2}}{2}$. From (1),(2) we now get that $x+y \leq r$, which contradicts the independence of $p_1$ and $p_2$. Therefore, $p$ can have at most $(2\pi / \frac{\pi}{4}) - 1$ $= 7$ independent neighbors.

\begin{figure}
	\centering
	\hspace{-0.4cm}
	\subfloat[]{		
		\includegraphics[clip=true, trim=0 0 0 0.5cm, height=2.7cm]{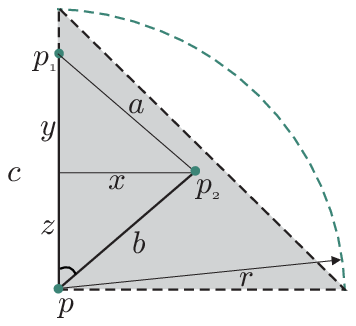}
		\label{fig:proofs-manhattan}
	}\hspace{-0.8cm}
	\subfloat[]{		
		\includegraphics[height=2.7cm]{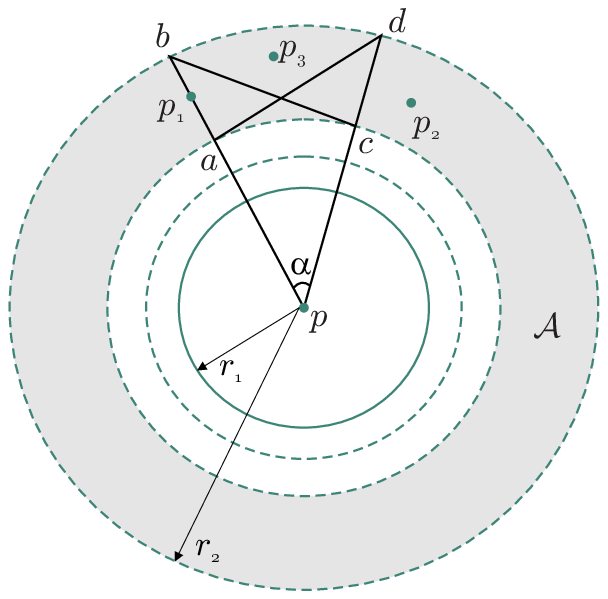}
		\label{fig:proofs-zoom-euclidean}
	}\hspace{-0.1cm}
	\subfloat[]{		
		\includegraphics[height=2.7cm]{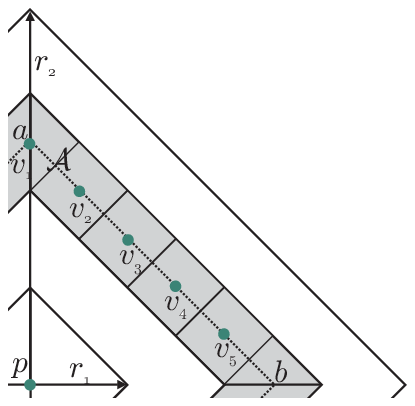}
		\label{fig:proofs-zoom-manhattan}
	}
	\caption{Independent neighbors.}
	\label{fig:proofs}
	\vspace{-0.2cm}
\end{figure}

\vspace{0.3cm}

\noindent\textbf{{Proof of Theorem~\ref{thm:amortized}}:}
We consider that inserting a node $u$ into $S$ has cost 1. We distribute this cost equally among all covered nodes, i.e., after being labeled grey, nodes are not charged anymore. Assume an optimal minimum dominating set $S^*$. The graph $G$ can be decomposed into a number of star-shaped subgraphs, each of which has one node from $S^*$ at its center. The cost of an optimal minimum dominating set is exactly 1 for each star-shaped subgraph. We show that for a non-optimal set $S$, the cost for each star-shaped subgraph is at most $\ln\Delta$, where $\Delta$ is the maximum degree of the graph. Consider a star-shaped subgraph of $S^*$ with $u$ at its center and let $N_r^W(u)$ be the number of white nodes in it. If a node in the star is labeled grey by \texttt{Greedy-DS}, these nodes are charged some cost. By the greedy condition of the algorithm, this cost can be at most $1 / |N_r^W(u)|$ per newly covered node. Otherwise, the algorithm would rather have chosen $u$ for the dominating set because $u$ would cover at least $|N_r^W(u)|$ nodes. In the worst case, no two nodes in the star of $u$ are covered at the same iteration. In this case, the first node that is labeled grey is charged at most $1 / (\delta(u) + 1)$, the second node is charged at most $1 / \delta(u)$ and so on, where $\delta(u)$ is the degree of $u$. Therefore, the total cost of a star is at most:
\vspace{-0.2cm}
\[
\frac{1}{\delta(u) + 1} + \frac{1}{\delta(u)} + \ldots + \frac{1}{2} + 1 = H(\delta(u) + 1) \leq H(\Delta + 1) \approx \ln\Delta
\vspace{-0.2cm}
\]
where $H(i)$ is the $i^{\text{th}}$ harmonic number. Since a minimum dominating set is equal or smaller than a minimum independent dominating set, the theorem holds.

\vspace{0.3cm}

\balance

\noindent\textbf{{Proof of Lemma~\ref{lem:independent-neighbors}}(i):}
For the proof, we use a technique for partitioning the annulus between $r_1$ and $r_2$ similar to the one in \cite{DBLP:journals/tmc/ThaiWLZD07} and \cite{DBLP:conf/icdcs/XingCPR08}. Let $r_1$ be the radius of an object $p$ (Figure~\ref{fig:proofs}\subref{fig:proofs-zoom-euclidean}) and $\alpha$ a real number with $0 < \alpha < \frac{\pi}{3}$. We draw circles around the object $p$ with radii $(2cos\alpha)^{x_p},$ $(2cos\alpha)^{x_p + 1},$ $(2cos\alpha)^{x_p + 2},$ $\ldots,$ $(2cos\alpha)^{y_p - 1},$ $(2cos\alpha)^{y_p}$, such that $(2cos\alpha)^{x_p} \leq r_1$ and $(2cos\alpha)^{x_p + 1} > r_2$ and $(2cos\alpha)^{y_p - 1} < r_2$ and $(2cos\alpha)^{y_p} \geq r_2$. It holds that $x_p = \left\lfloor \frac{\ln r_1}{\ln(2\cos\alpha)} \right\rfloor$ and $y_p = \left\lceil \frac{\ln r_2}{\ln(2\cos\alpha)} \right\rceil$. In this way, the area around $p$ is partitioned into $y_p - x_p$ annuluses plus the $r_1$-disk around $p$.

Consider an annulus $\mathcal{A}$. Let $p_1$ and $p_2$ be two neighbors of $p$ in $\mathcal{A}$ with \ $dist(p_1, p_2) > r_1$. Then, it must hold that $\angle p_1pp_2 > \alpha$. To see this, we draw two segments from $p$ crossing the inner and outer circles of $\mathcal{A}$ at $a$, $b$ and $c$, $d$ such that $p_1$ resides in $pb$ and $\angle bpd = \alpha$, as shown in the figure. Due to the construction of the circles, it holds that $\frac{\left|pb\right|}{\left|pc\right|} = \frac{\left|pd\right|}{\left|pa\right|} = 2\cos\alpha$. From the cosine law for $p\overset{\triangle}{a}d$, we get that $\left|ad\right| = \left|pa\right|$ and, therefore, it holds that $\left|cb\right| = \left|ad\right| = \left|pa\right| = \left|pc\right|$. Therefore, for any object $p_3$ in the area $abcd$ of $\mathcal{A}$, it holds that $\left|pp_3\right| > \left|bp_3\right|$ which means that all objects in that area are neighbors of $p_1$, i.e., at distance less or equal to $r_1$. For this reason, $p_2$ must reside outside this area which means that $\angle p_1pp_2 > \alpha$. Based on this, we see that there exist at most $\frac{2\pi}{\alpha} - 1$ independent (for $r_1$) nodes in $\mathcal{A}$.

The same holds for all annuluses. Therefore, we have at most $(y_p - x_p)\left(\frac{2\pi}{\alpha} - 1\right)$ independent nodes in the annuluses. For $0 < \alpha < \frac{\pi}{3}$, this has a minimum when $\alpha$ is close to $\frac{\pi}{5}$ and that minimum value is $9 \left\lceil \frac{\ln(r_2/r_1)}{\ln(2\cos(\pi/5))} \right\rceil$ $=$ $9 \left\lceil \log_\beta(r_2/r_1) \right\rceil$, where $\beta = \frac{1 + \sqrt5}{2}$.


\vspace{0.3cm}

\noindent\textbf{{Proof of Lemma~\ref{lem:independent-neighbors}}(ii):}
Let $r_1$ be the radius of an object $p$. We draw Manhattan circles around the object $p$ with radii $r_1, 2r_1, \ldots$ until the radius $r_2$ is reached. In this way, the area around $p$ is partitioned into $\gamma$ $=$ $\left\lceil \frac{r_2-r_1}{r_1}\right\rceil$ Manhattan annuluses plus the $r_1$-Manhattan-disk around $p$.

Consider an annulus $\mathcal{A}$. The objects shown in Figure~\ref{fig:proofs}\subref{fig:proofs-zoom-manhattan} cover the whole annulus and their Manhattan pairwise distances are all greater or equal to $r_1$. Assume that the annulus spans among distance $ir_1$ and $(i+1)r_1$ from $p$, where $i$ is an integer with $i > 1$. Then, $|ab|$ $=$ $\sqrt{2\left(ir_1+r_1/2\right)^2}$.  Also, for two objects $p_1$, $p_2$ it holds that $|p_1p_2|$ $=$ $\sqrt{2\left(r_1/2\right)^2}$. Therefore, at one quadrant of the annulus there are $\frac{|ab|}{|p_1p_2|}$ $=$ $2i+1$ independent neighbors which means that there are $4(2i+1)$ independent neighbors in $\mathcal{A}$. Therefore, there are in total $\sum_{i=1}^{\gamma} 4(2i+1)$ independent (for $r_1$) neighbors of $p$.

}


\begin{thebibliography}{10}

\bibitem{acme}
{\em \textnormal{Acme digital cameras database}}.
\newblock http://acme.com/digicams.

\bibitem{data:cities}
{\em \textnormal{Greek cities dataset}}.
\newblock http://www.rtreeportal.org.

\bibitem{DBLP:conf/wsdm/AgrawalGHI09}
R.~Agrawal, S.~Gollapudi, A.~Halverson, and S.~Ieong.
\newblock Diversifying search results.
\newblock In {\em WSDM}, 2009.

\bibitem{koudas}
A.~Angel and N.~Koudas.
\newblock Efficient diversity-aware search.
\newblock In {\em SIGMOD}, 2011.

\bibitem{DBLP:conf/cikm/BoimMN11}
R.~Boim, T.~Milo, and S.~Novgorodov.
\newblock Diversification and refinement in collaborative filtering
  recommender.
\newblock In {\em CIKM}, 2011.

\bibitem{pods12}
A.~Borodin, H.~C. Lee, and Y.~Ye.
\newblock Max-sum diversifcation, monotone submodular functions and dynamic
  updates.
\newblock In {\em PODS}, 2012.

\bibitem{DBLP:journals/iandc/ChlebikC08}
M.~Chleb\'{\i}k and J.~Chleb\'{\i}kov{\'a}.
\newblock Approximation hardness of dominating set problems in bounded degree
  graphs.
\newblock {\em Inf. Comput.}, 206(11), 2008.

\bibitem{DBLP:journals/dm/ClarkCJ90}
B.~N. Clark, C.~J. Colbourn, and D.~S. Johnson.
\newblock Unit disk graphs.
\newblock {\em Discrete Mathematics}, 86(1-3), 1990.

\bibitem{DBLP:conf/sigir/ClarkeKCVABM08}
C.~L.~A. Clarke, M.~Kolla, G.~V. Cormack, O.~Vechtomova, A.~Ashkan,
  S.~B{\"u}ttcher, and I.~MacKinnon.
\newblock Novelty and diversity in information retrieval evaluation.
\newblock In {\em SIGIR}, 2008.

\bibitem{DBLP:journals/sigmod/DrosouP10}
M.~Drosou and E.~Pitoura.
\newblock Search result diversification.
\newblock {\em SIGMOD Record}, 39(1), 2010.

\bibitem{tr}
M.~Drosou and E.~Pitoura.
\newblock Dis\uppercase{C} diversity: Result~diversification~based~on
  dissimilarity and coverage, \uppercase{T}echnical \uppercase{R}eport.
\newblock {\em University of Ioannina}, 2011.

\bibitem{edbt12}
M.~Drosou and E.~Pitoura.
\newblock Dynamic diversification of continuous data.
\newblock In {\em EDBT}, 2012.

\bibitem{DBLP:journals/cor/ErkutUY94}
E.~Erkut, Y.~{\"U}lk{\"u}sal, and O.~Yeni\c{c}erioglu.
\newblock A comparison of {\it p}-dispersion heuristics.
\newblock {\em Computers {\&} OR}, 21(10), 1994.

\bibitem{sigmod12}
P.~Fraternali, D.~Martinenghi, and M.~Tagliasacchi.
\newblock Top-k bounded diversification.
\newblock In {\em SIGMOD}, 2012.

\bibitem{DBLP:books/fm/GareyJ79}
M.~R. Garey and D.~S. Johnson.
\newblock {\em Computers and Intractability: A Guide to the Theory of
  NP-Completeness}.
\newblock W. H. Freeman, 1979.

\bibitem{DBLP:conf/www/GollapudiS09}
S.~Gollapudi and A.~Sharma.
\newblock An axiomatic approach for result diversification.
\newblock In {\em WWW}, 2009.

\bibitem{DBLP:journals/ipl/Halldorsson93a}
M.~M. Halld{\'o}rsson.
\newblock Approximating the minimum maximal independence number.
\newblock {\em Inf. Process. Lett.}, 46(4), 1993.

\bibitem{DBLP:conf/pakdd/JainSH04}
A.~Jain, P.~Sarda, and J.~R. Haritsa.
\newblock Providing diversity in k-nearest neighbor query results.
\newblock In {\em PAKDD}, 2004.

\bibitem{DBLP:journals/pvldb/LiuJ09}
B.~Liu and H.~V. Jagadish.
\newblock Using trees to depict a forest.
\newblock {\em PVLDB}, 2(1), 2009.

\bibitem{DBLP:conf/sigir/MinackSN11}
E.~Minack, W.~Siberski, and W.~Nejdl.
\newblock Incremental diversification for very large sets: a streaming-based
  approach.
\newblock In {\em SIGIR}, 2011.

\bibitem{wsdm12}
D.~Panigrahi, A.~D. Sarma, G.~Aggarwal, and A.~Tomkins.
\newblock Online selection of diverse results.
\newblock In {\em WSDM}, 2012.

\bibitem{DBLP:journals/tmc/ThaiWLZD07}
M.~T. Thai, F.~Wang, D.~Liu, S.~Zhu, and D.-Z. Du.
\newblock Connected dominating sets in wireless networks with different
  transmission ranges.
\newblock {\em IEEE Trans. Mob. Comput.}, 6(7), 2007.

\bibitem{DBLP:journals/tcs/ThaiZTX07}
M.~T. Thai, N.~Zhang, R.~Tiwari, and X.~Xu.
\newblock On approximation algorithms of k-connected m-dominating sets in disk
  graphs.
\newblock {\em Theor. Comput. Sci.}, 385(1-3), 2007.

\bibitem{DBLP:journals/tkde/TrainaTFS02}
C.~J. Traina, A.~J.~M. Traina, C.~Faloutsos, and B.~Seeger.
\newblock Fast indexing and visualization of metric data sets using slim-trees.
\newblock {\em IEEE Trans. Knowl. Data Eng.}, 14(2), 2002.

\bibitem{DBLP:conf/icde/VeeSSBA08}
E.~Vee, U.~Srivastava, J.~Shanmugasundaram, P.~Bhat, and S.~Amer-Yahia.
\newblock Efficient computation of diverse query results.
\newblock In {\em ICDE}, 2008.

\bibitem{DBLP:conf/icde/VieiraRBHSTT11}
M.~R. Vieira, H.~L. Razente, M.~C.~N. Barioni, M.~Hadjieleftheriou,
  D.~Srivastava, C.~Traina, and V.~J. Tsotras.
\newblock On query result diversification.
\newblock In {\em ICDE}, 2011.

\bibitem{DBLP:conf/icdcs/XingCPR08}
K.~Xing, W.~Cheng, E.~K. Park, and S.~Rotenstreich.
\newblock Distributed connected dominating set construction in geometric k-disk
  graphs.
\newblock In {\em ICDCS}, 2008.

\bibitem{DBLP:conf/edbt/YuLA09}
C.~Yu, L.~V.~S. Lakshmanan, and S.~Amer-Yahia.
\newblock It takes variety to make a world: diversification in recommender
  systems.
\newblock In {\em EDBT}, 2009.

\bibitem{similaritysearch2006}
P.~Zezula, G.~Amato, V.~Dohnal, and M.~Batko.
\newblock {\em Similarity Search - The Metric Space Approach}.
\newblock Springer, 2006.

\bibitem{DBLP:conf/recsys/ZhangH08}
M.~Zhang and N.~Hurley.
\newblock Avoiding monotony: improving the diversity of recommendation lists.
\newblock In {\em RecSys}, 2008.

\bibitem{DBLP:conf/www/ZieglerMKL05}
C.-N. Ziegler, S.~M. McNee, J.~A. Konstan, and G.~Lausen.
\newblock Improving recommendation lists through topic diversification.
\newblock In {\em WWW}, 2005.

\end{thebibliography}
\end{document}